\documentclass[
submission
, nomarks
]{dmtcs-episciences}

\usepackage[utf8]{inputenc}
\usepackage{amsmath, amssymb}
\usepackage{xspace}

\usepackage[round]{natbib}

\usepackage{xcolor}
\usepackage{todonotes}

\usepackage{tikz}
\usetikzlibrary{positioning, shapes.geometric}
\usetikzlibrary{arrows}
\usetikzlibrary{fit}
\usepackage{mathrsfs}
\usetikzlibrary{arrows}
\usetikzlibrary{backgrounds}
\usetikzlibrary{hobby}
\usetikzlibrary{decorations.markings}

\usetikzlibrary{positioning}

\usetikzlibrary{math}

\usepackage[ruled, vlined, linesnumbered]{algorithm2e}

\usepackage{tikzcd}
\usepackage{pgfplots}
\usepackage{adjustbox}
\usepackage{paralist}

\newtheorem{theorem}{Theorem}[section]
\newtheorem{lemma}[theorem]{Lemma}
\newtheorem{proposition}[theorem]{Proposition}
\newtheorem{definition}[theorem]{Definition}
\newtheorem{observation}[theorem]{Observation}

\newcommand{\sep}{\;|\;}

\newcommand{\seq}{\subseteq}
\newcommand{\bb}{\mathbb}
\newcommand{\ca}{\mathcal}

\newcommand{\symdif}{\mathbin{\triangle}}

\DeclareMathOperator{\range}{range}
\DeclareMathOperator{\dom}{dom}
\DeclareMathOperator{\dist}{dist}

\newcommand{\kb}{K^+}
\newcommand{\kout}{K}
\newcommand{\kstar}{K^*}

\SetCommentSty{textit}

\SetCommentSty{mycommfont}

\newcommand{\FPT}{\textsf{FPT}\xspace}
\newcommand{\NP}{\textsf{NP}\xspace}
\newcommand{\XP}{\textsf{XP}\xspace}
\newcommand{\XNLP}{\textsf{XNLP}\xspace}

\newcommand{\bcoloring}{$b$-\textsc{coloring}\xspace}

\author[Jakub Balab\'an]{Jakub Balab\'an\thanks{Supported by the project 26-21334S of the Czech Science Foundation. Brno Ph.D. Talent Scholarship Holder -- Funded by the Brno City Municipality}}
\title[Finding $b$-colorings Using Feedback Edges]{Finding $b$-colorings Using Feedback Edges}
\affiliation{
  Faculty of Informatics, Masaryk University, Brno, Czechia}
\keywords{$b$-coloring, fixed-parameter algorithms, feedback edge number, distance to co-cluster} %%TODO mandatory; please add comma-separated list of keywords

\begin{document}
% This is the form of collecting the publication data, starting with vol 25.
% It only applies when the class is loaded with the "final" option -- the
% "submission" option auto-generates its own placeholder publication info,
% so \publicationdata is skipped (and would in fact error out) in that case.
% %%TODO Please replace with the information received from the editorial
% board upon acceptance (volume, year, article number, DOI, dates, and
% acceptance date).
\makeatletter
\@ifundefined{processDatesSC}{}{%
\publicationdata
{vol. XX:X} %%TODO volume, e.g. "vol. 27:1"
{2026} %%TODO year of publication
{1} %%TODO article number within the volume
{10.46298/dmtcs.XXXXX} %%TODO DOI assigned by DMTCS
{2026-1-1; None} %%TODO submission date; revision date(s), separated by "; ", or "None"
{2026-1-1} %%TODO date of acceptance
}
\makeatother
\maketitle

\begin{abstract}
A $b$-coloring of a graph is a proper vertex coloring such that each color class contains a vertex that sees all other colors in its neighborhood.
The \bcoloring problem, in which the task is to decide whether a graph admits a $b$-coloring with $k$ colors, is \textsf{NP}-complete in general but polytime solvable on trees.
Moreover, it is known that \bcoloring is in \textsf{XP} but \textsf{W[$t$]}-hard for all $t \in \mathbb{N}$ when parameterized by tree-width.
In fact, only very few parameters, such as the vertex cover number, were known to admit an \FPT algorithm for \bcoloring.

In this paper, we consider a more restrictive parameter measuring similarity to trees than tree-width, namely the feedback edge number, and show that \bcoloring is fixed-parameter tractable under this parameterization.
Our algorithm combines standard techniques used in parameterized algorithmics with the problem-specific ideas used in the polytime algorithm for trees.
In addition, we present an \FPT algorithm for \bcoloring parameterized by distance to co-cluster, which is a parameter measuring similarity to complete multipartite graphs.
Finally, we make several observations based on known results, including that \bcoloring is \textsf{W[$1$]}-hard when parameterized by tree-depth.
\end{abstract}

%\newpage

\section{Introduction}\label{sec:intro}

A \emph{proper coloring} of a graph is an assignment of colors to its vertices such that no two adjacent vertices receive the same color.
The \textsc{Chromatic Number} problem, in which the task is to decide whether a graph admits a proper coloring with $k$ colors, is one of the famous 21 \NP-complete problems listed by \cite{Karp72}.
This problem is \NP-complete even when $k$ is any fixed constant higher than 2.
Consequently, various heuristics must be employed in practice.
One simple heuristic works as follows.
We start with an arbitrary proper coloring and, in each iteration, we try to decrease the number of colors by one.
We search for a color $c$ such that for each vertex $u$ colored with $c$, there is a color $c_u \ne c$ that is not present in the neighborhood of $u$.
If we can find such a color $c$, then we assign the color $c_u$ to each vertex $u$ currently colored with $c$, and the obtained coloring remains proper.
Otherwise, the heuristic terminates.

\cite{IrvingM99} initiated the theoretical study of this heuristic in 1999 and defined a proper coloring to be a \emph{$b$-coloring} if each color class contains a vertex, a so-called \emph{$b$-vertex}, that sees all other colors in its neighborhood.
It is easy to see that $b$-colorings are exactly those colorings that cannot be improved by the heuristic.
In addition, they defined the \emph{$b$-chromatic number} of a graph $G$, denoted $\chi_b(G)$, as the largest integer $k$ such that $G$ admits a $b$-coloring with $k$ colors; observe that the $b$-chromatic number describes the worst-case behavior of the heuristic.
These concepts received a lot of attention since their inception, see \cite{KratochvilTV02, corteel2005approximating, kouider2006bounds, jakovac2010b, HavetSS12, CamposLMSSS15, melo2021matheuristic, JaffkeLL24}. A survey of the area was written by~\cite{JakovacP18}.

There are two natural computational problems related to $b$-colorings: the \bcoloring problem, which asks for a $b$-coloring with exactly $k$ colors, and the \textsc{$b$-Chromatic Number} problem, which asks whether $\chi_b(G) \ge k$.
Notice that a polytime algorithm for \bcoloring can be used to solve \textsc{$b$-Chromatic Number} in polynomial time.
However, the converse does not hold in general: if $\chi(G) < k < \chi_b(G)$, where $\chi(G)$ is the chromatic number of $G$, then a $b$-coloring with $k$ colors does not have to exist\footnote{For example, the bipartite graph with vertex set $\{u_i, v_i \sep i \in [4]\}$ and edge set $\{u_iv_j \sep i \ne j\}$ admits a $b$-coloring with two or four colors but not with three colors, see \cite{IrvingM99}.}. 
Hence, \textsc{$b$-Chromatic Number} is, in some sense, an ``easier'' problem than \bcoloring.

In their seminal paper, \cite{IrvingM99} proved that \textsc{$b$-Chromatic Number} is \NP-complete on general graphs and solvable in polynomial time on trees.
The \NP-completeness was subsequently extended to bipartite graphs by \cite{KratochvilTV02}, to chordal graphs by \cite{HavetSS12}, to co-bipartite graphs by \cite{BonomoSSV15}, and to line graphs by \cite{CamposLMSSS15}.
On the positive side, polytime algorithms for \textsc{$b$-Chromatic Number} have been designed for several classes with few $P_4$s by \cite{BonomoDMMV09, VelasquezBK11, CamposSSM14} and for various tree-like classes, such as tree-cographs by \cite{BonomoSSV15}, cacti by \cite{CamposSMS09}, and graphs of large girth by \cite{CamposLS15}, see also \cite{FerreiraDaSilva10}.
However, almost all of these algorithms were unified by an \XP algorithm for \bcoloring parameterized by clique-width, which was recently developed by \cite{JaffkeLL24}.
Indeed, \XP means that the algorithm runs in polynomial time on any class of bounded clique-width, and all of the above-mentioned classes (except for the graphs of large girth) have bounded clique-width. %\cite{JaffkeLL24}.

Once an \XP algorithm is known, the natural question is whether an \FPT algorithm is possible, see \cite{CyganFKLMPPS15}.
Unfortunately, \cite{JaffkeLS23} proved that \bcoloring is \XNLP-complete when parameterized by path-width, which implies that it is \textsf{W[$t$]}-hard for each $t \in \mathbb{N}$ under this parameterization.
Hence, it is very unlikely that an \FPT algorithm parameterized by path-width (or tree-width) exists.
Let us emphasize that this hardness result relies on the fact that the number of colors is \emph{not} part of the parameterization, i.e., it may be as large as the number of vertices.
In fact, if we parameterize by tree-width plus the number of colors, then an \FPT algorithm is possible, see \cite{JaffkeLL24, Davi_Andrade} and Proposition~\ref{prop:tree-width}.

A trivial upper-bound on the $b$-chromatic number, called the \emph{$m$-degree} by \cite{IrvingM99}, is defined as the largest integer $k$ such that there are at least $k$ vertices of degree at least $k-1$.
Since the vertex cover number upper bounds both the $m$-degree and tree-width, there is an \FPT algorithm for \bcoloring parameterized by the vertex cover number, see \cite{JaffkeLL24}.
To the best of our knowledge, the only other structural parameters that were known to allow for an \FPT algorithm for \bcoloring are the twin cover number and neighborhood diversity, see \cite{JaffkeLS23} and Figure~\ref{fig:hierarchy}.

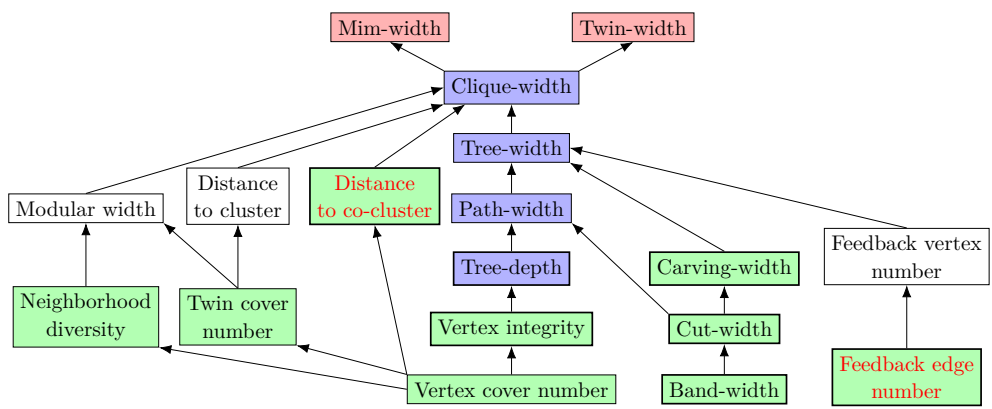
\begin{figure}[t]
\scalebox{0.8}{
\begin{tikzpicture}[every node/.style={draw, rectangle}]
  \node[fill=red!30] (twin) at (0,-2) {Twin-width};
  \node[fill=red!30] (mim) at (-4,-2) {Mim-width};
  \node[fill=blue!30] (clique) at (-2,-3) {Clique-width};
  \node[fill=blue!30]  (tw) at (-2,-4) {Tree-width};
  \node[fill=blue!30] (pw) at (-2, -5) {Path-width};
  \node[fill=blue!30, very thick] (td) at (-2,-6) {Tree-depth};
  \node[fill=green!30, very thick] (vi) at (-2,-7) {Vertex integrity};
  \node [align=center, fill=green!30] (vc) at (-2,-8) {Vertex cover number};
  
  \node [align=center, fill=green!30] (nd) at (-9.3,-6.8) {Neighborhood \\diversity};
  \node [align=center, fill=green!30] (tc) at (-6.8,-6.8) {Twin cover\\ number};
   
  \node [align=center, fill=green!30, very thick, font=\bfseries] (fen) at (4.5,-7.8) {Feedback edge \\number};
  \node [align=center] (fvn) at (4.5,-5.8) {Feedback vertex \\number};
  
  \node [align=center, fill=green!30, very thick] (bw) at (1.5,-8) {Band-width};
  \node [align=center, fill=green!30, very thick] (cutw) at (1.5,-7) {Cut-width};
  \node [align=center, fill=green!30, very thick] (carving) at (1.5,-6) {Carving-width};
  
  \draw[{Latex[length=2mm]}-] (tw.south east) -- (carving.north);
  \draw[{Latex[length=2mm]}-] (carving.south) -- (cutw.north);
  \draw[{Latex[length=2mm]}-] (cutw.south) -- (bw.north);
  \draw[{Latex[length=2mm]}-] (pw.south east) -- (cutw.north west);
  
  \node (modw) at (-9.3, -5) {Modular width};
  
  \node[align=center] (dtc) at (-6.8, -4.8) {Distance \\to cluster};
  \node[align=center, fill=green!30, very thick, font=\bfseries] (dtcc) at (-4.5, -4.8) {Distance \\to co-cluster};

  \draw[{Latex[length=2mm]}-] (dtc.south) -- (tc.north);
  \draw[{Latex[length=2mm]}-] (clique.south west) -- (dtc.north);
  
  \draw[{Latex[length=2mm]}-] ([xshift=10pt]clique.south west) -- (dtcc.north);
  \draw[-{Latex[length=2mm]}] (vc.north west) -- (dtcc.south);

  \draw[{Latex[length=2mm]}-] (twin.south) -- (clique.north east);
  \draw[{Latex[length=2mm]}-] (mim.south) -- (clique.north west);
  
  \draw[{Latex[length=2mm]}-] (clique.south) -- (tw.north);
  \draw[{Latex[length=2mm]}-] (tw.south) -- (pw.north);
  \draw[{Latex[length=2mm]}-] (pw.south) -- (td.north);
  \draw[{Latex[length=2mm]}-] (td.south) -- (vi.north);
  \draw[{Latex[length=2mm]}-] (vi.south) -- (vc.north);
  
  \draw[{Latex[length=2mm]}-] (tw.east) -- (fvn.north);
  \draw[{Latex[length=2mm]}-] (fvn.south) -- (fen.north);
  
  \draw[{Latex[length=2mm]}-] (modw.south) -- (nd.north);
  \draw[{Latex[length=2mm]}-] (modw.south east) -- (tc.north);
 
  \draw[{Latex[length=2mm]}-] (nd.south east) -- (vc.west);
  \draw[{Latex[length=2mm]}-] (tc.south east) -- (vc.north west);
  
  \draw[{Latex[length=2mm]}-] (clique.west) -- (modw.north);
\end{tikzpicture}}

\caption{Parameterized complexity of \bcoloring under various structural parameters. A directed path from a parameter $\alpha$ to a parameter $\beta$ indicates that $\beta \le f(\alpha)$ for some computable function $f$. Green stands for \FPT, blue is \textsf{W[1]}-hard and \textsf{XP}, red is \textsf{paraNP}-complete, and white means that it is unknown whether the problem is \FPT or \textsf{W[1]}-hard under given parameterization. Our new results are marked with thick boundary, and the two parameters for which our \FPT algorithms are non-trivial are written in bold text.}
\label{fig:hierarchy}
\end{figure}

\subparagraph*{Our contribution.}
First, in Section~\ref{sec:prelims}, we observe that there are other well-studied parameters that functionally upper bound tree-width and the $m$-degree, which yields new \FPT algorithms for \bcoloring, see Propositions~\ref{prop:vertex-integrity} and~\ref{prop:carving-width}.
These parameters are vertex integrity (see \cite{hanaka_et_al:LIPIcs.MFCS.2024.58}), band-width (see \cite{diaz2002survey}), cut-width (see \cite{diaz2002survey}), and carving-width (see \cite{seymour1994call}); the relations between them are illustrated in Figure~\ref{fig:hierarchy}.
We complement these algorithms with a hardness result implied by known results (\cite{JaffkeLS23, jong2024analysing}): \bcoloring is \textsf{W[1]}-hard when parameterized by tree-depth, see Proposition~\ref{prop:tree-depth-hardness}.

In Section~\ref{sec:co-cluster}, we present an \FPT algorithm for \bcoloring parameterized by distance to co-cluster, which is a parameter that has been used in several \FPT algorithms, see \cite{komusiewicz2020matching, REDDY2024181, galby2023metric}.
Informally, our algorithm works as follows.
Let $G$ be the input graph.
First, we find a small set $S$ of vertices such that $G - S$ is a complete multipartite graph (such a set exists by definition of the parameter), and we guess a proper coloring $\chi$ of $S$.
Second, we partition the maximal independent sets of $G - S$ into ``types'' based on their connection to $S$.
Crucially, the number of types will be small.
Third, for each color $c$ used by $\chi$, we guess whether some vertex of $G- S$ should be colored with $c$, and if yes, we guess the type of the independent set that should contain such a vertex.
The key technical lemma states that two sets with the same type are interchangeable, i.e., if a color $c$ should have a vertex in a set of type $t$, then it does not matter which one of such sets we choose.
Finally, we must also handle the colors \emph{not} used by $\chi$: there may be many such colors but their treatment is relatively straightforward.

In Section~\ref{sec:fen}, we show an \FPT algorithm for \bcoloring parameterized by the feedback edge number, which is the smallest number of edges that need to be deleted to obtain an acyclic graph; this parameter has also been used in several \FPT algorithms, see \cite{DBLP:journals/siamdm/BalabanGR25, DBLP:journals/tcs/UhlmannW13, DBLP:journals/jgaa/BannisterCE18, DBLP:journals/algorithmica/GanianO21}.
This algorithm is the main contribution of this paper.
We combine standard techniques used in parameterized algorithms with the concept of a \emph{pivot}, which was introduced in the polytime algorithm for \textsc{$b$-Chromatic Number} on trees by \cite{IrvingM99}.
However, our definition of a pivot will be more technical because it will be ``parameterized'' by a partial coloring of the input graph.
In order to smoothly introduce pivots to the reader, we now briefly describe the original algorithm.

\subparagraph*{Algorithm for trees.}
Let $T$ be a tree and let $k$ be the $m$-degree of $T$; recall that $k$ is an upper bound on the $b$-chromatic number of $T$, which we denote by $\chi_b(T)$.
Informally, this upper bound can be achieved ``almost always'': the only exception are so-called \emph{pivoted} trees, which can be recognized in linear time, see \cite{IrvingM99}.
More precisely, $\chi_b(T) = k-1$ if $T$ is pivoted, and $\chi_b(T) = k$ otherwise.
Let us call vertices of degree at least $k-1$ \emph{candidates}; notice that candidates are those vertices that may serve as $b$-vertices in a $k$-coloring of $T$.
We say that $T$ is \emph{pivoted} if there are exactly $k$ candidates and there is a vertex $u$, a so-called \emph{pivot}, such that:
\begin{compactenum}
\item $u$ is not a candidate;
\item each candidate is adjacent to $u$ or to another candidate that is adjacent to $u$;
\item each candidate that is adjacent to $u$ and to another candidate has degree $k-1$.
\end{compactenum}
It is not hard to see that a pivoted tree does not admit a $b$-coloring with $k$ colors: each candidate would have to be a $b$-vertex and a vertex of degree $k-1$ can be a $b$-vertex only if it sees no color twice, which means that no color may be assigned to the pivot $u$.

We have already mentioned that several polytime algorithms for \textsc{$b$-Chromatic Number} on tree-like graphs have been developed, see \cite{BonomoSSV15, CamposSMS09, CamposLS15}.
These algorithms build on the concept of a pivoted tree.
For example, the concept of a \emph{pivoted cactus} has been introduced, see \cite{CamposSMS09}.

\subparagraph*{Overview of our algorithm}
Now we provide a high-level overview of the algorithm for \bcoloring parameterized by the feedback edge number, see also Figure~\ref{fig:overview}.
If some notation used in the following text is unclear to the reader, we refer to Section~\ref{sec:prelims} for definitions.

Let $G$ be a graph with feedback edge number $p$. Our goal is to decide whether $G$ admits a $b$-coloring with $k$ colors.
If $\chi$ is a partial coloring of $G$, then a \emph{$\chi$-candidate} is a vertex that may become a $b$-vertex in a $k$-coloring extending $\chi$ (for example, if $\chi$ colors no vertices, then $\chi$-candidates are vertices of degree at least $k-1$).
The first step is to compute a set $S \seq V(G)$ of size $\ca O(p)$ such that all cycles in $G$ have an edge in $G[S]$ and each vertex not in $S$ is close to at most one vertex of $S$. For simplicity, assume that $p = |S|$.
The second step is to guess a coloring $\chi$ of $S$ and a set $B \seq S$; now it suffices to decide whether there is a $b$-coloring $\psi$ of $G$ such that $\psi$ extends $\chi$, and each vertex in $B$ is a $b$-vertex in $\psi$.
We may assume that $\chi(S) = [p]$ and $\chi(B) = [b]$ for some integer $b \le p$.

\begin{figure}[t]
\scalebox{0.8}{
\begin{tikzpicture}[every node/.style={draw, rectangle}]
\node[align=center] (1) {Computing $S$, \\ Lemma~\ref{lem:S-exists}};
\node[align=center, below of = 1, yshift=-0.5cm] (2) {Guessing $\chi$ and $B$, \\ Definition~\ref{def:S-profile}};
\node[align=center, right of = 2, xshift=3.3cm] (3) {Guessing a color plan,\\ Definition~\ref{def:color-plan}};
\node[align=center, above of = 3, yshift=0.5cm] (4) {Computing a color\\ realization,  Lemma~\ref{lem:compute-damage-free}};
\node[align=center, right of = 4, xshift=4cm] (5) {Getting rid of the\\ pivot, Section~\ref{sub:getting-rid}};
\node[align=center, below of = 5, yshift=-1cm] (6) {Coloring the\\ neighborhoods of \\ candidates, Section~\ref{sub:coloring-neighbors}};
\node[align=center, right of = 6, xshift=3.5cm] (7) {Coloring the\\ remaining vertices, \\ Lemma~\ref{lem:finish-coloring}};

\draw[-{Latex[length=2mm]}] (1.south) -- (2.north);
\draw[-{Latex[length=2mm]}] (2.east) -- (3.west);
\draw[-{Latex[length=2mm]}] (3.north) -- (4.south);
\draw[-{Latex[length=2mm]}] (4.east) -- (5.west);
\draw[-{Latex[length=2mm]}] (5.south) -- (6.north);
\draw[-{Latex[length=2mm]}] (6.east) -- (7.west);
\end{tikzpicture}}

\caption{Outline of the algorithm for \bcoloring parameterized by the feedback edge number. Each step refers to the relevant part of the proof.}
\label{fig:overview}
\end{figure}
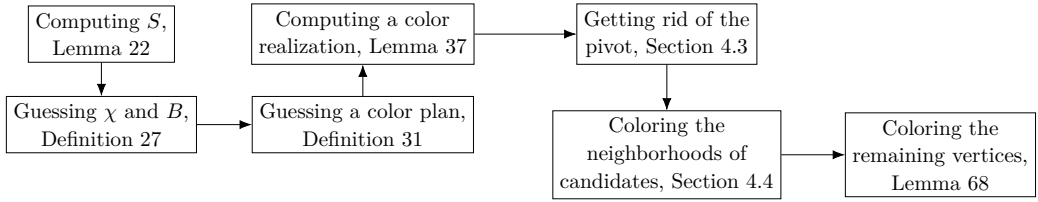

The next step is to choose $b$-vertices for colors in $[b+1, k]$.
What needs to be carefully handled are the colors in $[b+1, p]$ because we need to ensure that $\psi$ will be proper and that vertices in $B$ can become $b$-vertices.
For example, if a vertex $u \in B$ has degree $k-1$ and it has a neighbor of color $p$ in $\chi$, then the $b$-vertex for $p$ cannot be adjacent to $u$.
We begin by guessing a so-called \emph{color plan}, which is a function determining the location of $b$-vertices for the colors in $[b+1, p]$.
Crucially, the number of color plans will be small.
The next step is to compute a so-called \emph{color realization} $\rho$, which is an injective assignment of colors in $[b+1, k]$ to $\chi$-candidates outside of $S$.
We show that using a ``valid'' color plan, we can compute a color realization such that each chosen candidate can indeed become a $b$-vertex.
More precisely, each vertex in $B \cup \range(\rho)$ will be $\chi_\rho$-candidate, where $\chi_\rho = \chi \cup \rho^{-1}$ is the coloring obtained by ``extending'' $\chi$ according to $\rho$.

Now it may occur that the coloring $\chi_\rho$ we have constructed has a ``pivot'', i.e., a vertex that cannot be assigned any color, see the above described algorithm for \textsc{$b$-Chromatic Number} on trees for an idea what a pivot is.
Sometimes, the realization $\rho$ can be slightly modified to get rid of the pivot.
However, this is not straightforward because we must ensure that each vertex in $B \cup \range(\rho')$ is a $\chi_{\rho'}$-candidate, where $\rho'$ is the new realization.
On the other hand, sometimes the pivot is ``unavoidable'', and we must try a different pair $(\chi, B)$.
In fact, there are more properties we need $\rho$ to have but achieving ``pivot-freeness'' will be the most challenging.

Once we have a ``feasible'' color realization $\rho$, we are ready to extend $\chi_\rho$ into a $b$-coloring of $G$.
First, we color the neighborhoods of all chosen candidates.
Roughly speaking, we proceed in a BFS-order from $S$, but there may be an ``exceptional'' candidate that needs to be handled first.
Finally, we color the remainder of the graph.
This last step is also non-trivial because there may be a vertex which already sees all colors.
However, one property we require from $\rho$ will ensure that this cannot happen.

\section{Preliminaries}\label{sec:prelims}
For integers $i$ and $j$, we let $[i,j] := \{a \in \mathbb Z \sep i \le a \le j\}$ and $[i] := [1, i]$.
For a set $S$, we denote the power set of $S$ by $2^S$, and $\bigcup S := \{u \sep \exists a \in S\colon u\in a\}$.
When we say a \emph{function}, we mean a total function.
Let $f\colon A \rightarrow B$ be a function.
The \emph{range} of $f$, denoted $\range(f)$, is the set $\{b \in B\sep \exists a\in A\colon f(a) = b\}$.
The \emph{kernel} of $f$, denoted $\ker(f)$, is the set $\{(a, a')\in A^2\sep f(a) = f(a')\}$.
For $A'\seq A$, we denote the function $\{(a, b) \in f \sep a \in A'\}$ by $f \upharpoonright A'$.
If $g\colon B \rightarrow C$ is a function, then $g \circ f\colon A \rightarrow C$ is the function such that for each $a \in A$, we have $(g \circ f)(a) = g(f(a))$.
For $a \in A$ and $b \in B$, we define $f[a \mapsto b]$ to be the function $f \setminus \{(a, f(a))\} \cup \{(a, b)\}$.
Instead of $f[a_1 \mapsto b_1][a_2 \mapsto b_2]$, we write $f[a_1 \mapsto b_1, a_2 \mapsto b_2]$.
A \emph{partial function} with domain $A$ is possibly undefined on some elements of $A$; we will use the terminology defined for functions also for partial functions.

Let $G$ be a graph (we will consider only finite simple graphs). We assume familiarity with basic concepts in graph theory, see \cite{Diestel}, and parameterized algorithmics, see \cite{CyganFKLMPPS15}.
Given a set of vertices $U\seq V(G)$, we will use $\overline{U}$ to denote the set $V(G) \setminus U$, $G[U]$ to denote the graph induced on $U$, and $G - U$ to denote the graph $G[\overline{U}]$.
Similarly, for an edge set $F \seq E(G)$, $G-F$ denotes the graph obtained from $G$ by removing the edges in $F$.
The \emph{length} of a path or a cycle is the number of edges it contains.
The \emph{distance} between two vertices $u$ and $v$, denoted $\dist_G(u, v)$,
is the length of the shortest path between them (or $\infty$ if no such path exists).
If $v \in V(G)$ and $U \seq V(G)$ is non-empty, then $\dist_G(v, U) = \min\{\dist_G(u, v)\sep u \in U\}$.
For $U \seq V(G)$ and $v \in V(G)$, a $v$-$U$ path is a path $P$ in $G$ with endpoints $v$ and $u$ such that $V(P)\cap U = \{u\}$.

For $u\in V(G)$, we define $N_G(u) := \{v \in V(G)\sep uv \in E(G)\}$ and $N_G[u] := N_G(u) \cup \{u\}$.
Similarly, for $S\seq V(G)$, we define $N_G(S) = \bigcup_{u \in S} N_G(u)$ and $N_G[S] = N_G(S) \cup S$.
Moreover, we define $N_G^S(u) = N_G(u) \cap S$.
We say that $u, v \in V(G)$ are \emph{false twins in $G$} if $N_G(u) = N_G(v)$ and \emph{true twins in $G$} if $N_G[u] = N_G[v]$. Two vertices are \emph{twins} if they are either false twins or true twins.
When $G$ is clear from the context, we will omit the subscript and write, e.g., $\dist(u, v)$ or $N(u)$.

If $\sim$ is an equivalence relation on $V(G)$, then $G/\mathord\sim$ is the graph such that $V(G/\mathord\sim)$ are the equivalence classes of $\sim$, and $CD \in E(G/\mathord\sim)$ if and only if there are vertices $u,v \in V(G)$ such that $u \in C$, $v\in D$, and $uv \in E(G)$.

\subparagraph*{Structural parameters}

Given a graph class $\ca G$ and a graph $G$, we say that $S\seq V(G)$ is a \emph{$\ca G$-modulator of $G$} if $G-S \in \ca G$. An integer $p$ is the \emph{distance to $\ca G$ of $G$} if $p$ is the size of a minimum $\ca G$-modulator of $G$.
A \emph{cluster graph} is a disjoint union of cliques and a \emph{co-cluster graph} is a complement of a cluster graph, i.e., a complete multipartite graph. Let us remark that a minimum cluster-modulator of graph can be computed in \FPT time.

\begin{theorem}[\cite{boral2016fast}]\label{thm:find-S}
Given an $n$-vertex graph $G$ with distance to cluster $p$, a minimum cluster-modulator $S\seq V(G)$ can be computed in time $\ca O(1.92^p \cdot n^2)$.
\end{theorem}

An edge set $F \seq E(G)$ is called a \emph{feedback edge set} if $G-F$ is acyclic, and the \emph{feedback edge number} of $G$ is the size of a minimum feedback edge set in $G$.
A \emph{dangling path} in $G$ is a path of vertices which all have degree $2$ in $G$, and a \emph{dangling tree} in $G$ is an induced subtree in $G$ which becomes a connected component after removing at most one edge from $G$. The following lemma states that a graph of bounded feedback edge number can be decomposed into dangling trees, dangling paths, and a bounded-size subgraph.
This lemma was implicitly used in \cite{DBLP:journals/siamdm/BalabanGR25}, but we provide a proof for completeness.

\begin{lemma}\label{lem:fen-structure}
If $G$ is a connected graph with feedback edge number $p$ that contains no dangling trees,
then there is a set $S \seq V(G)$ such that $|S| \le 4p$, and if $\ca P$ is the set of connected components of $G - S$, then $|\ca P| \le 4p$ and each $P\in\ca P$ is a dangling path in $G$.
Moreover, $S$ can be found in polynomial time.
\end{lemma}

\begin{proof}
Let $F \seq E(G)$ be a minimum feedback edge set in $G$,
let $S_0 = \{u \in V(G) \sep u$ is incident to an edge of $F$ in $G\}$, $S_1 = \{u \in V(G) \sep u$ has degree higher than $2$ in $G\}$, and $S = S_0 \cup S_1$.
Since $F$ can be found in linear time using the DFS algorithm, $S$ can be efficiently computed as well.
It is easy to see that $|S_0| \le 2p$. Let $T = G - F$; clearly, $T$ is a tree. Observe that all leaves of $T$ belong to $S_0$ (a leaf not in $S_0$ would be single-vertex dangling tree), which means that $T$ has at most $2p$ leaves. By simple induction, $T$ contains at most $2p$ vertices of degree higher than $2$. Now observe that a vertex $u \in S_1\setminus S_0$ has degree higher than $2$ not only in $G$ but also in $T$, which implies that $|S| \le 4p$.

Next, observe that each vertex $u \in V(G) \setminus S$ has degree exactly $2$ in $G$ (a degree-1 vertex would belong to a dangling tree). Since $G - S$ is acyclic, each $P\in\ca P$ is indeed a dangling path in $G$.
Finally, we prove that $|\ca P|\le 4p$. Let us choose an arbitrary vertex of $S$ to be the root of $T$. For $P \in\ca P$, let $S_P \seq S$ be the two-element set containing the vertices of $S$ that are adjacent to a vertex of $P$ in $G$.
Let $f\colon \ca P\rightarrow S$ be a function such that $f(P)$ is the vertex of $S_P$ that is farther from the root in $T$.
Since $f(P) = u$ implies that the parent of $u$ in $T$ belongs to $P$, we know that $f$ is injective, which implies $|\ca P| \le |S| \le 4p$.
\end{proof}

\subparagraph*{Colorings and homomorphisms}

A (partial) \emph{$k$-coloring of $G$} is a (partial) function $\chi\colon V(G)\rightarrow [k]$.
We say that $\chi$ is \emph{proper} if $\chi(u)\ne\chi(v)$ for each $uv\in E(G)$.
We say that $u\in V(G)$ is a \emph{$b$-vertex in $\chi$} if $\chi(N[u]) = [k]$, and that $\chi$ is a \emph{$k$-$b$-coloring} if it is proper and for each color $c \in [k]$, there is a $b$-vertex $u$ in $\chi$ such that $\chi(u) = c$.
A coloring is a \emph{$b$-coloring} if it is a $k$-$b$-coloring for some integer $k$. 
The decision problem \textsc{$b$-coloring} receives as input a graph $G$ and an integer $k$ and asks whether there is a $k$-$b$-coloring of $G$.

If $G$ and $H$ are graphs, then a function $h\colon V(G) \rightarrow V(H)$ is a \emph{homomorphism from $G$ to $H$}, denoted $h\colon G\rightarrow H$, if $uv \in E(G)$ implies $h(u)h(v) \in E(H)$ for every $u,v \in V(G)$.
We say that $h$ is a \emph{two-way} homomorphism if also $h(u)h(v) \in E(H)$ implies $uv \in E(G)$. The following lemma says that, under some conditions, surjective two-way homomorphisms preserve and reflect the existence of a $b$-coloring.

\begin{lemma}\label{lem:two-way-hom}
Let $G$ and $H$ be graphs and let $h\colon G\rightarrow H$ be a surjective two-way homomorphism.
If $\chi$ is a $k$-$b$-coloring of $H$, then $\psi = \chi \circ h$ is a $k$-$b$-coloring of $G$.
Conversely, if $\psi$ is a $k$-$b$-coloring of $G$ such that $\ker(h) \seq \ker(\psi)$, then there is a $k$-$b$-coloring $\chi$ of $H$ such that $\psi = \chi \circ h$.
\end{lemma}
\begin{proof}
Let $G$, $H$, and $h$ be as in the statement.
Observe that $G$ is isomorphic to a graph that can be obtained from $H$ by adding one or more false twins to some vertices of $H$; hence, we may assume that $H$ is an induced subgraph of $G$. Indeed, for $u \in V(H)$, the set $h^{-1}(u) \seq V(G)$ consists of $u$ and all false twins added to $u$ (in particular, $h(u) = u$).
If $\chi$ is a $k$-$b$-coloring of $H$, then $\psi = \chi \circ h$ is clearly a proper coloring of $G$, and each $b$-vertex in $\chi$ is also a $b$-vertex in $\psi$ (for the same color), so $\psi$ is a $k$-$b$-coloring of $G$.

Conversely, assume that $\psi$ is a $k$-$b$-coloring of $G$ such that $\ker(h) \seq \ker(\psi)$.
Let us define $\chi = \psi \upharpoonright V(H)$; clearly, $\chi$ is a proper coloring of $H$ and $\psi = \chi \circ h$. If $u$ is a $b$-vertex in $\psi$, then $h(u)\in V(H)\seq V(G)$ is also a $b$-vertex in $\psi$ because $u$ and $h(u)$ are twins in $G$ that are colored with the same color by $\psi$. Moreover, if $uv \in E(G)$, then $h(u)h(v) \in E(H)$ and $\chi(h(v)) = \psi(v)$. Hence, $h(u)$ is a $b$-vertex also in $\chi$ (for the same color as in $\psi$), and so $\chi$ is a $k$-$b$-coloring of $H$.
\end{proof}

\subparagraph*{Structural parameterizations of \textsc{$b$-coloring}}

Using dynamic programming along a tree-decomposition, it can be shown that \textsc{$b$-coloring} is \FPT when parameterized by tree-width and the number of colors.

\begin{proposition}[\cite{JaffkeLL24, Davi_Andrade}]\label{prop:tree-width}
The \textsc{$b$-coloring} problem can be solved in time $2^{\ca O(w\cdot k)}\cdot n$, where $n$ is the number of vertices in the input graph, $w$ is its tree-width, and $k$ is the number of colors.
\end{proposition}

Proposition~\ref{prop:tree-width} was used in \cite{JaffkeLL24} to show that \textsc{$b$-coloring} is \FPT when parameterized by the vertex cover number.
Now we use this proposition to obtain two other \FPT algorithms.
The \emph{vertex integrity} of a graph $G$ is the smallest integer $p$ such that there is a set $S \subsetneq V(G)$ such that for each connected component $H$ of $G - S$, we have $|S \cup V(H)| \le p$; the relationship of vertex integrity to other parameters is depicted in Figure~\ref{fig:hierarchy}.

\begin{proposition}\label{prop:vertex-integrity}
The \textsc{$b$-coloring} problem can be solved in time $2^{\ca O(p^2)}\cdot n$, where $n$ is the number of vertices in the input graph and $p$ is its vertex integrity.
\end{proposition}
\begin{proof}
Let $(G, k)$ be the input instance, where $G$ is a graph with $n$ vertices and vertex integrity $p$.
First, we use the algorithm from \cite{DrangeDH16} to compute a set $S \subsetneq V(G)$ certifying the vertex integrity of $G$ in time $\ca O(p^{p+1} \cdot n)$.
Using $S$, it is easy to compute a path decomposition of width at most $p$.
If $k \le p$, then we may use Proposition~\ref{prop:tree-width} to solve the problem.
Otherwise, we answer that $G$ admits no $k$-$b$-coloring.
Indeed, in such a coloring, there would have to be a $b$-vertex $u \in V(G) \setminus S$ of degree at least $k - 1 \ge p$, which is impossible.
\end{proof}

The second parameter we consider is \emph{carving-width}, see \cite{seymour1994call}.
This width parameter is certified by a so-called \emph{carving}: a subcubic tree $T$ and a bijection between the leaves of $T$ and the vertices of $G$ (we will not need the full definition).
Note that the following proposition implies that \bcoloring is \FPT when parameterized by one of two other well-known parameters, namely \emph{cut-width} and \emph{band-width}; see \cite{diaz2002survey} and Figure~\ref{fig:hierarchy}. 

\begin{proposition}\label{prop:carving-width}
The \textsc{$b$-coloring} problem can be solved in time $f(w) \cdot n$, where $n$ is the number of vertices in the input graph, $w$ is its carving-width, and $f$ is a computable function.
\end{proposition}
\begin{proof}
Let $(G, k)$ be the input instance, where $G$ is a graph with $n$ vertices and carving-width $w$.
First, we use the algorithm from \cite{ThilikosSB00} to compute a carving $C$ certifying the carving-width of $G$ in time $f_0(w) \cdot n$ for some computable function $f_0$.
Using $C$, we can construct a tree decomposition of width at most $3w$, see \cite{ThilikosSB00}.
If $k \le w+1$, then we may use Proposition~\ref{prop:tree-width} to solve the problem.
Otherwise, we answer that $G$ admits no $k$-$b$-coloring.
Indeed, in such a coloring, there would have to be a $b$-vertex of degree at least $w+1$, and there are no such vertices in graphs of carving-width $w$, see \cite{BelmonteHKPT13}.
\end{proof}

\cite{JaffkeLS23} proved that \bcoloring is \textsf{XNLP}-complete when parameterized by path-width.
They obtained this result by reducing from \textsc{Circulating Orientation} parameterized by path-width.
Since this problem is \textsf{W[1]}-hard when parameterized by tree-depth, see \cite[Theorem~6.1.1.]{jong2024analysing}, and it can be easily observed that the reduction to \bcoloring by \cite{JaffkeLS23} preserves not only bounded path-width but also bounded tree-depth, we obtain the following result.

\begin{proposition}\label{prop:tree-depth-hardness}
The \textsc{$b$-coloring} problem is \textsf{W[1]}-hard when parameterized by the tree-depth of the input graph.
\end{proposition}

\section{Distance to Co-Cluster}\label{sec:co-cluster}

In this section, we prove the following theorem. Recall that its proof was briefly sketched in Section~\ref{sec:intro}. 

\begin{theorem}\label{thm:co-cluster}
The \textsc{$b$-Coloring} problem can be solved in time $2^{2^{\ca O(p)}} \cdot n^{\ca O(1)}$, where $p$ is the distance to co-cluster of the input $n$-vertex graph.
\end{theorem}

In Section~\ref{sub:dcc-init}, we provide necessary definitions, such as the definition of a signature, and we prove several simple lemmas.
In Section~\ref{sub:dcc-any}, we prove the key technical lemma of this section, which allows us to build a coloring based on a signature.
In Section~\ref{sub:dcc-not-used}, we show how colors not used on the co-cluster-modulator $S$ can be handled.
Finally, Theorem~\ref{thm:co-cluster} is proven in Section~\ref{sub:dcc-finish}.

\subsection{Initial setup}\label{sub:dcc-init}

Let us fix an instance of the \textsc{$b$-coloring} problem $(G, k)$.
Let $S \seq V(G)$ be a minimum co-cluster-modulator of $G$, let $p = \min(|S|, k)$, and let $\ca U$ be the set containing all maximal independent sets of $G - S$. Note that if $U_1, U_2 \in \ca U$, $v_1 \in U_1$, and $v_2 \in U_2$, then $v_1v_2 \in E(G)$ if and only if $U_1 \ne U_2$.

In the following definition, we assign a type to each set $U \in \ca U$ based on the neighborhoods that vertices in $U$ have in $S$.
If many vertices have the same neighborhood in $S$, then we do not need to record the precise number of such vertices, which makes the number of types bounded by a function of $p$, see Figure~\ref{fig:co-cluster}.

\begin{figure}[t]
\scalebox{1.1}{
\begin{tikzpicture}[every node/.style={draw, circle, minimum width=4pt, inner sep=0pt}]
\foreach \i in {1,2,3,4} {
    \node (a\i) at (0, \i*0.75) {};
    \node (c\i) at (4, \i*0.75) {};
}
\foreach \i in {4,2,3} {
    \node (b\i) at (2, \i*0.75-0.25) {};
}
\foreach \i in {1,2,3,4} {
\foreach \j in {1,2,3,4} {
    \draw[black!20!white] (a\i)-- (c\j);
}}
\foreach \i in {1,2,3,4} {
\foreach \j in {4,2,3} {
    \draw[black!20!white] (a\i)-- (b\j);
    \draw[black!20!white] (c\i)-- (b\j);
}}
\node[fill=gray, minimum width=6pt] (d) at (1, 3.5) {};
\node[fill=gray, minimum width=6pt] (e) at (3, 3.5) {};

\foreach \i in {2,3,4} {
    \draw (a\i) -- (d);
    \draw (b\i) -- (d);
}
\draw (a1)--(d);
\draw (e)--(c4)--(d);
\draw (e)--(c3)--(d); 
\draw (c2)--(e); 

\node[draw=none] (u) at (0.75, 3.75) {$u$};
\node[draw=none] (v) at (2.75, 3.75) {$v$};
\end{tikzpicture}}
\centering
\caption{An illustration of Definition~\ref{def:type} with $S = \{u, v\}$ and $p = 2$. The edges in $G - S$ are drawn in gray. The left-hand set and the middle set are of the same type, namely $\{(\emptyset, 0), (\{u\}, 3), (\{v\}, 0), (\{u, v\}, 0)\}$.
}
\label{fig:co-cluster}
\end{figure}
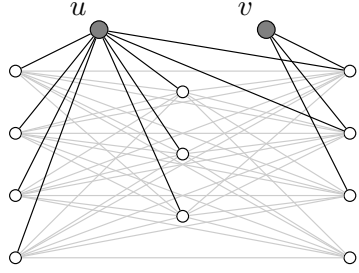

\begin{definition}\label{def:type}
A vertex $u \in \overline{S}$ has \emph{vertex-type} $A \seq S$ if $N^S(u) = A$.
A \emph{set-type} is a function $t\colon 2^S \rightarrow [0, p+1]$. A set $U \in \ca U$ \emph{has set-type} $t$ if for each $A \seq S$, we have $t(A) = \min(p+1, |\{u \in U \sep N^S(u) = A\}|)$. Let $T$ be the set of all set-types. We may simply use the term \emph{type} when it is clear from context whether we mean a vertex-type or a set-type.
\end{definition}

Now we define \emph{signatures}, which will be used to describe each proper $k$-coloring $\psi$ of $G$ in terms of the colors used on $S$.
Let us present an intuition behind the components of a signature.
The first component, $\chi$, will be the restriction of $\psi$ to $S$.
By permuting the colors, we may assume that $\range(\chi) \seq [p]$.
The other parts of a signature describe how colors in $[p]$ behave in $G - S$ according to $\psi$.
Let $U_1, \ldots, U_q \in \ca U$ be the sets such that some color in $[p]$ is used on some set $U_i$; the function $\tau$ describes the types of these sets.
Since $\psi$ is proper, each color in $[p]$ is used on at most one such set: the function $\lambda$ describes which set it is (if any).
Finally, the function $\xi$ describes, for each color $c \in [p]$, the connections to $S$ of vertices in $G - S$ colored with $c$.

\begin{definition}\label{def:signature}
A \emph{signature} is a tuple $\sigma = (\chi, q, \tau, \lambda, \xi)$, where:
\begin{itemize}
\item $\chi\colon S \rightarrow [p]$ is a proper coloring of $G[S]$;
\item $q \in [0, p]$ is an integer;
\item $\tau \colon [q] \rightarrow T$, $\lambda\colon [p] \rightarrow [0, q]$, and $\xi\colon [p] \rightarrow 2^{2^S}$ are functions.
\end{itemize}
A signature is required to have the following properties.
\begin{enumerate}
\item It holds that $[q]\seq \range(\lambda)$.\label{sign:surjective}
\item For each $t \in T$, $|\tau^{-1}(t)|$ is at most the number of sets of type $t$ in $\ca U$.\label{sign:set-types}
\item For each $i \in [q]$ and $A \seq S$, we have $|\{c \in [p]\colon \lambda(c) = i \land A \in \xi(c)\}| \le \tau(i)(A)$.\label{sign:vertex-types}
\end{enumerate}
\end{definition}

Let us remark that the last two properties required by Definition~\ref{def:signature} ensure that the signature is ``realizable'' by some coloring: Property~\ref{sign:set-types} states that there are enough sets of each set-type, and Property~\ref{sign:vertex-types} states that there are enough vertices of each vertex-type.

It is not hard to see that the number of signatures is bounded by a function of $p$.

\begin{observation}\label{obs:number-of-signatures}
The size of $T$ and the number of signatures are both in $2^{2^{\ca O(p)}}$.
\end{observation}
\begin{proof}
The fact that $|T| \in 2^{2^{\ca O(p)}}$ is obvious.
There are $2^{\ca O(p\cdot \log p)}$ choices for $\chi$, $\ca O(p)$ choices for $q$, $\ca O(|T|^p) = 2^{2^{\ca O(p)}}$ choices for $\tau$, $2^{\ca O(p\cdot \log p)}$ choices for $\lambda$, and $2^{2^{\ca O(p)}}$ choices for $\xi$. Altogether, there $2^{2^{\ca O(p)}}$ signatures.
\end{proof}

Now we formally define the correspondence between signatures and colorings of $G$, which we have foreshadowed.

\begin{definition}\label{def:represent}
Let $\psi\colon V(G) \rightarrow [k]$ be a partial coloring and $\sigma = (\chi, q, \tau, \lambda, \xi)$ be a signature.
We say that $\psi$ \emph{has signature $\sigma$} (or that it is a \emph{$\sigma$-coloring})
if $\chi = \psi \upharpoonright S$ and:
\begin{enumerate}
\item there is a set $\ca U_\psi = \{U_1, \ldots, U_q\} \seq \ca U$ such that $\psi^{-1}([p]) \seq S\cup\bigcup \ca U_\psi$;\label{repr:blocks}
\item for each $i \in [q]$, $\tau(i)$ is the set-type of $U_i$;\label{repr:types}
\item for each $c \in [p]$, if $\lambda(c) = 0$, then $\psi^{-1}(c) \seq S$, and if $\lambda(c) = i \ne 0$, then $c \in \psi(U_i)$;\label{repr:f0}
\item for each $c \in [p]$, if $\lambda(c) = i \ne 0$, then $\xi(c) = \{A\seq S\sep\exists u\in U_i\colon  \psi(u) = c\land N^S(u) = A\}$;\label{repr:f1}
\end{enumerate}
\end{definition}

We say that $\psi$ is a \emph{minimal $\sigma$-coloring} if it is a $\sigma$-coloring, $\psi(V(G)) \seq [p]$, and for each $c \in [p]$ and $A \in \xi(c)$, there is a \emph{unique} vertex $u \in U_{\lambda(c)}$ such that $\psi(u) = c$ and $N^S(u) = A$.
It can be easily observed that minimal $\sigma$-colorings are exactly inclusion-wise minimal colorings that have signature $\sigma$.
Let us prove that a minimal $\sigma$-coloring can be easily computed.

\begin{lemma}\label{lem:minimal-sigma-exists}
For a signature $\sigma$, a minimal $\sigma$-coloring $\psi$ exists and can be computed in polynomial time. Moreover, the number of vertices colored by $\psi$ is $2^{\ca O(p)}$.
\end{lemma}
\begin{proof}
Let $\sigma = (\chi, q, \tau, \lambda, \xi)$ be a signature. We define $\psi(u) = \chi(u)$ for each $u \in S$. For each $i \in [q]$, we find a set $U_i \in \ca U$ of set-type $\tau(i)$. We ensure that $U_i\ne U_j$ for $i\ne j$; this can be achieved by Property~\ref{sign:set-types} of Definition~\ref{def:signature}. For each $c \in[p]$ such that $\lambda(c) = i \ne 0$ and each $A\in \xi(c)$, we color a single vertex $u\in U_i$ of vertex-type $A$ with color $c$; there are enough vertices of vertex-type $A$ in $U_i$ by Property~\ref{sign:vertex-types} of Definition~\ref{def:signature}. The obtained coloring $\psi$ is clearly a minimal $\sigma$-coloring, and the described construction can be executed in polynomial time.
Finally, observe that at most $p + p\cdot 2^p \in 2^{\ca O(p)}$ vertices are colored by $\psi$: $p$ in $S$, and for each color $c \in [p]$, at most $2^p$ vertices are colored with $c$. 
\end{proof}

The following lemma states the ``converse'' of Lemma~\ref{lem:minimal-sigma-exists}: each coloring has a signature.

\begin{lemma}\label{lem:b-coloring-has-signature}
If $\psi$ is a $k$-$b$-coloring of $G$, then there is a signature $\sigma = (\chi, q, \tau, \lambda, \xi)$ such that $\psi$ is a $\sigma$-coloring (up to permutation of colors in $\psi$).
\end{lemma}
\begin{proof}
We may assume that $\psi(S) \seq [p]$. 
Let us now construct $\sigma$ based on $\psi$.
Let $\chi = \psi \upharpoonright S$. 
Let $\ca U_\psi = \{U_1, \ldots, U_q\}$ be the minimal subset of $\ca U$ such that $\psi^{-1}([p]) \seq S\cup\bigcup \ca U_\psi$, and let $\tau(i)$ be the set-type of $U_i$ for $i \in [q]$.
Since $G[U_i \cup U_j]$ is a complete bipartite graph for each $i,j \in [q]$ and $\psi$ is proper, there is, for each $c \in [p]$, at most one $i \in [q]$ such that $c \in \psi(U_i)$. If such $i$ exists, we define $\lambda(c) = i$; otherwise, $\lambda(c) = 0$. Finally, we define
\[\xi(c) = \begin{cases}
\emptyset \text{ if } \lambda(c) = 0\\
\{A \seq S \sep \exists u \in U_i\colon \psi(u) = c \land N^S(u) = A\}  \text{ if } \lambda(c) = i \ne 0.
\end{cases}\]

Now we need to verify that $\sigma$ is indeed a signature. Suppose for contradiction that $q > p$. By minimality of $\ca U_\psi$, we have $\psi(U_i) \cap [p] \ne \emptyset$ for every $i \in [q]$. Hence, by the pigeonhole principle, there is $c \in [p]$ and $i, j \in [q]$ such that $i\ne j$ and $c \in \psi(U_i) \cap \psi(U_{j})$, which is a contradiction with $\psi$ being proper.
Hence, $q \in [0, p]$, as required.

What remains is to show that the three properties required by Definition~\ref{def:signature} are satisfied.
First, Property~\ref{sign:surjective} holds by minimality of $\ca U_\psi$.
Second, Property~\ref{sign:set-types} is satisfied because for each $i \in [q]$, the corresponding set $U_i$ has set-type $\tau(i)$.
Finally, let $i \in [q]$, $A \seq S$, and let $m \in [0, p]$ be the number of colors from $[p]$ that have a vertex in $U_i$ of vertex-type $A$, i.e., $m$ is the left-hand side of Property~\ref{sign:vertex-types}. By the pigeonhole principle, $m \le m' := |\{u \in U_i \sep N^S(u) = A\}|$. By Definition~\ref{def:type}, we know that $\tau(i)(A) \in \{m', p+1\}$, and so $m \le \tau(i)(A)$, and Property~\ref{sign:vertex-types} is satisfied.
Therefore, $\sigma$ is indeed a signature (and the fact that $\psi$ is a $\sigma$-coloring is obvious from the construction).
\end{proof}

Let us now observe that the behavior of colors \emph{not} in $[p]$ is significantly restricted in any $b$-coloring of $G$. 

\begin{observation}\label{obs:anonym}
If $\psi$ is a $k$-$b$-coloring of $G$ with signature $\sigma$, then for distinct colors $c, d \in [p+1, k]$, there are sets $U_c, U_d \in \ca U$ such that $\psi^{-1}(c) \seq U_c$, $\psi^{-1}(d) \seq U_d$, and $U_c \ne U_d$.
\end{observation}
\begin{proof}
Let $c \in [p+1,k]$ be a color. Since $\psi$ is a $b$-coloring, there is a vertex of color $c$ in $\psi$, and since $\psi$ is proper and $\psi(S) \seq [p]$ by Definition~\ref{def:represent}, there is a set $U_c\in\ca U$ such that $\psi^{-1}(c) \seq U_c$.
If $U_c = U_d$ for distinct colors $c, d \in [p+1, k]$, then no vertex of color $c$ would have a neighbor of color $d$, which is a contradiction with $\psi$ being a $b$-coloring.
\end{proof}

\subsection{Any minimal $\sigma$-coloring can be used}\label{sub:dcc-any}

Informally, the following lemma says that two minimal $\sigma$-colorings are interchangeable, i.e., if one such coloring can be extended into a $k$-$b$-coloring of $G$, then any other such coloring can be extended as well.
It is perhaps the most technical lemma of this section.

\begin{lemma}\label{lem:any-minimal-coloring-can-be-used}
If $\sigma = (\chi_S, q, \tau, \lambda, \xi)$ is a signature, $\psi\colon V(G) \rightarrow [k]$ is a $b$-coloring of $G$ with signature $\sigma$, and $\chi \colon V(G) \rightarrow [p]$ is a minimal $\sigma$-coloring, then there is a $b$-coloring $\chi' \colon V(G) \rightarrow [k]$ of $G$ such that $\chi \seq \chi'$.
\end{lemma}

\begin{proof}
Let $\ca U_\chi = \{U_1, \ldots, U_q\}$ and $\ca U_\psi = \{V_1, \ldots, V_q\}$ be the sets as per Definition~\ref{def:represent}.
For $u,v \in V(G)$, we say that $u \sim_0 v$ if $u,v \in W$ for some $W \in \ca U_\chi\cup \ca U_\psi$, $u$ and $v$ are twins, and $\psi(u) = \psi(v)$. Let $\sim$ be the equivalence relation on $V(G)$ generated by $\sim_0$, let $H = G/\mathord\sim$, and let $\eta\colon G \rightarrow H$ be the natural homomorphism, i.e., $\eta(u) = [u]_\sim$.
Since each class $[u]_\sim \in V(H)$ is an independent set of twins in $G$, we know that $\eta$ is a surjective two-way homomorphism.
Moreover, since $u \sim v$ implies $\psi(u) = \psi(v)$, we have that $\ker(\eta) \seq \ker(\psi)$, and so by Lemma~\ref{lem:two-way-hom}, there is a $k$-$b$-coloring $\psi_H$ of $H$ such that $\psi = \psi_H \circ\eta$, see Figure~\ref{fig:diagram}.

We define a bijection $g\colon\ca U\rightarrow\ca U$ as follows. For $U_i\in\ca U_\chi$, we define $g(U_i)=V_i$, and for $U\notin\ca U_\chi\cup\ca U_\psi$, we define $g(U)=U$. Now it only remains to define $g(V)$ for each $V\in\ca U_\psi\setminus\ca U_\chi$. For each set-type $t \in T$, let $\ca U^t = \{U \in \ca U_\psi \cap \ca U_\chi \sep $the type of $U$ is $t\}$, let $\ca U_\chi^t = \{U \in \ca U_\chi\setminus \ca U_\psi \sep $the type of $U$ is $t\}$, and $\ca U_\psi^t = \{U \in \ca U_\psi \setminus\ca U_\chi \sep $the type of $U$ is $t\}$.
Observe that $|\ca U_\chi^t| =  |\tau^{-1}(t)| - |\ca U^t| = |\ca U^t_\psi|$. Hence, there is a bijection $g_t\colon \ca U^t_\psi \rightarrow \ca U^t_\chi$, and we may define $g(V) = g_t(V)$ for each $V\in\ca U_\psi\setminus\ca U_\chi$ of set-type $t$. It can be easily observed that $g$ is indeed a bijection and that the set-type of $U\in\ca U$ equals the set type of $g(U)$.

\begin{figure}[t]
\adjustbox{scale=1.2,center}{
$\begin{tikzcd}
	& {K_k} & {K_p} \\
	G & H & G & {G_\chi} \\
	&& {G[S]}
	\arrow[hook', from=1-3, to=1-2]
	\arrow["\psi", from=2-1, to=1-2]
	\arrow["\eta", from=2-1, to=2-2]
	\arrow["{\psi_H}"', from=2-2, to=1-2]
	\arrow["{\chi'}"', from=2-3, to=1-2]
	\arrow["h"{pos=0.6}, from=2-3, to=2-2]
	\arrow["\chi"', from=2-4, to=1-3]
	\arrow[hook', from=2-4, to=2-3]
	\arrow[hook', from=3-3, to=2-1]
	\arrow[hook, from=3-3, to=2-4]
\end{tikzcd}$
}
\caption{A commutative diagram depicting the proof of Lemma~\ref{lem:any-minimal-coloring-can-be-used}. The arrows are graph homomorphisms: colorings are viewed as homomorphisms into cliques ($K_k$ and $K_p$ are cliques on vertex sets $[k]$ and $[p]$, respectively). $G_\chi$ is the subgraph of $G$ induced by the vertices on which $\chi$ is defined. The unnamed arrows are inclusion maps. Note that all named arrows are surjective and that $\eta$ and $h$ are two-way homomorphisms.}
\label{fig:diagram}
\end{figure}
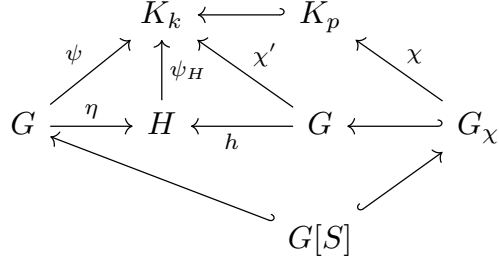

Our goal now is to define \emph{another} surjective two-way homomorphism $h\colon G\rightarrow H$ so that we can define $\chi' = \psi_H \circ h$, and then again use Lemma~\ref{lem:two-way-hom} to show that $\chi'$ is a $k$-$b$-coloring of $G$.
Moreover, we will need to ensure that whenever $\chi(u)$ is defined for some $u \in V(G)$, then $\chi'(u) = \chi(u)$, see Figure~\ref{fig:diagram}.
Another property of $h$ we will ensure is that if $u \in U \in \ca U$ and $h(u) = [v]_\sim$, then $u$ and $v$ have the same vertex-type and $v \in g(U)$; this property makes sense because all vertices in $[v]_\sim$ have the same vertex-type and $[v]_\sim \seq W$ for some $W \in\ca U$.

First, let $u\in U_i\in\ca U_\chi$ be a vertex of type $A\seq S$. If $\chi(u)$ is defined, then since $\chi$ and $\psi$ have the same signature, there is a vertex $v \in V_i$ of type $A$ such that $\psi(v) = \chi(u)$, and we set $h(u) = [v]_\sim$.
Observe that $\chi'(u) = (\psi_H \circ h)(u) = \psi_H([v]_\sim) = (\psi_H \circ\eta)(v) = \psi(v) = \chi(u)$ as required. 
Now suppose $\chi(u)$ is undefined. If there is a color $c \in [p+1, k]$ and a vertex $v \in V_i$ of type $A$ such that $\psi(v) = c$, then we define $h(u) = [v]_\sim$.
Otherwise, there is a vertex $v \in V_i$ of type $A$ (since $U_i$ and $V_i$ have the same set-type): we choose any such vertex $v$ 
and we define $h(u) = [v]_\sim$.

Second, let $v\in V \in\ca U_\psi\setminus \ca U_\chi$ be a vertex of type $A\seq S$. Since the set-type of $V$ equals the set-type of $g(V)$, there is a vertex $u \in g(V)$ of vertex-type $A$, and we define $h(v) = [u]_\sim$.
Observe that $\chi(v)$ is undefined because $v \notin \bigcup \ca U_\chi$.
Finally, for all other $u \in V(G)$, we define $h(u) = [u]_\sim$.
If $\chi(u)$ is defined, then $u \in S$ (because $u \notin \bigcup \ca U_\chi$), and $\chi'(u) = (\psi_H \circ h)(u) = (\psi_H \circ\eta)(u) = \psi(u) = \chi_S(u) = \chi(u)$ as required.
In all cases, we have showed that $\chi \seq\chi'$.

\subparagraph*{$h$ is a two-way homomorphism.}
Let $u,v \in V(G)$.
If $h(u) = [u]_\sim$ and $h(v) = [v]_\sim$, then $h(u)h(v) \in E(G)$ if and only if $uv \in E(G)$ by definition of $\sim$.
Hence, without loss of generality, we may assume that $h(u) = [w]_\sim \ne [u]_\sim$.
By construction of $h$, this implies that $u \in U\in\ca U$, that $u$ and $w$ have the same vertex-type, and that $w \in g(U)$. If $v \in S$, then $uv \in E(G) \leftrightarrow wv \in E(G) \leftrightarrow [w]_\sim[v]_\sim \in E(H) \leftrightarrow h(u)h(v) \in E(H)$. Otherwise, $v \in V \in\ca U$, $h(v) = [x]_\sim$, and $x \in g(V)$. Using the fact that $g$ is a bijection (and that $G-S$ is a co-cluster graph), we deduce that $uv \in E(G) \leftrightarrow U \ne V \leftrightarrow g(U) \ne g(V) \leftrightarrow wx \in E(G) \leftrightarrow [w]_\sim[x]_\sim \in E(H) \leftrightarrow h(u)h(v) \in E(H)$. Hence, $h$ is indeed a two-way homomorphism.

\subparagraph*{$h$ is surjective}
Let $[v]_\sim \in V(H)$ and suppose that $h(v) \ne [v]_\sim$, which means that $v \in W \in \ca U_\chi\cup\ca U_\psi$, by construction of $h$. Let $A$ be the vertex-type of $v$ and let $\psi(v) = c$. First, suppose that $v \in V_i \in \ca U_\psi$. If $c \in [p]$, then by Definition~\ref{def:represent}, there is a vertex $u \in U_i$ of vertex-type $A$ such that $\chi(u) = c$, and $h(u) = [v]_\sim$ by construction of $h$.

Now suppose $c \in [p+1, k]$.
Let $m = |\{w \in V_i \sep N^S(w) = A\}|$ and $m' = |\{w \in U_i \sep N^S(w) = A\}|$.
By Definition~\ref{def:type}, $m' \ge \min\{m, p+1\}$ because $V_i$ and $U_i$ have the same set-type.
Let $C = \{d \in [p]\colon \lambda(d) = i \land A \in \xi(d)\}$ and observe that $m \ge |C|+1$: there are at least $|C|$ vertices of type $A$ colored with colors from $C$ in $V_i$ in $\psi$ plus at least one vertex of color $c$ (namely, $v$). Recall that $\chi$ is a minimal $\sigma$-coloring, and so there are exactly $|C|$ colored vertices of type $A$ in $U_i$ in $\chi$: one for each color in $C$.
Since $p \ge |C|$, we deduce that $m' \ge \min\{m, p+1\} \ge |C|+1$, i.e., there must be a vertex $u$ of type $A$ in $U_i$ uncolored by $\chi$. By construction of $h$, we know that $h(u) = [w]_\sim$ for some $w \in V_i$ of type $A$ such that $\psi(w) \in [p+1, k]$. However, by Observation~\ref{obs:anonym}, $\psi(w) = c$ and $v \sim w$. Hence, $h(u) = [v]_\sim$.

Finally, suppose that $v \in U \in\ca U_\chi\setminus\ca U_\psi$.
Since $g$ is a bijection that preserves set-types, there is a vertex $u \in g^{-1}(U)$ of type $A$.
By construction of $h$, we know that $h(u) = [w]_\sim$ for some vertex $w \in U$ of type $A$.
However, by Observation~\ref{obs:anonym}, $\psi(U) = \{c\}$ because $U \notin \ca U_\psi$.
Hence, $w \sim v$ by definition of $\sim$, which concludes the proof of surjectivity of $h$.

\smallskip

We have proven that $h$ is a surjective two-way homomorphism, which implies, by Lemma~\ref{lem:two-way-hom}, that $\chi'$ is indeed a $k$-$b$-coloring.
\end{proof}

\subsection{Handling colors not used on $S$}\label{sub:dcc-not-used}

Before we continue, we prove a simple lemma concerning two $\sigma$-colorings that are comparable by inclusion.

\begin{lemma}\label{lem:twins}
If $\sigma = (\chi_S, q, \tau, \lambda, \xi)$ is a signature and $\chi, \psi$ are two proper $\sigma$-colorings such that $\chi\seq\psi$, then $\ca U_\psi = \ca U_\chi$ (see Definition~\ref{def:represent}).
Moreover, each vertex $u\in V(G)$ such that $\psi(u) \in [p]$ has a twin $u'$ such that $\chi(u') = \psi(u)$.
\end{lemma}
\begin{proof}
Let $\ca U_\chi = \{U_1, \ldots, U_q\}$, $\ca U_\psi = \{V_1, \ldots, V_q\}$, and let $i\in [q]$ be any index. By Property~\ref{sign:surjective} of Definition~\ref{def:signature}, there is $c \in [p]$ such that $\lambda(c) = i$, and by Condition~\ref{repr:f0} of Definition~\ref{def:represent}, there is $u \in U_i$ such that $\chi(u) = c$. Since $\psi$ is proper, $\psi(u) = c$, and $c \in \psi(V_i)$, we obtain that $u \in V_i$, which means that $V_i = U_i$.

Now suppose that $u\in V(G)$ is a vertex such that $\psi(u) = c \in [p]$. If $\chi(u)$ is defined, we choose $u' = u$. Suppose $\chi(u)$ is undefined. By Definition~\ref{def:represent}, $\chi \upharpoonright S = \chi_S$, which implies that $u \notin S$. Since $\psi$ is a $\sigma$-coloring, we know that $\lambda(c) = i \ne 0$, $u \in U_i$, and $N^S(u) \in \xi(c)$. Since $\chi$ is a $\sigma$-coloring, there is a vertex $u' \in U_i$ with vertex-type $N^S(u)$ such that $\chi(u') = c$. Since $u$ and $u'$ are both in $U_i$, they are indeed twins in $G$. 
\end{proof}

We will need one more definition.

\begin{definition}\label{def:candidates-and-flexible}
Let $\sigma = (\chi_S, q, \tau, \lambda, \xi)$ be a signature and $\chi\colon V(G) \rightarrow [p]$ be a minimal $\sigma$-coloring.
We say that $B \seq V(G)$ is a \emph{$\chi$-candidate subset} if $|B| = p$, $\chi(B) = [p]$, and for each $u \in B$, it holds that $[p] \seq \chi(N[u])$.
Let us fix a $\chi$-candidate subset $B$.
We say that $U \in \ca U$ is a \emph{$(\chi, B)$-candidate} if each $u \in B$ has a neighbor $v \in U$ that is uncolored by $\chi$, and
there is a vertex $u\in U$ that is uncolored by $\chi$ such that $[p] \seq \chi(N(u))$.
We say that $U$ is \emph{$\chi$-flexible} if each $u\in U$ uncolored by $\chi$ has a twin $v \in U$ that is colored by $\chi$.
\end{definition}

We have defined two properties of sets in $\ca U$: \emph{$(\chi, B)$-candidates} are those that may contain a $b$-vertex for some color in $[p+1, k]$, and \emph{$\chi$-flexible} are those that \emph{do not have to} contain such a $b$-vertex.
If some set will have neither of these properties, then $\chi$ cannot be extended into a $b$-coloring of $G$ such that $B$ is a set containing a $b$-vertex for each color in $[p]$.
Another issue occurs when there are too many or too few $(\chi, B)$-candidates that are not $\chi$-flexible: the reason is that we need to assign each color in $[p+1, k]$ to exactly one $(\chi, B)$-candidate.
The following lemma precisely states the sufficient and necessary conditions for a $b$-coloring extending $\chi$ to exist.  

\begin{lemma}\label{lem:find-coloring}
Let $\sigma = (\chi_S, q, \tau, \lambda, \xi)$ be a signature and $\chi\colon V(G) \rightarrow [p]$ be a minimal $\sigma$-coloring. There is a $b$-coloring $\psi\colon V(G) \rightarrow [k]$ of $G$ with signature $\sigma$ such that $\chi \seq \psi$ if and only if $\chi$ is proper and there is a $\chi$-candidate subset $B$ and a set $\ca C\seq\ca U$ such that $|\ca C| = k-p$, all sets in $\ca C$ are $(\chi, B)$-candidates, and all sets in $\ca U\setminus\ca C$ are $\chi$-flexible. Moreover, given $B$ and $\ca C$ satisfying these properties, the $b$-coloring $\psi$ can be computed in polynomial time.
\end{lemma}
\begin{proof}
First, we prove the left-to-right implication. Let $\psi\colon V(G) \rightarrow [k]$ be a $b$-coloring of $G$ with signature $\sigma$ such that $\chi \seq \psi$. Since $\psi$ is proper, $\chi$ is proper as well.
Let $u_c \in V(G)$ be an arbitrary $b$-vertex in $\psi$ for some color $c \in [p]$. By Lemma~\ref{lem:twins}, $u_c$ has a twin $u_c'$ such that $\chi(u_c') = c$. Clearly, $u_c'$ is also a $b$-vertex in $\psi$.
Let $B = \{u_c' \sep c\in [p]\}$.
To show that $B$ is a $\chi$-candidate subset, we only need to show that $[p] \seq \chi(N[u])$ for each $u \in B$. Let $c \in [p]$ be such that $\chi(u) \ne c$, and let $v \in N(u)$ be such that $\psi(v) = c$. By Lemma~\ref{lem:twins}, $v$ has a twin $v'$ such that $\chi(v') = c$. In particular, $uv' \in E(G)$. Hence, we indeed have $[p] \seq \chi(N[u])$.

Let $\ca C = \{U_c \sep c\in [p+1,k]\}$, where $U_c \in \ca U$ is the set such that $\psi^{-1}(c) \seq U_c$; it is unique by Observation~\ref{obs:anonym}.
Recall that the observation also says that $U_c\ne U_d$ for distinct colors $c,d\in[p+1,k]$, which easily implies that $|\ca C| = k-p$. Now let $U_c \in \ca C$. We need to prove that $U_c$ is a $(\chi, B)$-candidate. First, let $u \in B$. Since $u$ is a $b$-vertex in $\psi$ by construction of $B$, there is a vertex $v \in N(u)$ such that $\psi(v) = c$. By definition of $U_c$, we know that $v \in U_c$. Moreover, $v$ is uncolored by $\chi$ because the codomain of $\chi$ is $[p]$. Second, let $u \in U_c$ be a $b$-vertex for $c$ in $\psi$, let $d \in [p]$, and let $v \in N(u)$ be a vertex such that $\psi(v) = d$. By Lemma~\ref{lem:twins}, $v$ has a twin $v'$ such that $\chi(v') = d$. In particular, $uv' \in E(G)$. Hence, $[p] \seq \chi(N(u))$, and $U_c$ is indeed a $(\chi, B)$-candidate.

Now let $U \in \ca U\setminus\ca C$.
By definition of $\ca C$, we know that $\psi(U)\seq[p]$.
If $u \in U$ is a vertex uncolored by $\chi$, then by Lemma~\ref{lem:twins}, $u$ has a twin $u'$ that is colored by $\chi$, which shows that $U$ is $\chi$-flexible.

\smallskip
Now we prove the right-to-left implication. Let $B\seq V(G)$ and $\ca C\seq\ca U$ be as in the statement. We may assume that $\ca C = \{U_c \sep c\in [p+1,k]\}$, i.e., we label the sets in $\ca C$ arbitrarily. Now we construct $\psi \supseteq \chi$ as follows. Let $u \in V(G)$ be a vertex uncolored by $\chi$. By Definition~\ref{def:represent}, $u\notin S$. If $u\in U_c \in\ca C$, we let $\psi(u) = c$, and if $u\in U \in\ca U \setminus\ca C$, then by Definition~\ref{def:candidates-and-flexible}, $u$ has a twin $v\in U$ that is colored by $\chi$, and we let $\psi(u) = \chi(v)$.
Observe that this coloring $\psi$ can be computed in polynomial time as required.
Suppose there is an edge $uv \in E(G)$ such that $\psi(u) = \psi(v) = c$.
If $c \in [p]$, then by Lemma~\ref{lem:twins}, there would be an edge $u'v' \in E(G)$ such that $\chi(u') = \chi(v') = c$, which would be a contradiction with $\chi$ being proper.
On the other hand, if $c\in[p+1, k]$, then $u,v \in U_{c}$ by construction of $\psi$, a contradiction with $uv \in E(G)$. Hence, $\psi$ is a proper coloring.

Now we show that $\psi$ is a $b$-coloring. First, let $c \in [p]$ and let $u \in B$ be the unique vertex in $B$ such that $\psi(u) = c$. Let $d \in [k]$ be a color. If $d \in [p]$, then by Definition~\ref{def:candidates-and-flexible}, $d \in \chi(N[u]) \seq \psi(N[u])$. On the other hand, if $d \in [p+1, k]$, then $U_d \in \ca C$ is a $(\chi, B)$-candidate, which means that there is $v \in U_d \cap N(u)$ that is uncolored by $\chi$. However, by construction of $\psi$, we know that $\psi(v) = d$. Hence, $u$ is indeed a $b$-vertex for $c$ in $\psi$.
Second, let $c\in [p+1,k]$. Since $U_c$ is a $(\chi, B)$-candidate, there is a vertex $u \in U_c$ such that $[p] \seq \chi(N(u)) \seq \psi(N(u))$. If $d \in [p+1,k]$ and $d \ne c$, then $U_d \ne U_c$, and so $u$ has a neighbor of color $d$ in $\psi$. Hence, $u$ is a $b$-vertex for $c$ in $\psi$.
\end{proof}

\subsection{Finishing the proof}\label{sub:dcc-finish}

Using Lemma~\ref{lem:find-coloring}, it is easy to design the desired algorithm.

\begin{proof}[of Theorem~\ref{thm:co-cluster}]
We need to design an algorithm that given an integer $k$ and a graph $G$ with distance to co-cluster $p_0$ decides whether there is a $k$-$b$-coloring of $G$. Let $p = \min(p_0, k)$.
First, we apply Theorem~\ref{thm:find-S} to the complement of $G$ to compute a minimum co-cluster-modulator $S$ of $G$. Second, we try all signatures $\sigma$ (recall that the definition of a signature depends on $S$, see Definition~\ref{def:signature}), and we find a minimal $\sigma$-coloring $\chi$ of $G$ using Lemma~\ref{lem:minimal-sigma-exists}.
If $\chi$ is not proper, we  discard it and continue with the next signature. Otherwise, we try all $\chi$-candidate subsets $B$, and we compute for each set $U\in \ca U$, whether it is $\chi$-flexible or a $(\chi, B)$-candidate.
If there is a set $U\in\ca U$ that is neither $\chi$-flexible nor a $(\chi, B)$-candidate, we reject $B$. Otherwise, we compute the sets $\ca C_0 = \{U\in\ca U\sep U$ is a $(\chi, B)$-candidate but not $\chi$-flexible$\}$ and $\ca C_1 = \{U\in\ca U\sep U$ is a $\chi$-flexible $(\chi, B)$-candidate$\}$. 
If $|\ca C_0| > k-p$ or $|\ca C_0| + |\ca C_1| < k-p$, we reject $B$. Otherwise, we can find a set $\ca C$ such that $|\ca C| = k-p$ and $\ca C_0 \seq\ca C\seq\ca C_0\cup\ca C_1$, which allows us to compute a $b$-coloring of $G$ by Lemma~\ref{lem:find-coloring}. If we do not find a $b$-coloring for any $\sigma$ and $B$, we reject $(G, k)$.

Let us prove that this algorithm is correct.
By Lemma~\ref{lem:find-coloring}, if we output a coloring, then it is a $k$-$b$-coloring of $G$.
In the other direction, suppose that $(G, k)$ is a YES instance. If $\psi$ is a $k$-$b$-coloring of $G$, then by Lemma~\ref{lem:b-coloring-has-signature}, there is a signature $\sigma$ such that $\psi$ is a $\sigma$-coloring. By Lemma~\ref{lem:any-minimal-coloring-can-be-used}, any minimal $\sigma$-coloring $\chi$ can be extended into a $b$-coloring $\chi'$ of $G$. 
By Lemma~\ref{lem:find-coloring}, the existence of $\chi'$ implies the existence of $B$ and $\ca C$ with the required properties.
Moreover, the existence of $\ca C$ implies that all sets $U\in\ca U$ are $\chi$-flexible or $(\chi, B)$-candidates, and $|\ca C| = k-p$ implies $|\ca C_0| \le k-p$ and $|\ca C_0| + |\ca C_1| \ge k-p$.
Therefore, the algorithm finds $B$ and $\ca C$, and based on them, it finds a $b$-coloring of $G$.

Finally, observe that the desired running time is achieved by definition of $p$, Theorem~\ref{thm:find-S}, Observation~\ref{obs:number-of-signatures}, Lemma~\ref{lem:minimal-sigma-exists}, and Lemma~\ref{lem:find-coloring}.
In particular, it is safe to try each $\chi$-candidate subset $B$ because the number of vertices colored by a minimal $\sigma$-coloring is $2^{\ca O(p)}$ by Lemma~\ref{lem:minimal-sigma-exists}.
\end{proof}

\section{Feedback Edge Number}\label{sec:fen}

In this section, we prove the following theorem. For a high-level overview of its proof, see the end of Section~\ref{sec:intro}.

\begin{theorem}\label{thm:fen}
The \textsc{$b$-Coloring} problem can be solved in time $2^{\ca O(p^2)} \cdot n^{\ca O(1)}$, where $p$ is the feedback edge number of the input $n$-vertex graph.
\end{theorem}

In Section~\ref{subsec:setup}, we define key concepts, such as the fen-core $S$ and the $S$-profile $(\chi, B)$, and we prove a few simple lemmas.
In Section~\ref{subsec:plans}, we define color plans and color realizations, and describe their properties.
A very technical part of the proof is making the color realization pivot-free; this is done in Section~\ref{sub:getting-rid}.

The purpose of Sections~\ref{subsec:plans} and~\ref{sub:getting-rid} is to find a feasible color realization $\rho$ (it exists if and only if the desired $b$-coloring exists).
This realization $\rho$ chooses vertices that will become $b$-vertices.
In Section~\ref{sub:coloring-neighbors}, we color the neighbors of the vertices chosen by $\rho$, and obtain a partial $b$-coloring.
Finally, in Section~\ref{sub:fen-finish}, we put everything together, and prove Theorem~\ref{thm:fen}.

\subsection{Initial setup}\label{subsec:setup}

Let us fix an instance $(G, k)$ of the \textsc{$b$-Coloring} problem, and let $p_G$ be the feedback edge number of $G$.
We will assume that $k \ge 96p_G + 18$ since otherwise we could simply use Proposition~\ref{prop:tree-width}.

We begin by defining a small vertex set $S$, which will play a central role in the algorithm, see Figure~\ref{fig:fen-core}.
To handle the connected components of $G$ which are trees, we also define a superset $S^+$ of $S$.

\begin{definition}\label{def:fen-core}
Let $S \seq V(G)$ be a set.
An \emph{$S$-outer path} $P$ is a path in $G$ such that $V(P) \cap S = \emptyset$ and $V(P) \cap N(S) = \{u, v\}$, where $u$ and $v$ are the endpoints of $P$.
We say that $S$ is a \emph{fen-core} if each cycle in $G$ has at least one edge in $G[S]$, there are at most two $(u, S)$-paths for each $u \in \overline{S}$, each $S$-outer path has length at least 7, and $|S| \le 32p_G$.
An \emph{extension} of $S$ is a minimal set $S^+ \supseteq S$ such that each connected component of $G$ contains a vertex of $S^+$.
\end{definition}

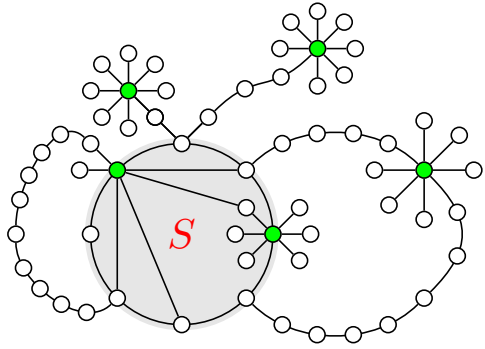
\begin{figure}[t]
\scalebox{1.5}{
\begin{tikzpicture}
\tikzset{
ff/.style = {draw,black, circle, fill=white, minimum width=4pt, inner sep=0pt}}    
\tikzset{st4/.style={postaction={decorate,
decoration={markings,
mark= at position 1/5 with {\node [ff]{};},
mark= at position 2/5 with {\node [ff]{};},
mark= at position 3/5 with {\node [ff]{};}, 
mark= at position 4/5 with {\node [ff]{};}, 
}}}}
\tikzset{st2/.style={postaction={decorate,
decoration={markings,
mark= at position 1/3 with {\node [ff]{};},
mark= at position 2/3 with {\node [ff]{};},
}}}}
\definecolor{Hcolor}{RGB}{182,255,193}
\tikzmath{ \cir = 0.8;}
\fill[gray!20!white] (0,0) circle (0.85 cm);
\draw (0,0) circle (\cir cm);
\def\n{8}
\node[red] at (0,0) {\large $S$};
\begin{scope}[every node/.style={draw, circle, minimum width=4pt, inner sep=0pt}]
\node[fill=white] (a1) at (360/\n*1: \cir cm) {};
\node[fill=white] (a2) at (360/\n*2: \cir cm) {};
\node[fill=green] (a3) at (360/\n*3: \cir cm) {};
\node[fill=white] (a4) at (360/\n*4: \cir cm) {};
\node[fill=white] (a5) at (360/\n*5: \cir cm) {};
\node[fill=white] (a6) at (360/\n*6: \cir cm) {};
\node[fill=white] (a7) at (360/\n*7: \cir cm) {};
\node[fill=green] (a8) at (360/\n*8: \cir cm) {};
\node[above left = 5pt of a2] (a21) {};
\node[above right = 5pt of a2] (a22) {};
\node[fill=green, above left = 5pt of a21] (a211) {};
\node[above right= 5pt of a211] (e2) {};
\draw (a211)--(e2);
\node[below right= 5pt of a211] (e3) {};
\draw (a211)--(e3);
\node[fill=white, left= 5pt of a211] (e4) {};
\draw (a211)--(e4);
\node[fill=white, above left= 5pt of a211] (e5) {};
\draw (a211)--(e5);
\node[fill=white, below left= 5pt of a211] (e6) {};
\draw (a211)--(e6);
\node[above= 5pt of a211] (e7) {};
\draw (a211)--(e7);
\node[below= 5pt of a211] (e8) {};
\draw (a211)--(e8);
\node[right= 5pt of a211] (e9) {};
\draw (a211)--(e9);
\draw (a21)--(a2)--(a22);    
\draw (a211)--(a21);   
\node[fill=green, right = 40pt of a1] (b1) {};
\node[right = 40pt of a7] (c1) {};      
\draw[st4] (a1) to[out=45,in=135, distance=15pt] (b1);
\draw[st2] (b1) to[out=-45,in=45, distance=15pt] (c1);
\draw[st4] (a7) to[out=-45,in=-135, distance=15pt] (c1);
\node[right= 5pt of a8] (d1) {};
\draw (a8)--(d1);
\node[above right= 5pt of a8] (d2) {};
\draw (a8)--(d2);
\node[below right= 5pt of a8] (d3) {};
\draw (a8)--(d3);
\node[fill=white, left= 5pt of a8] (d4) {};
\draw (a8)--(d4);
\node[fill=white, above left= 5pt of a8] (d5) {};
\draw (a8)--(d5);
\node[fill=white, below left= 5pt of a8] (d6) {};
\draw (a8)--(d6);
\node[right = 8pt of b1] (b3) {};
\draw (b1)--(b3);    
\node[above right = 8pt of b1] (b4) {};
\draw (b1)--(b4);     
\node[below = 8pt of b1] (b51) {};
\draw (b1)--(b51);
\node[left = 8pt of b1] (b31) {};
\draw (b1)--(b31);    
\node[below left = 8pt of b1] (b41) {};
\draw (b1)--(b41);     
\node[above = 8pt of b1] (b5) {};
\draw (b1)--(b5);
\node[fill=green, above right=20pt of a22, xshift=10pt] (b2) {};
\draw[st2] (a22) to[out =45, in = -135, distance=15] (b2);
\node[above right= 5pt of b2] (f2) {};
\draw (b2)--(f2);
\node[below right= 5pt of b2] (f3) {};
\draw (b2)--(f3);
\node[fill=white, above left= 5pt of b2] (f5) {};
\draw (b2)--(f5);
\node[fill=white, left= 5pt of b2] (f6) {};
\draw (b2)--(f6);
\node[above= 5pt of b2] (f7) {};
\draw (b2)--(f7);
\node[below= 5pt of b2] (f8) {};
\draw (b2)--(f8);
\node[right= 5pt of b2] (f9) {};
\draw (b2)--(f9);
\draw (a21)--(a2)--(a22);
\draw (a5)--(a3)--(a1);
\draw (a6)--(a3)--(d5);
\node[above left = 5pt of a3] (g) {};
\node[left = 5pt of a3] (g2) {};
\draw (g)--(a3)--(g2);
\node[left = 0.5cm of a4] (h) {};
\draw[st4] (g) to[out=135, in=90, distance = 10pt] (h);
\draw[st4] (a5) to[out=-135, in=-90, distance = 10pt] (h);
\end{scope}
\end{tikzpicture}}
\centering
\caption{A depiction of the fen-core $S$, see Definition~\ref{def:fen-core}. If $k \le 9$, then the green vertices may be used as $b$-vertices. Note that in contrast to this figure, $G[S]$ may be disconnected.}
\label{fig:fen-core}
\end{figure}

Observe that if $S$ is a fen-core, then no vertex in $\overline{S}$ has two neighbors in $S$ because it would induce an $S$-outer path of length 0.

Now we show that $S$ can be efficiently computed using Lemma~\ref{lem:fen-structure}.
Informally, if an $S$-outer path is too short, then we add it to $S$.

\begin{lemma}\label{lem:S-exists}
A fen-core exists and it can be found in polynomial time.
\end{lemma}
\begin{proof}
Let $D\seq V(G)$ be the set of vertices contained in dangling trees in $G$, and
let $\ca C$ be the set of connected components of $G - D$. For each $H \in \ca C$, let $p_H$ be the feedback edge number of $H$, and
let $S_H \seq V(H)$ and $\ca P_H$ be the sets obtained by applying Lemma~\ref{lem:fen-structure} on $H$.
Since $\sum_{H \in \ca C} p_H = p_G$, we know that the sets $S_G = \bigcup_{H \in \ca C} S_H$ and $\ca P_G= \bigcup_{H \in \ca C} \ca P_H$ satisfy $|\ca P_G|, |S_G| \le 4p_G$. By Lemma~\ref{lem:fen-structure}, $\ca P_G$ is the set of connected components of $G - (D \cup S_G)$, and each $P \in\ca P_G$ is a dangling path in $G - D$.

Let $\ca P_{short} = \{P \in \ca P_G\colon |V(P)|\le 7\}$, let $\ca P=\ca P_G\setminus\ca P_{short}$, and let $S = S_G\cup V(\bigcup \ca P_{short})$, where $\bigcup \ca P_{short}$ is the graph obtained by taking the union of all paths in $\ca P_{short}$.
Observe that $|S| \le (4 + 4\cdot 7)p_G = 32p_G$.
Since each path $P \in \ca P$ is an $S$-outer path, there are exactly two $(u, S)$-paths for each $u \in V(\bigcup \ca P)$.
Moreover, if $u \notin S \cup V(\bigcup \ca P)$, then $u \in D$, and there is at most one $(u, S)$-path.
Now it can be easily observed that $S$ is a fen-core and that the described construction can be performed in polynomial time.
\end{proof}

Let us fix a fen-core $S$, and let $p = |S|$.
From now on, we may view $p$ as the parameter of our algorithm because $p \in \ca O(p_G)$.
Note that $k \ge 3p + 18$.

Since $S$ is fixed, we will call $S$-outer paths simply \emph{outer paths}.
The following two observations follow directly from Definition~\ref{def:fen-core}.

\begin{observation}\label{obs:short-cycles-are-in-S}
If $C$ is a cycle in $G$ with at most eight vertices, then $V(C)\seq S$.
\end{observation}
\begin{proof}
Let $C$ be a cycle in $G$ with at most eight vertices.
If $V(C)\cap S = \emptyset$, then $C$ is a cycle with no edge in $G[S]$, which contradicts Definition~\ref{def:fen-core}.
If $\emptyset \ne V(C)\cap S \ne V(C)$, then some subgraph of $C$ is an outer path of length at most 6, which again contradicts Definition~\ref{def:fen-core}. Hence, we have $V(C) \seq S$.
\end{proof}

\begin{observation}\label{obs:S+}
For each vertex $u \in \overline{S}$, there are at most two distinct $u$-$S^+$ paths.
\end{observation}
\begin{proof}
Let $C$ be the connected component of $G$ containing $u$.
If $C \cap S \ne \emptyset$, then each $u$-$S^+$ path is a $u$-$S$ path, and the statement follows by Definition~\ref{def:fen-core}.
Otherwise, $C$ is a tree, and there is exactly one $u$-$S^+$ path by definition of $S^+$.
\end{proof}

\subsubsection{General definition of pivots}\label{subsub:general-pivot}

As discussed in Section~\ref{sec:intro}, a key role in our algorithm will be played by vertices called \emph{pivots}.
Later, the definition of a pivot will depend on some partial coloring of $G$.
However, it will be useful to first define pivots generally, without referring to any coloring.

\begin{definition}\label{def:pivot}
Let $D \seq V(G)$ be a set. We say that a vertex $u \in \overline{S \cup D}$ is a \emph{$D$-pivot} if for each $v \in D$, either $uv \in E(G)$ or there is a vertex $w \in N(u) \cap N(v) \cap D$.
The vertices in $N(u) \cap N(D) \cap D$ will be called \emph{$(u, D)$-links}.
A vertex $v \in D$ is called \emph{$S$-influenced} if $N[v] \cap N[S] \nsubseteq \{u\}$, i.e., there is a $v$-$S$ path of length at most 2 not containing $u$.
\end{definition}

We will sometimes encounter a situation in which a vertex of $D$ has a neighbor in $S \setminus D$ or in $N(S) \setminus (D \cup \{u\})$; this motivates our definition of being $S$-influenced.
Observe that the definition depends on $u$; when we say that a vertex is $S$-influenced, the pivot $u$ will always be clear from context.

The following lemma describes how a ``pivoted'' set $D$ can interact with $S$, and its proof follows easily from Definitions~\ref{def:fen-core} and~\ref{def:pivot}.

\begin{lemma}\label{lem:pivot-properties}
Let $D \seq V(G)$ be a set, let $u \in \overline{S \cup D}$ be a $D$-pivot, and let $D^+ = D \cup \{u\}$.
\begin{enumerate}[(a)]

\item If $C$ is a connected component of $G[D]$, then there is a vertex $v_C \in C$ such that $N(u) \cap C = \{v_C\}$ and $C \seq N[v_C]$.
\label{aux:component-of-GD}

\item If $v \in D \setminus N(u)$ and $w \in D \cap N(v)$, then either $w \in N(u)$ or $v, w \in S$.\label{aux:unique-neighbor} 

\item There is a connected component $C$ of $G[D]$ such that all $S$-influenced vertices are in $C$.\label{aux:contains-all-S-influenced}

\item If $v_1, v_2 \in D \cap S$ are distinct vertices, then there is a $(u, D)$-link $w \in S \cap N[v_1] \cap N[v_2]$.\label{aux:link-in-S}

\item Let $v_1, v_2 \in D$ be distinct vertices and let $w \in N(v_1) \cap N(v_2)$ be such that $\{v_1, v_2, w\} \nsubseteq S$. If $w \in \overline{D}$ or $v_1 \in N(u)$, then $w = u$.
\label{aux:equals-pivot}

\item If $v_1 \in D \setminus S$ is such that $N(v_1) \cap \overline{D^+} \cap N[S] \ne \emptyset$ and $v_2 \in D \setminus \{v_1\}$, then $N[S] \cap N[v_2] \seq \{v_1\}$. In particular, $S \cap D = \emptyset$, all $S$-influenced vertices are in $N[v_1]$, and if $v_1 \notin N(S)$, then $v_1$ is the only $S$-influenced vertex.\label{aux:outer-S}
\end{enumerate}

\end{lemma}
\begin{proof}
For~(\ref{aux:component-of-GD}), let $C$ be a component of $G[D]$.
By Definition~\ref{def:pivot}, for each $v \in C$, there is a vertex $w_v$ such that $w_v \in N(u) \cap N[v] \cap D$.
Observe that if $v \ne w_v$, then $uv \notin E(G)$; otherwise $G[\{u, v, w_v\}]$ would be a cycle containing $u \notin S$, contradicting Observation~\ref{obs:short-cycles-are-in-S}.
Suppose that there are vertices $v, x \in C$ such that $w_v \ne w_x$; we may choose $v$ and $x$ so that $vx \in E(G)$.
Now $G[\{u, w_v, v, x, w_x\}]$ is a cycle, which again contradicts Observation~\ref{obs:short-cycles-are-in-S}.
Hence, we may define $v_C$ to be $w_v$ for any $v \in C$.

For~(\ref{aux:unique-neighbor}), let $v, w \in D$ be two adjacent vertices, and suppose that $v, w \notin N(u)$.
By~(\ref{aux:component-of-GD}), there is a vertex $x \in N(u) \cap N(v) \cap N(w)$.
Since $G[\{v, w, x\}]$ is a cycle, we obtain $v, w \in S$ by Observation~\ref{obs:short-cycles-are-in-S}, which concludes the proof.

For~(\ref{aux:contains-all-S-influenced}), suppose for contradiction that there are two $S$-influenced vertices $v_1, v_2 \in D$ in different components of $G[D]$.
For $i \in [2]$, if $v_i \in N[S]$, then we set $w_i := v_i$, and otherwise we define $w_i$ to be a vertex in $N(S) \cap N(v_i)$ such that $w_i \ne u$; note that $w_i$ exists by definition of being $S$-influenced.
Let $P$ be a $v_1$-$v_2$ path in $G[D^+]$.
Since $v_1$ and $v_2$ are in different components of $G[D]$, we have $u \in V(P)$.
Let $P'$ be the path obtained from $P$ by adding the edge $v_1w_1$ if $w_1 \notin V(P)$ and the edge $v_2w_2$ if $w_2 \notin V(P)$.
Now observe that a subpath of $P'$ containing $u$ is an outer path of length at most 6, which contradicts Definition~\ref{def:fen-core}.

For~(\ref{aux:link-in-S}), let $v_1$ and $v_2$ be as in the statement.
By~(\ref{aux:contains-all-S-influenced}), there is a connected component $C$ of $G[D]$ such that $v_1, v_2 \in C$.
By~(\ref{aux:component-of-GD}), there is a vertex $w \in C \cap N(u)$ such that $v_1, v_2 \in N[w]$.
Without loss of generality, $v_1 \ne w$, which implies that $w$ is a $(u, D)$-link.
If $w \notin S$, then $v_1, v_2 \in S \cap N(w)$ and $G[\{w\}]$ is an outer path, which contradicts Definition~\ref{def:fen-core}.

For~(\ref{aux:equals-pivot}), let $v_1, v_2$, and $w$ be as in the statement.
If $w \in D$ and $v_1 \in N(u)$, then $v_1, v_2$, and $w$ are in the same connected component of $G[D]$.
By~(\ref{aux:component-of-GD}), $v_1v_2 \in E(G)$, which is a contradiction with Observation~\ref{obs:short-cycles-are-in-S} since $G[\{w, v_1, v_2\}]$ is a cycle.
Suppose that $w \notin D^+$ and let $P$ be the $v_1$-$v_2$ path in $G[D^+]$.
Now $G[V(P) \cup w]$ is a cycle of length at most 6, which is again a contradiction.
In both cases, we used $\{v_1, v_2, w\} \nsubseteq S$.

For~(\ref{aux:outer-S}), let $v_1$ and $v_2$ be as in the statement.
Let $w_1 \in N(v_1) \cap \overline{D^+} \cap N[S]$, and suppose for contradiction that there is a vertex $w_2 \in N[v_2] \cap N[S]$ such that $w_2 \ne v_1$.
Let $P$ be a $v_1$-$v_2$ path in $G[D^+]$, and let $P'$ be the path obtained from $P$ by adding the edge $v_1w_1$ and, if $w_2 \notin V(P)$, also the edge $v_2w_2$.
Now a subpath of $P'$ is an outer path of length at most 6, which contradicts Definition~\ref{def:fen-core}.
Hence, $N[S] \cap N[v_2] \seq \{v_1\}$ as desired.
The second sentence of~(\ref{aux:outer-S}) follows easily.
\end{proof}

\subsubsection{$S$-profiles}

\newcommand{\red}{\textsf{red}}

Now we present a key definition.
Informally, a vertex is a $\chi$-candidate if it may become a $b$-vertex in a coloring of $G$ that extends a partial coloring $\chi$.
We also define an \emph{$S$-profile}, which consists of a coloring of $S$ and a set $B \seq S$ of vertices that should become $b$-vertices. 

\begin{definition}\label{def:S-profile}
Let $U \seq V(G)$ and let $\chi$ be a proper coloring of $G[U]$.
The \emph{$\chi$-redundancy} of a vertex $v \in V(G)$ is the integer $\red_\chi(v) = |N[v]\setminus U| + |\chi(N[v])| - k$.
If $\red_\chi(u) \ge 0$, we say that $u$ is a \emph{$\chi$-candidate.}
If $\red_\chi(u) = 0$, we say that $u$ is \emph{$\chi$-tight.}

Let $\chi\colon S \rightarrow [p]$ be a proper coloring of $G[S]$, and let $B \seq S$ be a set such that each vertex in $B$ is a $\chi$-candidate and $\chi(B) = [b]$, where $b = |B|$.
We say that $(\chi, B)$ is an \emph{$S$-profile}.
We call a $k$-$b$-coloring $\psi$ of $G$ \emph{described} by $(\chi, B)$ if 
$\chi = \psi \upharpoonright S$, all vertices in $B$ are $b$-vertices in $\psi$, and no color in $[b+1, p]$ has a $b$-vertex in $S$. 
\end{definition}

Let us fix an $S$-profile $(\chi, B)$. Recall that the next step of the algorithm after computing $S$ is trying all $S$-profiles, see Figure~\ref{fig:overview}.

Let $\kout = \{u \in \overline{S} \sep \red_\chi(u) \ge 0\}$ and $\kb = \kout \cup B$.
The reader should think of $\kb$ as the set of our candidates for $b$-vertices.
However, $\kb$ is \emph{not} the set of all $\chi$-candidates because we do not care about vertices in $S \setminus B$.
Indeed, when such a vertex is a  $b$-vertex for some color $c \in [p]$, then $c \le b$ and $c$ has another $b$-vertex in $B$.
The following observation formalizes this idea.

\begin{observation}\label{obs:b-vertices-are-candidates}
If $\psi$ is a $k$-$b$-coloring of $G$ described by $(\chi, B)$, then each color $c \in [k]$ has a $b$-vertex $u\in \kb$ in $\psi$.
Moreover, if $|\kb| = k$, then all vertices in $\kb$ are $b$-vertices in~$\psi$.
\end{observation}
\begin{proof}
Let $c \in [k]$ be a color and let $u \in V(G)$ be a $b$-vertex for $c$ in $\psi$. If $u \in S \setminus B$, then by Definition~\ref{def:S-profile}, $c \in [b]$ and there is another $b$-vertex for $c$ in $\psi$ that is in $B$. Hence, we may assume $u \notin S \setminus B$. By definition of a $b$-vertex, we know that $|\psi(N[u])| = k$. Since $|N[u]\setminus S| + |\chi(N[u])| \ge |\psi(N[u])|$, we  deduce that the $\chi$-redundancy of $u$ is non-negative, which implies that $u \in \kb$. The second part of the statement follows by the pigeonhole principle.
\end{proof}

Now we define a \emph{link}, which is a vertex that cannot become a $b$-vertex if the two linked vertices have the same color. 
We also specify Definition~\ref{def:pivot} and present a more concrete definition of a pivot.

\begin{definition}\label{def:links}
Let $\psi$ be a partial coloring of $G$ such that $\chi \seq \psi$.
If $u,v \in V(G)$ are distinct vertices, then we say that $u$ and $v$ are \emph{$\psi$-linked} if there is a $\psi$-tight vertex $w \in N(u) \cap N(v) \cap K^+$.
We also say that $u$ and $v$ are \emph{$\psi$-linked via $w$}
and that $w$ is a $(\psi, u, v)$-link.

Let $D \seq V(G)$ be a set.
We say that $u \in \overline{D \cup S}$ is an $(\psi, D)$-\emph{pivot} if $u$ is a $D$-pivot and each $(u, D)$-link is a $(\psi, u, v)$-link for some $v \in D$.
\end{definition}

Let us remark that there can be at most one $(\psi, u, v)$-link unless $u, v \in S$.
However, the definition will not be used for $u, v \in S$, so we may speak of \emph{the} $(\psi, u, v)$-link.

Now we observe that the $(\chi, u, v)$-link cannot be a $b$-vertex in a coloring described by $(\chi, B)$ if $u$ and $v$ have the same color.

\begin{observation}\label{obs:links}
If $\psi$ is a $k$-$b$-coloring of $G$ described by $(\chi, B)$, $u \in V(G)$, $v \in \overline{S}$, and $w \in \kb$ is the $(\chi, u, v)$-link and also a $b$-vertex in $\psi$, then $\psi(u) \ne \psi(v)$.
\end{observation}
\begin{proof}
Let $u, v, w$, and $\psi$ be as in the statement.
Let $f\colon [k] \rightarrow \chi(N[w]) \cup (N[w]\setminus S)$ be a function such that for $c \in \chi(N[w])$, we have $f(c) = c$, and otherwise $f(c)$ is a vertex of $N[w]$ colored with $c$ by $\psi$.
Since $w$ is a $b$-vertex in $\psi$, $f$ is indeed a well-defined function.
By Definition~\ref{def:links}, $w$ is $\chi$-tight, which means that $|N[w]\setminus S| + |\chi(N[w])| = k$ and that $f$ is a bijection.
Suppose for contradiction that $c := \psi(u) = \psi(v)$.
If $u \in S$, then $f(c) = c$, and otherwise $f(c) = u$ since $f$ is a bijection.
Both cases lead to a contradiction since $f(c) = v$ (and $u \ne v$ by Definition~\ref{def:links}).  
\end{proof}

\subsection{Color plans and realizations}\label{subsec:plans}

As in the previous section, let $S$ be a fen-core, $(\chi, B)$ be an $S$-profile, $\kout = \{u \in \overline{S} \sep \red_\chi(u) \ge 0\}$, and $\kb = \kout \cup B$.
Recall that $\chi(S) \seq [p]$ and $\chi(B) = [b]$.
In addition, we define $\kstar = \kout \setminus N(S)$.
Let us also fix an extension $S^+$ of $S$, see Definition~\ref{def:fen-core}.

Now we define a \emph{color plan}, which is a function deciding, for each color $c \in [b+1, p]$, whether the $b$-vertex for $c$ should be in $N(S)$, and if yes, which vertex of $S$ it should be adjacent to (it may be adjacent to at most one vertex of $S$ by Definition~\ref{def:fen-core}).
The $b$-vertices that should not be in $N(S)$ must be in $K^*$, and the color plan assigns $*$ to their colors.
Only some color plans may be used to produce a $b$-coloring: we call such plans \emph{valid}.

\begin{definition}\label{def:color-plan}
A function $\pi\colon [b+1, p]\rightarrow S \sqcup \{*\}$ is called a \emph{color plan} of $(\chi, B)$ (or simply a color plan when $(\chi, B)$ is clear from context).
If $c \in [b+1, p]$ is a color, then $\pi$ is \emph{$c$-critical} if $\pi(c) = *$ and each vertex in $\kstar$ is $\chi$-linked to a vertex in $\chi^{-1}(c)$.
We say that $\pi$ is \emph{critical} if it is $c$-critical for some $c \in [b+1, p]$.

We say that $\pi$ is \emph{valid} if it satisfies the following properties:
\begin{enumerate}
\item $|\pi^{-1}(*)| \le |\kstar|$;\label{cp:star}
\item for each $u \in S$, we have $|\pi^{-1}(u)| \le |N(u) \cap \kout|$;\label{cp:enough-neighbor}
\item for each $u \in S$, if $\chi(u) \in [b+1, p]$, then $\pi(\chi(u)) \ne u$;\label{cp:proper}
\item for each $u \in B$, we have $|\pi^{-1}(u) \cap \chi(N(u))| \le \red_\chi(u)$;\label{cp:redundancy}
\item if $\pi$ is critical, then $|\kb| > k$ and there is a vertex $u \in S$, called a \emph{$\pi$-anchor}, such that $|\pi^{-1}(u)| < |N(u) \cap \kout|$ and there is a vertex in $\kstar$ $\chi$-linked to $u$.\label{cp:block}
\end{enumerate}
\end{definition}

Let us present an intuition behind the five properties of a valid color plan.
First, Properties~\ref{cp:star} and~\ref{cp:enough-neighbor} ensure that there are enough vertices for the plan to be realized.
Second, a violation of Property~\ref{cp:proper} would mean that the obtained coloring cannot be proper.
Third, Property~\ref{cp:redundancy} ensures that each vertex in $B$ can still become a $b$-vertex.
Finally, if $\pi$ is $c$-critical, then the $b$-vertex for $c$ must be linked to a vertex in $S$ colored with $c$, which means that the link cannot be a $b$-vertex.
Property~\ref{cp:block} ensures that there exists such a link that does not have to be a $b$-vertex, which will enable us to find a $b$-vertex for $c$.

Now we define a \emph{color realization} $\rho$, which is a function selecting concrete vertices as $b$-vertices for colors in $[b+1, k]$, based on some color plan.
We also define two properties of $\rho$: \emph{damage-freeness}, which means that the chosen vertices do not interfere with each other, and a technical property called \emph{safeness}, which will enable us to color the entire graph at the very last step of the algorithm, see Figure~\ref{fig:overview}.

\begin{definition}\label{def:color-realization}
An injective function $\rho \colon [b+1, k]\rightarrow \kout$ is called a \emph{color realization}.
Let $B_\rho := B\cup \range(\rho)$, $S_\rho := S \cup B_\rho$, and $\chi_\rho := \chi \cup \rho^{-1}$; note that $\chi_\rho$ is a $k$-coloring of $G[S_\rho]$.
We say that $\rho$ \emph{realizes} a color plan $\pi$ (or is a color realization of $\pi$) if for each $c \in [b+1, p]$:
\begin{enumerate}[(a)]
\item $\pi(c) = *$ if and only if $\rho(c) \in \kstar$;\label{cr:star}
\item for each $u \in S$, it holds that $\pi(c) = u$ if and only if $\rho(c) \in N(u)$.\label{cr:S}
\end{enumerate}
A vertex $u$ is \emph{damaged} in $\rho$ if there is a color $c \in [b+1, p]$ and a vertex $v \in \chi^{-1}(c)$ such that $\rho(c) \in \kstar$ and $v$ and $\rho(c)$ are $\chi$-linked via $u$.
We say that $\rho$ is \emph{damage-free} if no vertex in $\range(\rho)$ is damaged in $\rho$.

A vertex $u \in V(G)$ is \emph{$\ell$-safe} in $\rho$ if $|\{c \in [p+1, k] \sep \dist(u,S^+) < \dist(\rho(c), S^+)\}| \le \ell$.
We say that $\rho$ is \emph{$\ell$-safe} if all vertices in $\kout \setminus B_\rho$ are $\ell$-safe in $\rho$.
We say that $\rho$ is \emph{almost $\ell$-safe} if all vertices in $\kout \setminus B_\rho$ are $\ell$-safe in $\rho$, possibly except for one vertex $v$ such that either $v \in N(B_\rho)$ or $v$ has at most $p+2$ neighbors in $N(B_\rho)$.
\end{definition}

Now we prove a simple lemma connecting some properties of a valid color plan with the damage-freeness of a color realization.

\begin{lemma}\label{lem:rho-to-pi}
If $\rho$ is a damage-free color realization of a color plan $\pi$, then $\pi$ satisfies Properties~\ref{cp:star}, \ref{cp:enough-neighbor}, and \ref{cp:block} of Definition~\ref{def:color-plan}.
\end{lemma}
\begin{proof}
First, observe that Property~\ref{cp:star} follows immediately from Condition~(\ref{cr:star}) of Definition~\ref{def:color-realization}, and Property~\ref{cp:enough-neighbor} follows from Condition~(\ref{cr:S}).
For Property~\ref{cp:block}, suppose that $\pi$ is $c$-critical for some $c\in [b+1, p]$, and let $u = \rho(c)$.
Since $\pi(c) = *$, we have $u \in \kstar$.
Hence, there is a vertex $v \in \chi^{-1}(c)$ such that $u$ is $\chi$-linked to $v$, via a vertex $w \in \kb$.
Since $u \in \kstar$, we have $w \notin S$, which implies $w \in \kout$.
Observe that $w$ is damaged in $\rho$, which means that $w \notin \range(\rho)$.
Since $w \in (N(v) \cap \kout) \setminus \range(\rho)$,
we obtain $|\pi^{-1}(v)| = |\rho^{-1}(N(v))| < |N(v) \cap \kout|$.
Moreover, $u$ is $\chi$-linked to $v$, which means that $v$ is a $\pi$-anchor.
Finally, observe that $|\kb| > k$ since $B_\rho \seq \kb$, $|B_\rho| = k$, and $w \in \kb \setminus B_\rho$.
Therefore, Property~\ref{cp:block} is satisfied by $\pi$.
\end{proof}

The rest of this section is organized as follows.
In Section~\ref{sub:are-candidates}, we prove that the vertices chosen by a damage-free color realization $\rho$ of a valid color plan are $\chi_\rho$-candidates.
In Section~\ref{sub:failing}, we define \emph{failing} $S$-profiles, which are exactly those that describe no $b$-coloring.
In Section~\ref{subsub:compute-damage-free}, we show how a damage-free color realization can be computed.
In Section~\ref{sub:change}, we describe what happens when a color plan or realization is slightly modified, which will be useful in Section~\ref{sub:getting-rid}.
Finally, in Section~\ref{sub:feasible}, we define \emph{feasible} color realizations, which are those that can be used to produce a $b$-coloring. Note that the goal of Section~\ref{sub:getting-rid} is to achieve \emph{pivot-freeness}, which is a property needed for feasibility.

\subsubsection{Vertices in $B_\rho$ are candidates}\label{sub:are-candidates}

Now we prove an important lemma. Informally, it says that the vertices chosen by a ``good'' color realization can become $b$-vertices.

\begin{lemma}\label{lem:valid-realization}
If $\rho$ is a damage-free color realization of a valid color plan $\pi$, then:
\begin{enumerate}[(a)]
\item $\chi_\rho$ is a proper partial coloring of $G$; and\label{vraux:proper}
\item each vertex $u \in B_\rho$ is a $\chi_\rho$-candidate.\label{vraux:candidate}
\end{enumerate}
\end{lemma}
\begin{proof}
For (a), let $u, v \in S_\rho$ be two neighbors in $G$, and let $c := \chi_\rho(u)$ and $d:= \chi_\rho(v)$. First, observe that if $u,v \in S$, then $c \ne d$ since $\chi$ is proper and $\chi \seq \chi_\rho$. Second, suppose that $u,v \notin S$.
Since $u \ne v$, we have $u = \rho(c) \ne \rho(d) = v$, which implies $c \ne d$.
Third, suppose, without loss of generality, that $u \in S$ and $v \notin S$.
Recall that $c \in [p]$ and $d \in [b+1, k]$, which means that $c \ne d$ unless $c,d \in [b+1, p]$.
By Property~\ref{cp:proper} of Definition~\ref{def:color-plan}, we have $\pi(c) \ne u$. Hence, $\rho(c) \notin N(u)$ by Definition~\ref{def:color-realization}. In particular, $v \ne \rho(c)$, which implies $d \ne c$.

For (b), let $u \in B_\rho$ be any vertex.
Since $B_\rho\seq \kb$, we know that $u$ is a $\chi$-candidate, i.e., it satisfies $|N[u] \setminus S| + |\chi(N[u])| = k + \red_\chi(u)$, where $\red_\chi(u) \ge 0$.
We need to show $|N[u] \setminus S_\rho| + |\chi_\rho(N[u])| \ge k$.
Hence, it suffices to show that:
\begin{equation}\label{eq:fen1}
|N[u] \setminus S| - |N[u] \setminus S_\rho| + |\chi(N[u])| - |\chi_\rho(N[u])| \le \red_\chi(u).
\end{equation}
Let $\alpha_u := (N[u] \setminus S) \setminus (N[u] \setminus S_\rho)$ and $\beta_u := \chi_\rho(N[u]) \setminus \chi(N[u])$.
Since $S \seq S_\rho$ and $\chi(N[u]) \seq \chi_\rho(N[u])$, we have
\begin{itemize}
\item $\alpha_u = N[u] \cap (S_\rho \setminus S) = N[u] \cap \range(\rho)$,
\item $|\alpha_u| = |N[u] \setminus S| - |N[u] \setminus S_\rho|$, and
\item $|\beta_u| = |\chi_\rho(N[u])| - |\chi(N[u])|$.
\end{itemize}
In other words, $\alpha_u$ is the set of neighbors of $u$ colored by $\chi_\rho$ but not by $\chi$, and $\beta_u$ is the set of ``new colors'' in the neighborhood of $u$.
Observe that to prove Equation~\ref{eq:fen1}, we need to show that $|\alpha_u| - |\beta_u| \le \red_\chi(u)$.
Let $\gamma_u = \{v \in \alpha_u \sep \chi_\rho(v) \in \chi(N(u))\}$, i.e., $\gamma_u$ is the set of neighbors of $u$ that are ``newly'' colored with a color that has already been present in the neighborhood. Observe that the function $v\mapsto \chi_\rho(v)$ of type 
$\alpha_u \setminus \gamma_u \rightarrow \beta_u$ is a bijection (informally, for each ``new color'' $c$, there is exactly one vertex colored with $c$ in the neighborhood, namely $\rho(c)$). Hence, $|\gamma_u| = |\alpha_u| - |\beta_u|$, and to show that Equation~\ref{eq:fen1} holds, it suffices to show that $|\gamma_u| \le \red_\chi(u)$.

First, suppose that $u \in B$, and recall that by Property~\ref{cp:redundancy} of Definition~\ref{def:color-plan}, we have $|\pi^{-1}(u) \cap \chi(N(u))| \le \red_\chi(u)$.
Let  $c \in [b+1, p]$ be a color, and observe that $c \in \pi^{-1}(u) \cap \chi(N(u))$ if and only if $\rho(c) \in N(u)\land c \in \chi(N(u))$ if and only if $\rho(c) \in \alpha_u \land \chi_\rho(\rho(c)) \in \chi(N(u))$ if and only if $\rho(c) \in \gamma_u$.
Since $v \in \gamma_u$ implies that $v = \rho(c)$ for some $c \in [b+1, p]$,
we obtain $|\gamma_u| = |\pi^{-1}(u) \cap \chi(N(u))| \le \red_\chi(u)$ as required.

Second, suppose that $u \notin B$.
Note that by Definition~\ref{def:color-realization}, we also have $u \notin S$.
If $\gamma_u = \emptyset$, then $|\gamma_u|\le \red_\chi(u)$ is obvious since  $\red_\chi(u) \ge 0$.
Hence, suppose $\gamma_u \ne \emptyset$.
Observe that this means that $\chi(N[u])\ne \emptyset$, which implies $u \in N(S)$.
By Definition~\ref{def:fen-core}, $|\chi(N[u])| = 1$, which implies $|\gamma_u|= 1$ since each vertex in $\gamma_u$ has different color in $\chi_\rho$. 
Hence, let $v \in \kout$ be the unique vertex in $\gamma_u$, let $c = \chi_\rho(v)$, and let $w \in N(u)$ be such that $\chi(w) = c$.
By Definition~\ref{def:fen-core}, $v \notin N(S)$ (otherwise $\{u, v\}$ would induce a short outer path), i.e., $v \in \kstar$.
If $\red_\chi(u) = 0$, then $u$ is the $(\chi, w, v)$-link and $u$ is damaged in $\rho$, which 
contradicts the damage-freeness of $\rho$.
Hence, $\red_\chi(u) \ge 1 = |\gamma_u|$ as required. 
\end{proof}

\subsubsection{Failing $S$-profiles describe no $b$-colorings}\label{sub:failing}

To simplify the notation, let us denote the $S$-profile $(\chi, B)$ by $\zeta$.
There are three possible reasons for the non-existence of a $b$-coloring described by $\zeta$.
The most obvious reason is that there are not enough candidates, i.e., $|\kb| < k$.
The second natural reason is that there is no valid color plan of $\zeta$.
In the following definition, we describe the third possible reason: informally, it says that there is an ``unavoidable'' pivot.

\begin{definition}\label{def:failing}
We say that $\zeta$ is \emph{candidate-failing} if $|\kb| < k$ and \emph{plan-failing} if there is no valid color plan of $\zeta$.
We say that $\zeta$ is \emph{pivot-failing} if there is a set $D \seq \kout \cup S$ and a vertex $u \in \overline{S \cup D}$ such that $u$ is a $(\chi, D)$-pivot, $\kout \seq D$,
$[b] \seq \chi(D)$, and (a) $|K^+| \le k$ or (b) $\{v, w\} \cap S \ne \emptyset$ for every $v, w \in D$ such that $w$ is the $(\chi, u, v)$-link.

We say that $\zeta$ is \emph{failing} if it is candidate-failing, pivot-failing or plan-failing.
\end{definition}

Let us remark that if Condition~(b) in Definition~\ref{def:failing} is satisfied, then there is actually only one $(u, D)$-link $w$ because it is $S$-influenced, see Section~\ref{subsub:general-pivot}.

Definition~\ref{def:failing} provides a sufficient and necessary condition for a $b$-coloring described by $\zeta$ to exist: $\zeta$ must not be failing.
In the following lemma, we show the necessity.

\begin{lemma}\label{lem:failing-imlies-no-b--coloring}
If $\zeta$ is failing, then there is no $k$-$b$-coloring of $G$ described by $\zeta$.
\end{lemma}
\begin{proof}
Suppose for contradiction that $\zeta$ is failing and there is a $k$-$b$-coloring $\psi$ of $G$ described by $\zeta$.
By Observation~\ref{obs:b-vertices-are-candidates}, $|\kb| \ge k$. Hence, $\zeta$ is plan-failing or pivot-failing.

First, suppose that $\zeta$ is pivot-failing.
Let $D$ and $u$ be as in Definition~\ref{def:failing}, and let $c := \psi(u)$.
If $c \in [b]$, let $v \in D$ be a vertex such that $\chi(v) = c$, and otherwise, let $v$ be the $b$-vertex for $c$ in $\kout$ ($v$ exists by Observation~\ref{obs:b-vertices-are-candidates}).
In both cases, we have $v \in D$.
Since $\psi$ is proper, we have $uv \notin E(G)$.
Hence, $u$ and $v$ are $\chi$-linked via a vertex $w \in \kb$.
Since $\psi(u) = \psi(v)$ and $u \notin S$, we may use Observation~\ref{obs:links} to deduce that $w$ is not a $b$-vertex in $\psi$.
By Observation~\ref{obs:b-vertices-are-candidates}, we obtain $|\kb| > k$, which means that $\{v, w\} \cap S \ne \emptyset$, see Definition~\ref{def:failing}.
If $w \in S$, then $w \in K^+ \cap S = B$, and $w$ would be a $b$-vertex in $\psi$, which it is not.
Hence, $w \notin S$ and $v \in S$.
Since $w$ is the only $(u, D)$-link in $N[v]$ by Lemma~\ref{lem:pivot-properties}(\ref{aux:component-of-GD}), we may use Lemma~\ref{lem:pivot-properties}(\ref{aux:link-in-S}) to deduce $D \cap S = \{v\}$.
Observe that $b \le 1$ because $[b] \seq \chi(D)$.
If $v', w' \in D$ are such that $w'$ is the $(\chi, u, v')$-link, then $\{v', w'\} \cap S \ne \emptyset$ by Definition~\ref{def:failing}, which implies $v' = v$.
Hence, $D \setminus N(u) = \{v\}$.
Since $|K| = |\kb| - b > k - b \ge k-1$, we have $|K| \ge k$.
Since $K \seq D$ and $v \notin K$, we have $K \seq N(u)$.
Hence, $u$ has degree at least $k$, which implies $u \in K$; a contradiction with $u \notin D$.

Now suppose that $\zeta$ is plan-failing. To obtain a contradiction, we will define a valid color plan $\pi$.
Let $\rho$ be a color realization such that $\rho(c)$ is a $b$-vertex for each $c \in [b+1, k]$ in $\psi$.
Suppose that $u \in \range(\rho)$ is damaged in $\rho$, i.e., there is a vertex $v \in S$ such that $u$ is the $(\chi, v, w)$-link, where $w = \rho(c) \in \kstar$ and $c = \chi(v) \in [b+1, p]$.
Since $u$ is a $b$-vertex in $\psi$, $w \notin S$, and $\psi(v) = \psi(w)$, we obtain a contradiction with Observation~\ref{obs:links}.
Hence, $\rho$ is damage-free.
Let $\pi$ be the color plan realized by $\rho$.
By Lemma~\ref{lem:rho-to-pi}, $\pi$ satisfies Properties~\ref{cp:star}, \ref{cp:enough-neighbor}, and~\ref{cp:block} of Definition~\ref{def:color-plan}.
Moreover, Property~\ref{cp:proper} holds because $\psi$ is proper (using Condition~(\ref{cr:S}) of Definition~\ref{def:color-realization}).
For Property~\ref{cp:redundancy}, suppose that there is a vertex $u \in B$ such that $C := \pi^{-1}(u) \cap \chi(N(u))$ satisfies $|C| > \red_\chi(u)$. By construction of $\pi$, $u$ has at least two neighbors of each color $c \in C$ in $\psi$: one in $S$ and one in $\kout$.
Hence, there are at most $\alpha := |\chi(N[u])| + |N[u]\setminus S| - |C|$ distinct colors in $\psi(N[u])$.
By definition of $\red_\chi$, we have $\alpha + |C| - k = \red_\chi(u) < |C|$. Hence, $\alpha < k$, which is a contradiction since $u$ must be a $b$-vertex in $\psi$.
Therefore, $\pi$ is a valid color plan.
\end{proof}

\subsubsection{Computing a damage-free color realization}\label{subsub:compute-damage-free}

Let us assume that $(\chi, B)$ is not failing.
Now we show how a damage-free and almost $0$-safe color realization can be computed, given a valid color plan $\pi$.
Informally, if $\pi$ is critical, then Algorithm~\ref{alg:find-damage-free} chooses a candidate that should be avoided.
Moreover, we must be careful about candidates at distance 2 to $S$ because they may be $\chi$-linked to a vertex of $S$: we use an auxiliary bipartite graph $H$ to handle this issue.
Finally, the almost $0$-safeness is achieved by choosing candidates close to $S^+$ for colors in $[p+1, k]$.

\newcommand{\com}[1]{\textup{\textbf{#1}}}
\newcommand{\Each}{\com{each}\xspace}
\newcommand{\From}{\com{from}\xspace}
\newcommand{\To}{\com{to}\xspace}

\begin{algorithm}[h]
\DontPrintSemicolon
\KwIn{A graph $G$, an integer $k$, a fen-core $S$, an $S$-profile $(\chi, B)$, a valid color plan~$\pi$.}
\SetKwProg{Fn}{Procedure}{:}{}

\If{$\pi$ is $c$-critical for some $c \in [b+1, p]$}
{
let $u_c$ be a $\pi$-anchor and let $v_c \in \kstar$ be $\chi$-linked to $u_c$ via $w_c$\;\label{l:anchor}
\tcp*{$u_c$, $v_c$, and $w_c$ exist by Definition~\ref{def:color-plan} since $\pi$ is valid}
$U := \kout \setminus \{w_c\}$;\label{l:critical}
}
\lElse{$U := K$\label{l:U}}
let $H$ be a bipartite graph with parts $C^* := \pi^{-1}(*)$ and $K^*$ such that $E(H) = \{cu \sep c \in C^*, u \in K^*, u$ is not $\chi$-linked to a vertex in $\chi^{-1}(c)\}$\;\label{l:H}

let $\mu$ be a maximal matching in $H$\;\label{l:matching}

\For{$c$ \From $b+1$ \To $p$}{
\If{$\pi(c) = u$ for some vertex $u \in S$}{
$\rho(c) := v$, where $v \in (U \cap N(u)) \setminus \range(\rho)$\;\label{l:NS}
\tcp*{we argue why $v$ exists in the proof of Lemma~\ref{lem:compute-damage-free}}
}
\lIf{$\pi$ is $c$-critical}{$\rho(c) := v_c$}
\lElseIf{$\pi(c) = *$}{$\rho(c) := u$, where $u \in K^*$ is the vertex matched with $c$ in $\mu$\label{l:use-match}}
\tcp*{we argue why $c$ is matched in $\mu$ in the proof of Lemma~\ref{lem:compute-damage-free}}
}
let $\prec$ be an ordering of $U$ such that if $u \prec v$, then $\dist(u, S^+) \le \dist(v, S^+)$\;
\lFor{$c$ \From $p+1$ \To $k$}{
$\rho(c) := u$, where $u$ is the $\prec$-smallest vertex in $U \setminus \range(\rho)$\label{l:p+1}}
\tcp*{we argue why $u$ exists in the proof of Lemma~\ref{lem:compute-damage-free}}
\Return $\rho$\;

\caption{Constructing a damage-free and $0$-safe color realization}
\label{alg:find-damage-free}
\end{algorithm}

\begin{lemma}\label{lem:compute-damage-free}
If $\pi$ is a valid color plan, then the function $\rho$ computed by Algorithm~\ref{alg:find-damage-free} with $\pi$ as input is a damage-free and almost $0$-safe color realization of $\pi$.
Moreover, Algorithm~\ref{alg:find-damage-free} runs in polynomial time, and if $\pi$ is not critical, then $\rho$ is $0$-safe.
\end{lemma}

\begin{proof}
Let $U$ be the set defined on line~\ref{l:critical} or~\ref{l:U}, let $H$ and $C^*$ be the objects defined on line~\ref{l:H} and let $\mu$ be the matching found on line~\ref{l:matching}. 

First, we show, for each color $c \in C^*$, that unless $\pi$ is $c$-critical, $c$ is matched in $\mu$, which is needed for line~\ref{l:use-match}.
Observe that each vertex $u \in K^*$ is $\chi$-linked to at most one vertex in $S$ (otherwise there would be an outer path of length 2).
This implies that $\pi$ is $c$-critical for at most one color $c \in [b+1, p]$ and that the degree of a vertex $u \in K^*$ in $H$ is $|C^*| -1$ or $|C^*|$.
If $\pi$ is $c$-critical for some $c \in [b+1, p]$, then $c \in C^*$ and $cu \notin E(H)$ for each $u \in K^*$, which means that $H - c$ is a complete bipartite graph and that all colors in $C^* \setminus \{c\}$ are matched in $\mu$.
Suppose that $\pi$ is not critical and let $C \seq C^*$ be a non-empty set.
Since $\pi$ is not critical, we have $N_H(C) \ne \emptyset$.
Moreover, if $|C| \ge 2$, we have $N_H(C) = K^*$ (since each $u \in K^*$ has high degree in $H$, as observed above).
Since $K^* \ge C^*$ by Property~\ref{cp:star} of Definition~\ref{def:color-plan}, all colors in $C^*$ are matched in $\mu$ by Hall's marriage theorem.

Second, we show that $\rho(c)$ can be always defined on line~\ref{l:NS}.
Let $u \in S$ and let $c \in \pi^{-1}(u)$.
By Property~\ref{cp:enough-neighbor} of Definition~\ref{def:color-plan}, we have $|\pi^{-1}(u)| \le |K \cap N(u)|$, which implies that $\rho(c)$ can be defined whenever $K \cap N(u) = U \cap N(u)$.
Suppose that $K \cap N(u) \ne U \cap N(u)$.
This implies that $\pi$ is $d$-critical for some $d \in [b+1, p]$; let $u_d$, $v_d$, and $w_d$ be the vertices defined on line~\ref{l:anchor}.
Clearly, we have $u = u_d$, i.e., $u$ is a $\pi$-anchor.
By Property~\ref{cp:block} of Definition~\ref{def:color-plan}, $|\pi^{-1}(u)| < |K \cap N(u)|$, which implies that $\rho(c)$ can again be defined (informally, we can afford to ignore $w_d$).

Third, we show that $\rho(c)$ can be always defined on line~\ref{l:p+1}.
Let $c \in [p+1, k]$.
Since $(\chi, B)$ is not candidate-failing, we have $|\kb| \ge k$, which means that $\rho(c)$ can be defined when $\pi$ is not critical (because $U = K$).
However, if $\pi$ is $d$-critical for some $d \in [b+1, p]$, then $|\kb| > k$ by Property~\ref{cp:block} of Definition~\ref{def:color-plan}, and $\rho(c)$ can again be defined (we can again afford to ignore $w_d$).

We have shown that $\rho(c)$ is defined for each $c \in [b+1, k]$.
Moreover, it can be easily observed that $\rho$ is injective and that it satisfies both conditions of Definition~\ref{def:color-realization}.
Hence, $\rho$ is a color realization of $\pi$.
To show that $\rho$ is damage-free, suppose that $w \in \range(\rho)$ is damaged in $\rho$, i.e., there is a vertex $u \in S$ and a color $c \in [b+1, p]$ such that $\chi(u) = c$, and $u$ and $v := \rho(c)$ are $\chi$-linked via $w$.
Observe that $vc \notin E(H)$.
Hence, $c$ and $v$ cannot be matched in $\mu$, which implies $v = v_c$, $w = w_c$, and $u = u_c$.
However, now we have a contradiction since $w_c \notin \range(\rho)$.
Hence, $\rho$ is damage-free.

Now let $u \in \kout \setminus B_\rho$ be a vertex.
We will show that $u$ is $0$-safe in $\rho$ unless $\pi$ is $c$-critical for some $c \in [b+1,p]$ and $u = w_c$.
This suffices because $w_c$ does not exist when $\pi$ is not critical, and if $\pi$ is critical, then $w_c$ has a neighbor in $B_\rho$, namely $v_c$, see Definition~\ref{def:color-realization}.
Suppose for contradiction that $u \in U$ is not $0$-safe in $\rho$, i.e., there is a color $c \in [p+1, k]$ such that $\dist(u,S^+) < \dist(\rho(c), S^+)$.
Since $u \notin B_\rho$, we know that $u \in U \setminus \range(\rho)$.
This is a contradiction since $\rho(c)$ should have been defined as $u$, see line~\ref{l:p+1}.

Finally, observe that Algorithm~\ref{alg:find-damage-free} runs in polynomial time.
In particular, the matching $\mu$ can be found by reducing to the maximum flow problem, and the ordering $\prec$ can be found using the BFS algorithm.
\end{proof}

\subsubsection{Changing the plan or the realization}\label{sub:change}

In Section~\ref{subsub:compute-damage-free}, we showed how a damage-free and almost $0$-safe color realization can be computed.
However, in some cases, we will need to modify the realization, which may even lead to the change of the color plan.
The following lemma will be used to show that the new plan is also valid.

\begin{lemma}\label{lem:is-valid-plan}
Let $\rho$ be a color realization of a valid color plan $\pi$.
\begin{enumerate}[(a)]
\item If $\rho'$ is a color realization such that for each $c \in [b+1, p]$ and $u \in S$, $\rho(c) \in N(u)$ if and only if $\rho'(c) \in N(u)$, then $\rho'$ realizes $\pi$.\label{vpaux:trivial}
\item If $c \in [b+1, p]$ and $\pi_1 = \pi[c \mapsto *]$ has a damage-free realization, then $\pi_1$ is valid.\label{vpaux:mapsto-*}
\item If $u \in S$, $c \in [b+1, p] \setminus \chi(N[u])$, and $\pi_1 := \pi[c \mapsto u]$ has a damage-free realization, then $\pi_1$ is valid.\label{vpaux:mapsto-u}
\end{enumerate}
\end{lemma}

\begin{proof}
Let $\rho$ and $\pi$ be as in the statement. First, observe that~(\ref{vpaux:trivial}) follows directly from Definition~\ref{def:color-realization}.

For~(\ref{vpaux:mapsto-*}), let $c$ and $\pi_1$ be as in the statement.
By Lemma~\ref{lem:rho-to-pi}, $\pi_1$ satisfies Properties~\ref{cp:star}, \ref{cp:enough-neighbor}, and \ref{cp:block} of Definition~\ref{def:color-plan}.
If $\pi_1(\chi(v)) = v$ for some $v \in S$, then also $\pi(\chi(v)) = v$, which contradicts the validity of $\pi$.
Hence, $\pi_1$ satisfies Property~\ref{cp:proper}.
Since $\pi_1^{-1}(v) \seq \pi^{-1}(v)$ for each $v \in B$, Property~\ref{cp:redundancy} is satisfied as well, and $\pi_1$ is valid as desired.

For~(\ref{vpaux:mapsto-u}), let $u, c$, and $\pi_1$ be as in the statement.
By Lemma~\ref{lem:rho-to-pi}, $\pi_1$ satisfies Properties~\ref{cp:star}, \ref{cp:enough-neighbor}, and \ref{cp:block} of Definition~\ref{def:color-plan}.
If $\pi_1(\chi(v)) = v$ for some $v \in S$, then by definition of $\pi_1$ and validity of $\pi$, we have $\chi(v) = c$ and $v = u$, which is a contradiction with $c \notin \chi(N[u])$.
Hence, $\pi_1$ satisfies Property~\ref{cp:proper}.  
Finally, suppose that $|\pi_1^{-1}(v) \cap \chi(N(v))| > \red_\chi(v)$ for some $v \in B$.
By definition of $\pi_1$ and validity of $\pi$, we have $v = u$.
However, $c \notin \chi(N(u))$, which implies $\pi_1^{-1}(u) \cap \chi(N(u)) = \pi^{-1}(u) \cap \chi(N(u))$, which contradicts the validity of $\pi$.
Hence, $\pi_1$ satisfies Property~\ref{cp:redundancy}, and $\pi_1$ is indeed valid.
\end{proof} 

The new color realization $\rho'$ may contain a damaged vertex.
Let us describe how the properties of this damaged vertex depend on $\rho'$.

\begin{lemma}\label{lem:damaged}
Let $\rho$ and $\rho'$ be color realizations such that $\rho$ is damage-free, and let $u \in \range(\rho')$ be a vertex damaged in $\rho'$, i.e., there is a color $c \in [b+1, p]$ such that $v := \rho'(c) \in N(u)$ and $c \in \chi(N(u))$.
\begin{enumerate}[(a)]
\item If $\rho(d) = \rho'(d)$ for each $d \in [b+1, p]$, then $u \notin \range(\rho)$.\label{daux:rhoc=rho'c}
\item If $\range(\rho) = \range(\rho')$, then $\rho(c) \ne v$.\label{daux:same-ranges}
\item If $\rho' = \rho[d \mapsto w]$ for a color $d \in [b+1, k]$ and a vertex $w \in K \setminus \range(\rho)$, then $w \in \{u, v\}$.\label{daux:shift}
\end{enumerate}
\end{lemma}

\begin{proof}
Let $\rho, \rho', u, c$, and $v$ be as in the statement.
For~(\ref{daux:rhoc=rho'c}), suppose that $\rho(d) = \rho'(d)$ for each $d \in [b+1, p]$.
In particular, $v = \rho(c)$.
Since $u$ is not damaged in $\rho$, we have $u \notin \range(\rho)$, as required.

For~(\ref{daux:same-ranges}), suppose that $\range(\rho) = \range(\rho')$.
If $\rho(c) = v$, then $u$ would be damaged in $\rho$, which is a contradiction.
Hence, $\rho(c) \ne v$, as required.

For~(\ref{daux:shift}), let $d$ and $w$ be as in the statement, and suppose that $w \notin \{u, v\}$.
Since $w \ne u$, we have $u \in \range(\rho)$, and since $w \ne v$, we have $\rho(c) = v$.
Hence, $u$ is damaged in $\rho$, which is a contradiction.
\end{proof}

Finally, let us show how a change of the realization affects the safeness of vertices.

\begin{observation}\label{obs:safety}
Let $\rho' = \rho[c \mapsto \rho(d), d \mapsto \rho(c)]$ for some colors $c, d \in [b+1, k]$.
If $\rho$ is $\ell$-safe, then $\rho'$ is $(\ell+2)$-safe, and if $\rho$ is almost $\ell$-safe, the $\rho'$ is almost $(\ell+2)$-safe.
\end{observation}
\begin{proof}
Let $\rho', c$, and $d$ be as in the statement.
Observe that if $\dist(u, S^+) < \dist(\rho'(c'), S^+)$ for some $u \in K \setminus B_{\rho'}$ and $c' \in [p+1, k]$, then either $\rho'(c') = \rho(c')$ or $c' \in \{c, d\}$.
Hence, each vertex in $K \setminus B_{\rho'}$ that is $\ell$-safe in $\rho$ is $(\ell+2)$-safe in $\rho'$.
Moreover, if $u \in N(B_\rho)$ or $|N(u) \cap N(B_\rho)| \le p+2$ for some $u \in K \setminus B_{\rho'}$, then $u \in N(B_{\rho'})$ or $|N(u) \cap N(B_{\rho'})| \le p+2$ because $B_\rho = B_{\rho'}$, which concludes the proof by definition of almost-safeness.
\end{proof}

\subsubsection{Feasible color realizations}\label{sub:feasible}

Now we define when a color realization $\rho$ is \emph{feasible}, which is a property that ensures that $\chi_\rho$ can be extended into a $b$-coloring.

\begin{definition}\label{def:feasible}
Let $\rho$ be a color realization. A vertex $u \in \overline{S_\rho}$ is a \emph{$\rho$-pivot} if there is a set $D \seq S_\rho$ such that $\chi_\rho(D) = [k]$ and $u$ is a $(\chi_\rho, D)$-pivot.
We say that $D$ is \emph{$\rho$-pivoted} by $u$.
We say that $\rho$ is \emph{pivot-free} if there is no $\rho$-pivot.

A vertex $u \in B_\rho$ is \emph{$\rho$-blocked} by a color $c \in [k]$ if $c \notin \chi_\rho(N[u])$ and for each $v \in N(u) \setminus S_\rho$, there is a vertex $w \in \chi_\rho^{-1}(c)$ such that $w \in N(v)$ or $w$ is $\chi_\rho$-linked to $v$ via a vertex in $B_\rho$.
We say that $\rho$ is \emph{block-free} if no vertex $u \in B_\rho$ is $\rho$-blocked by any color $c \in [k]$.

We say that $\rho$ is \emph{feasible} if it is damage-free, almost $13$-safe, pivot-free, and block-free, and it realizes a valid color plan.
\end{definition}

Informally, a $\rho$-pivot $u$ is a vertex that cannot be assigned any color: each color $c$ is either in its neighborhood or $u$ is $\chi_\rho$-linked to a vertex colored with $c$, see also Definition~\ref{def:links}.
Similarly, a vertex $u$ is $\rho$-blocked by $c$ if $c$ cannot be added to the neighborhood of $u$: each uncolored vertex either has a neighbor colored with $c$ or it is linked to a vertex colored with~$c$.

We will show how to make a realization pivot-free in Section~\ref{sub:getting-rid}.
In the remainder of this section, we show how to make a realization block-free (assuming it is pivot-free).
First, we describe the structure of a realization that is \emph{not} block-free.

\begin{lemma}\label{lem:structure-of-block-but-not-pivoted}
Let $\rho$ be a damage-free and pivot-free color realization of a valid color plan. If a vertex $u \in B_\rho$ is $\rho$-blocked by a color $c \in [k]$, then $|\chi_\rho(N[u])| \ge k-2$.
Moreover, there are distinct vertices $v_1, v_2, w_1, w_2$ such that $N(u) \setminus S_\rho = \{v_1, v_2\}$, $w_1 \in \chi^{-1}(c)$, $\rho(c) = w_2$, and for $i \in [2]$, either $w_i \in N(v_i)$ or $v_i$ and $w_i$ are $\chi_\rho$-linked.
\end{lemma}
\begin{proof}
Let $u \in B_\rho$ be as in the statement.
By Lemma~\ref{lem:valid-realization}(\ref{vraux:candidate}), $u$ is a $\chi_\rho$-candidate.
Since $c \notin \chi_\rho(N[u])$, there is a vertex $v_1 \in N(u) \setminus S_\rho$.
If $N(u) \setminus S_\rho = \{v_1\}$, then $u$ is $\chi_\rho$-tight and each color $d \in [k]$ such that $d \ne c$ is in $\chi_\rho(N[u])$.
By definition of being $\rho$-blocked, $v_1$ is a $\rho$-pivot, which contradicts the pivot-freeness of $\rho$.
Hence, there is a vertex $v_2 \in N(u) \setminus S_\rho$ such that $v_1 \ne v_2$.
Let $v_3 \in N(u) \setminus S_\rho$ be a vertex such that, if possible, $v_3 \notin \{v_1, v_2\}$.
By definition of a $\rho$-blocked vertex, there is, for each $i \in [3]$, a vertex $w_i \in \chi_\rho^{-1}(c)$ such that $w_i \in N(v_i)$ or $w_i$ is $\chi_\rho$-linked to $v_i$ via a vertex $x_i \in B_\rho$.
If $w_i \in N(v_i)$, we set $x_i := v_i$.

Let $i,j \in [3]$ be such that $v_i \ne v_j$.
By Observation~\ref{obs:short-cycles-are-in-S}, $w_i \ne w_j$, since otherwise $G[\{w_i, x_i, v_i, u, v_j, x_j\}]$ would be a cycle (and $v_i \notin S$).
Moreover, $\{w_i, w_j\} \nsubseteq S$, since otherwise $G[\{x_i, v_i, u, v_j, x_j\}]$ would be a short outer path.
If $w_i, w_j \notin S$, then $w_i = \rho(c) = w_j$, a contradiction.
Hence, without loss of generality, $w_1 \in S$, $w_2 = \rho(c) = w_3$ and $v_2 = v_3$.
By choice of $v_3$, $N(u) \setminus S_\rho = \{v_1, v_2\}$.
Since $u$ is a $\chi_\rho$-candidate, we have $|\chi_\rho(N[u])| \ge k -2$, as desired.
\end{proof}

Now we are ready to make the realization feasible.

\begin{lemma}\label{lem:handle-block}
If $\rho$ is a damage-free, pivot-free, and almost $11$-safe color realization of a valid color plan $\pi$, then a feasible color realization $\rho'$ can be computed in polynomial time.
\end{lemma}

\begin{proof}
Let $\rho$ be as in the statement.
If $\rho$ is block-free, then we may set $\rho' := \rho$.
Suppose that $\rho$ is not block-free, i.e., there is a vertex $u \in B_\rho$ $\rho$-blocked by a color $c \in [k]$.
By Lemma~\ref{lem:structure-of-block-but-not-pivoted}, we have $|\chi_\rho(N[u])| \ge k -2$, and there are distinct vertices $v_1, v_2, w_1, w_2$ such that $N(u) \setminus S_\rho = \{v_1, v_2\}$, $w_1 \in \chi^{-1}(c)$, $\rho(c) = w_2$, and for $i \in [2]$, either $w_i \in N(v_i)$ or $v_i$ and $w_i$ are $\chi_\rho$-linked via a vertex $x_i$ (if $w_i \in N(v_i)$, we set $x_i := v_i$), see Figure~\ref{fig:blocked}.
Since $k \ge p + 3$, there is a color $d \in [p+1, k]$ such that $v := \rho(d) \in N(u)$.
Let us define $\rho' := \rho[c \mapsto v, d \mapsto w_2]$.
Observe that $u, v, v_2, x_2, w_2 \notin N[S]$; otherwise, there would be a short outer path because $w_1 \in S$.

\begin{figure}[h]
\begin{tikzpicture}
\tikzmath{ \dist = 20;}
\begin{scope}[every node/.style={draw, circle, minimum width=17pt, inner sep=2pt, fill=gray!25!white}]
\node (u) {$u$};
\node[right= \dist pt of u, fill=white] (v1) {$v_1$};
\node[right = \dist pt of v1] (x1) {$x_1$};
\node[above = \dist pt of x1, rectangle, minimum height=17pt] (w1) {$w_1$};
\node[below left = \dist pt of u] (y1) {$y_1$};
\node[below right = \dist pt of u] (y2) {$y_2$};
\node[left = \dist pt of u] (v) {$v$};
\node[above = \dist pt of u, fill=white] (v2) {$v_2$};
\node[left = \dist pt of v2] (w2) {$w_2$};
\draw (w2)--(v2)--(u)--(v1)--(x1)--(w1);
\draw (y1)--(u)--(y2);
\draw (u)--(v);
\end{scope}
\node[right= 0pt of v2] {$= x_2$};
\node[left= 2pt of w2, yshift=2pt] (c) {$c$};
\node[left= 2pt of v] (d) {$d$};
\draw[red, <->] (c)--(d);
\end{tikzpicture}
\centering
\caption{An illustration of the proof of Lemma~\ref{lem:handle-block}. All depicted vertices are in $S_\rho$ except for $v_1$ and $v_2$. The only vertex in $S$ is $w_1$. The modification leading to $\rho'$ is symbolized by the double-arrow. Note that in contrast to this figure, $x_1 = v_1$ is possible.
For completeness, let us remark that $v_1 \ne x_1 \land v_2 \ne x_2$ is impossible because then $c$, $\rho^{-1}(x_1)$, and $\rho^{-1}(x_2)$ would be three colors missing in $\chi_\rho(N[u])$; this observation is not needed in the proof though.}
\label{fig:blocked}
\end{figure}
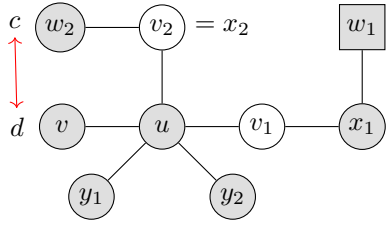

First, observe that $\rho'$ is a color realization of $\pi$ by Lemma~\ref{lem:is-valid-plan}(\ref{vpaux:trivial}) because $\rho(c') = \rho'(c')$ for $c' \in [b+1, p] \setminus \{c\}$ and $\rho(c), \rho'(c) \notin N(S)$.
Second, suppose that there is a vertex $w \in \range(\rho') \cap N(S)$ that is damaged in $\rho'$.
By Lemma~\ref{lem:damaged}(\ref{daux:same-ranges}), there is a vertex $x \in N(w) \cap \range(\rho)$ such that $\chi_\rho(x) \ne \chi_{\rho'}(x)$, i.e., $w \in N(w_2) \cup N(v)$.
Now $G[\{w, w_2, x_2, v_2, u, v_1, x_1\}]$ or $G[\{w, v, u, v_1, x_1\}]$ is an outer path of length at most 6, which is a contradiction.
Hence, $\rho'$ is damage-free.

Third, suppose that $w \in \overline{S_{\rho'}}$ is a $\rho'$-pivot and let $D \seq S_{\rho'}$ be a set $\rho'$-pivoted by $w$.
Since $k \ge p + 5$, there are colors $c_1, c_2 \in [p+1, k]$ such that $y_1, y_2 \in N(u)$, where $y_1 = \rho'(c_1)$ and $y_2 = \rho'(c_2)$.
Observe that $y_1, y_2 \in D$ because $\chi_{\rho'}(D) = [k]$ by Definition~\ref{def:feasible}.
Hence, if $u \notin D$, then we may use Lemma~\ref{lem:pivot-properties}(\ref{aux:equals-pivot}) to deduce $u = w$, which is a contradiction since $u \in B_{\rho'}$.
Now we know that $u \in D$.
Since $y_1, y_2 \in D \cap N(u)$, we deduce that $u$ is a $(w, D)$-link.
In particular, $w \in N(u)$, which implies $w \in \{v_1, v_2\}$.
If $w = v_1$, then $d \notin \chi_{\rho'}(D)$ since $\rho'(d) = w_2$ and $\dist(v_1, w_2) \ge 3$, which is a contradiction since $D$ is $\rho'$-pivoted.
If $w = v_2$, then $v_2$ would be a $\rho$-pivot, which is again a contradiction since $\rho$ is pivot-free.
Hence, $\rho'$ is pivot-free.

Fourth, suppose that some vertex $u' \in B_{\rho'}$ is $\rho'$-blocked by some color $c' \in [k]$.
By Lemma~\ref{lem:structure-of-block-but-not-pivoted}, $|\chi_{\rho'}(N[u'])| \ge k -2$.
Hence, if $u' \ne u$, then $|N(u) \cap N(u')| \ge 2$, which implies the existence of a length-4 cycle containing $u \notin S$, a contradiction with Observation~\ref{obs:short-cycles-are-in-S}.
Now we know that $u' = u$.
Since $c \in \chi_{\rho'}(u)$, we have $c' \ne c$.
By Lemma~\ref{lem:structure-of-block-but-not-pivoted}, there is a $u$-$S$ path $P$ of length at most 3 such that the endpoint of $P$ in $S$ is in $\chi^{-1}(c')$.
Now there is an outer path of length at most 4 in $G[V(P) \cup \{v_1, x_1\}]$, which contradicts Definition~\ref{def:fen-core}.
Hence, $\rho'$ is block-free.
Finally, $\rho'$ is almost $13$-safe by Observation~\ref{obs:safety}.
Hence, $\rho'$ is feasible, as required.
\end{proof}

\subsection{Eliminating the pivot}\label{sub:getting-rid}

The goal of this section is to show how a pivot-free color realization can be computed, see Definition~\ref{def:feasible}.
As before, let $S$ be a fen-core, $S^+$ be an extension of $S$, $(\chi, B)$ be an $S$-profile that is not failing, $\kout = \{u \in \overline{S} \sep \red_\chi(u) \ge 0\}$, and $\kb = \kout \cup B$.
Recall that $\chi(S) \seq [p]$ and $\chi(B) = [b]$.

Since this section is rather extensive, we begin with an intuitive overview.
Let $\rho$ be the color realization computed by Algorithm~\ref{alg:find-damage-free}, and suppose that there is a set $D \seq S_\rho$ that is $\rho$-pivoted by a vertex $u \in \overline{S_\rho}$, see Definition~\ref{def:feasible}.
We say that $\rho'$ is obtained by \emph{swapping} of two colors $c, d \in [b+1, k]$ if $\rho' = \rho[c \mapsto \rho(d), d \mapsto \rho(c)]$, and that $\rho'$ is obtained by \emph{shifting} a color $c \in [b+1, k]$ if $\rho' = \rho[c \mapsto v]$ for some vertex $v \in K \setminus \range(\rho)$.
Our goal is to make a few swaps and shifts so that the result is a pivot-free and damage-free color realization of a valid color plan.
Recall that there are two more properties needed for feasibility: the first one is block-freeness, which is achievable by Lemma~\ref{lem:handle-block}.
The second property is almost $13$-safeness. We have already shown that a swap increases the safeness by at most 2, see Observation~\ref{obs:safety}.
Similarly, Lemma~\ref{lem:get-safety} will be used to handle the ``safeness'' of a shift.

Recall that the $\rho$-pivot $u$ is also a $(\chi_\rho, D)$-pivot by Definition~\ref{def:feasible}.
However, $u$ is not necessarily a $(\chi, D)$-pivot.
Indeed, if a $(u, D)$-link $v$ has two neighbors of the same color in $\chi_\rho$, one in $S$ and one in $\range(\rho)$, then $v$ is not $\chi$-tight.
The case when $u$ is not $(\chi, D)$-pivot will be the easier one: in Section~\ref{subsub:not-pivot}, we show that a single swap suffices to obtain the desired realization $\rho'$.
From now on, suppose that $u$ is a $(\chi, D)$-pivot.

The basic idea how to modify $\rho$ so that $u$ is no longer a $\rho$-pivot is simple: choose a color $c \in [p+1, k]$ and move it far away from $u$ (via a swap or a shift).
After such a movement, we obtain a color realization $\rho'$ such that $u$ is not a $\rho'$-pivot because there is only one vertex colored with $c$, namely $\rho'(c)$.
The second possibility is to move $c$ to $u$ itself; this is, of course, possible only if $u \in K$.
We formalize this idea as follows: let $Q \seq S \cup K$ be the maximal set such that $u$ is a $(\chi, Q)$-pivot.
Observe that $D \seq Q$.
If $K \nsubseteq Q$, then we can simply move $c$ to a vertex in $K \setminus Q$.
The more interesting case occurs when $K \seq Q$.
Now we use the fact that $(\chi, B)$ is not pivot-failing.
By Definition~\ref{def:failing}, we can find two vertices $v, w \in Q \setminus S$ such that $v$ is the $(\chi, u, w)$-link.
Using a few swaps and/or shifts, we can achieve the situation in which $\rho'(c) = w$ and $v \notin \range(\rho')$; recall that $c \in [p+1, k]$.
Now it can again be easily observed that $u$ is not a $\rho'$-pivot.
In other words, $u$ can be safely colored with $c$.

If our goal was to find a color realization $\rho'$ such that $u$ is not a $\rho'$-pivot, then we would be done.
Unfortunately, there may be a $\rho'$-pivot $u'$ such that $u' \ne u$.
We are able to show that if $\rho'$ is obtained by a swap, than such a vertex $u'$ cannot exist, see Lemma~\ref{lem:not-different-pivot-swap}.
However, this is not true for a shift: three cases in which a different pivot emerges are depicted in Figure~\ref{fig:different-pivot}.
Hence, we must be very careful which color $c$ we choose to shift.

\begin{figure}[t]
\begin{minipage}{0.2\textwidth}
\centering
\begin{tikzpicture}
\tikzmath{ \dist = 20;}
\begin{scope}[every node/.style={draw, circle, inner sep=2pt, fill=gray!20!white, minimum width=10pt}]
\node[minimum width=17pt, fill=white] (u) {$u$};
\node[minimum width=17pt, below= \dist pt of u] (u') {$u'$};
\node[above= \dist pt of u] (a) {};
\node[above left= \dist pt of u] (b) {};
\node[above right= \dist pt of u] (c) {};
\node[below= \dist pt of u'] (d) {};
\node[below right= \dist pt of u'] (e) {};
\node[below left= \dist pt of u'] (f) {};
\end{scope}
\node[left=2 pt of u'] (col) {$c$};
\draw (a)--(u)--(u')--(d);
\draw (b)--(u)--(c);
\draw (e)--(u')--(f);
\draw[->, red] (col)-- (u);
\end{tikzpicture}
\end{minipage}
\hfill
\begin{minipage}{0.35\textwidth}
\centering
\begin{tikzpicture}
\tikzmath{ \dist = 20;}
\begin{scope}[every node/.style={draw, circle, inner sep=2pt, minimum width=10pt, fill=gray!20!white}]
\node[minimum width=17pt, fill=white] (u) {$u$};
\node[below left= \dist pt of u, rectangle, minimum width=17pt, minimum height=17pt] (v) {$v$};
\node[below right= \dist pt of u, minimum width=17pt] (w) {$w$};
\node[minimum width=17pt, above right= \dist pt of w] (u') {$u'$};
\node[below left= \dist pt of w] (aa) {};
\node[below right= \dist pt of w] (bb) {};
\node[below= \dist pt of w] (cc) {};
\node[below right=\dist pt of u'] (d) {};
\node[right= \dist pt of u', fill=white] (e) {};
\draw (v)--(u)--(w)--(u');
\draw (aa)--(w)--(bb);
\draw (w)--(cc);
\draw (d)--(u')--(e);
\end{scope}
\node[above=1pt of v] {$d$};
\node[below=1pt of d] {$d$};
\node[above=1pt of u'] (col) {$c$};
\draw[->, red] (col)-- (e);
\end{tikzpicture}
\end{minipage}
\hfill
\begin{minipage}{0.35\textwidth}
\centering
\begin{tikzpicture}
\tikzmath{ \dist = 20;}
\begin{scope}[every node/.style={draw, circle, inner sep=2pt, minimum width=10pt, fill=gray!20!white}]
\node[minimum width=17pt, fill=white] (u) {$u$};
\node[below left= \dist pt of u, rectangle, minimum width=17pt, minimum height=17pt] (v) {$v$};
\node[below right= \dist pt of u, minimum width=17pt] (w) {$w$};
\node[minimum width=17pt, above right= \dist pt of w, fill=white] (u') {$u'$};
\node[below left= \dist pt of w] (aa) {};
\node[below right= \dist pt of w] (bb) {};
\node[below= \dist pt of w] (cc) {};
\node[below right=\dist pt of u'] (d) {};
\node[right= \dist pt of u', fill=white] (e) {};
\node[above right=\dist pt of u, minimum width=17pt] (x) {$x$};
\draw (v)--(u)--(w)--(u');
\draw (aa)--(w)--(bb);
\draw (w)--(cc);
\draw (d)--(u')--(e);
\draw (u)--(x);
\end{scope}
\node[above=1pt of v] {$d$};
\node[below=1pt of d] {$d$};
\node[right=1pt of x] (col) {$c$};
\draw[->, red] (col)-- (e);
\end{tikzpicture}
\end{minipage}
\caption{Three possible cases, in which there is a $\rho$-pivot $u$ and a $\rho'$-pivot $u'$. The gray vertices are in $S_\rho$ and the square vertices are in $S$. The shift leading to $\rho'$ is depicted by the red arrow. \textbf{Left:} Both $u$ and $u'$ are $\chi$-tight, and for each $c \in [b+1, k]$, we have $\rho(c) \in N(u) \cup N(u')$. \textbf{Middle and right:} The vertex $w$ is $\chi$-tight and for each color $c' \in [b+1, k] \setminus \{c, d\}$, we have $\rho(c') \in N[w]$. It does not matter that $v$ is far from $u'$ because $\rho(d) \in N(u')$.}
\label{fig:different-pivot}
\end{figure}
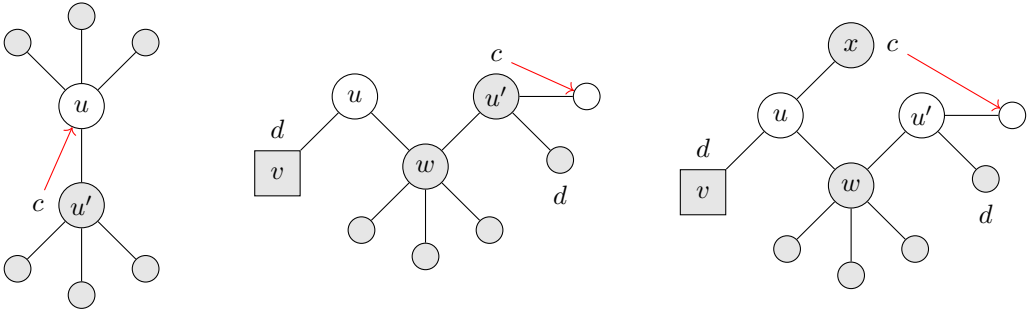

The rest of this section is organized as follows.
In Section~\ref{subsub:init-lemmas}, we prove several lemmas about $\rho$-pivoted sets.
Then, in Section~\ref{subsub:different-pivot}, we investigate how a different $\rho'$-pivot $u'$ can be avoided, see Figure~\ref{fig:different-pivot}.
In Section~\ref{subsub:not-pivot}, we handle the easier case in which $u$ is not a $(\chi, D)$-pivot.
In Section~\ref{subsub:operations}, we discuss when a change of the realization does not create a damaged vertex.
The core part of the proof, which shows how to get rid of the $\rho$-pivot $u$, is presented in Section~\ref{subsub:preventing-u}.
Finally, everything is combined in Section~\ref{subsub:compute-feasible}, and we obtain a feasible color realization.

\subsubsection{Initial lemmas}\label{subsub:init-lemmas}

We begin by observing a few simple properties of $\rho$-pivoted sets.

\begin{lemma}\label{lem:rho-pivot-properties}
Let $\rho$ be a color realization, let $u \in \overline{S_\rho}$ be a $\rho$-pivot, and let $D \seq S_\rho$ be a set $\rho$-pivoted by $u$.
\begin{enumerate}[(a)]
\item If $v$ is a $(u, D)$-link, then $\chi_\rho(N[v]) \ne [k]$.\label{aux:missing-color}
\item $|N(u)\cap D|\ge 2$.\label{aux:pivot-has-two-neighbors}
\item If $v \in D \setminus N(u)$, then $v \in N(B_\rho)$.\label{aux:in-NBrho}
\item If $u$ is a $(\chi, D)$-pivot and $\rho'$ is a damage-free color realization of a valid color plan such that $\rho^{-1}(D) = (\rho')^{-1}(D)$, then $D$ is $\rho'$-pivoted by $u$.\label{aux:rho'-pivoted}
\end{enumerate}
\end{lemma}

\begin{proof}
Let $\rho$, $u$, and $D$ be as in the statement.
For~(\ref{aux:missing-color}), let $v$ be a $(u, D)$-link. Since $u \in N(v) \setminus S_\rho$ and $v$ is $\chi_\rho$-tight, we have  $\chi_\rho(N[v]) \ne [k]$ as desired.

For~(\ref{aux:pivot-has-two-neighbors}), suppose that $N(u)\cap D = \{v\}$.
By Lemma~\ref{lem:pivot-properties}(\ref{aux:component-of-GD}), $D \seq N[v]$.
Since $|D| \ge k \ge 2$, we know that $v$ is a $(u, D)$-link.
By~(\ref{aux:missing-color}), $\chi_\rho(N[v]) \ne [k]$, which contradicts the $\rho$-pivotedness of $D$.

For~(\ref{aux:in-NBrho}), suppose that there are vertices $w \in D \cap N(u)$ and $v \in D \cap N(w)$.
By Definitions~\ref{def:links} and~\ref{def:feasible}, $w$ is a $(\chi_\rho, u, v)$-link, which implies $w \in K^+$.
Since $K^+ \cap S_\rho = B_\rho$, we have $w \in B_\rho$ as desired.

For~(\ref{aux:rho'-pivoted}), suppose that $u$ is a $(\chi, D)$-pivot and let $\rho'$ be as in the statement.
Observe that $\range(\rho) \cap D = \range(\rho') \cap D$ and $\chi_{\rho'}(D) = [k]$.
Hence, it suffices to show that each $(u, D)$-link $v$ is $\chi_{\rho'}$-tight.
By Definition~\ref{def:links}, $v \in K^+$, which implies $v \in B_{\rho'}$. 
Hence, by Lemma~\ref{lem:valid-realization}(\ref{vraux:candidate}), $v$ has $\chi_{\rho'}$-redundancy at least $0$.
Since $u$ is a $(\chi, D)$-pivot, we know that $v$ is $\chi$-tight, which implies that $v$ has $\chi_{\rho'}$-redundancy at most $0$.
Hence, $v$ is $\chi_{\rho'}$-tight, as desired.
\end{proof}

Recall that by Lemma~\ref{lem:compute-damage-free}, we can compute a damage-free color realization that is $0$-safe if the used color plan is not critical, see Definition~\ref{def:color-plan}.
Since $0$-safeness will be useful, we prove the following lemma.

\begin{lemma}\label{lem:not-critical}
If $\rho$ is a color realization of a valid color plan $\pi$ such that there is a $\rho$-pivot $u\in \overline{S_\rho}$, then $\pi$ is not critical.
\end{lemma}

\begin{proof}
Let $\rho$, $\pi$, and $u$ be as in the statement, and suppose for contradiction that $\pi$ is critical.
Let $D \seq K \cup S$ be the maximal set $\rho$-pivoted by $u$.
By Lemma~\ref{lem:rho-pivot-properties}(\ref{aux:pivot-has-two-neighbors}), $G[D]$ is disconnected.
Let $v_1, v_2 \in D$ be two vertices in different components of $G[D]$ such that, if possible, $v_1 \in S$.
Moreover, let us choose $v_1$ and $v_2$ so that the distance between them is as large as possible.
Let us fix $i \in [2]$.
If $v_i \in N[S]$, then let $w_i$ be the vertex in $S \cap N[v_i]$.
Otherwise, by Definition~\ref{def:color-plan}, there are vertices $w_i, x_i$ such that $x_i \in S$, $\chi(x_i) \in [b+1, p]$, and $w_i$ is the $(\chi, v_i, x_i)$-link.

If $w_1, w_2 \notin \{u\}$, then both $v_1$ and $v_2$ are $S$-influenced, and we have a contradiction with Lemma~\ref{lem:pivot-properties}(\ref{aux:contains-all-S-influenced}).
Hence, suppose that $w_i = u$ for some $i \in [2]$.
Since $u \notin S$, we have $x_i \ne w_i$, and by maximality of $D$, we have $x_i \in D \cap S$.
By choice of $v_1$, we have $v_1 \in S$ and $v_1 = w_1 = x_1$, which implies $i = 2$.
Since $\chi(x_2) \in [b+1, p]$, we have $x_2 \notin B$, which means that $x_2$ is not a $(u, D)$-link.
By Lemma~\ref{lem:pivot-properties}(\ref{aux:contains-all-S-influenced}), $S \cap D = \{x_2\}$, which implies $v_1 = x_2$.
Since $v_2 \in N(u)$, we have $D \seq N(u)$ by choice of $v_1$ and $v_2$ (otherwise a different choice would increase the distance between them).
Hence, $u = w_2$ has degree at least $k$, which is a contradiction since $w_2$ is $\chi$-tight.
\end{proof}

We will need to ensure that the obtained color realization is almost $13$-safe, see Definitions~\ref{def:color-realization} and~\ref{def:feasible}.
We have already observed that ``swapping'' two colors does not increase the safeness too much, see Observation~\ref{obs:safety}.
Now we discuss the situation in which a color is ``shifted'' to a different candidate.

\begin{lemma}\label{lem:get-safety}
Let $\rho$ be an $\ell$-safe color realization for some $\ell \in \bb N$, let $u \in \overline{S_\rho}$ be a $\rho$-pivot,
let $D \seq S_\rho$ be a set $\rho$-pivoted by $u$, and let $c \in [b+1, k]$ be a color.
If $v^+ \in K \setminus \range(\rho)$, then $\rho' := \rho[c \mapsto v^+]$ is almost $(\ell+1)$-safe.
\end{lemma}
\begin{proof}
Let $\rho, u, D, c, v^+$, and $\rho'$ be as in the statement, and let $v^- := \rho(c)$.
First, we show that each vertex $w \in \kout\setminus (\range(\rho') \cup \{v^-\})$ is $(\ell+1)$-safe in $\rho'$.
Let $X = \{x \in \rho'([p+1, k]) \sep \dist(w,S^+) < \dist(x, S^+)\}$; we need to show that $|X| \le \ell+1$.
Since $w \in\kout\setminus \range(\rho)$, we know that $w$ is $\ell$-safe in $\rho$, which means $|X \cap \rho([p+1, k])| \le \ell$.
Since $\rho'([p+1, k]) \setminus \rho([p+1,k]) \seq \{v^+\}$, we have $|X| \le \ell+1$ as desired.

Second, we need to show that either $v^- \in N(B_{\rho'})$ or $|N(v^-) \cap N(B_{\rho'})| \le p +2$, see Definition~\ref{def:color-realization}.
Suppose that $v^- \notin N(B_{\rho'})$ and observe that also $v^- \notin N(B_\rho)$ because $B_\rho \seq B_{\rho'} \setminus \{v^-\}$.
Since $v^- \in D$, we have $v^- \in N(u)$ by Lemma~\ref{lem:rho-pivot-properties}(\ref{aux:in-NBrho}).
Let $W = B \cup \rho'([b+1, p]) \cup \{v^+\}$ and observe that $|N(v^-) \cap N(w)| \le 1$ for each $w \in W$ by Observation~\ref{obs:short-cycles-are-in-S} (otherwise there would be a cycle of length 4 not in $S$).
Hence, $|N(v^-) \cap N(W)| \le p+1$.
Let $x \in B_{\rho'} \setminus W$ and observe that $x = \rho'(c')$ for some $c' \in [p+1, k]$.
Since $x \notin W$, we have $x \ne v^+$, which means that $x = \rho(c')$.
Since $D$ is $\rho$-pivoted, we have $x \in D$.
Let $u' \in N(x) \cap N(v^-)$ be a vertex.
Since $v^- \in N(u) \setminus S$, we have $u' = u$ by Lemma~\ref{lem:pivot-properties}(\ref{aux:equals-pivot}).
Hence, $|N(v^-) \cap N(B_{\rho'} \setminus W)| \le 1$, which implies $|N(v^-) \cap N(B_{\rho'})| \le p+2$ as desired.
\end{proof}

It is not hard to see that if $\rho'$ is obtained from $\rho$ by ``moving'' the candidate for a color $c \in [p+1, k]$ far from $u$, then $u$ will not be a $\rho'$-pivot.
The following lemma formalizes this idea in the case when $u$ is a $(\chi, D)$-pivot.

\begin{lemma}\label{lem:u-is-not-rho'-pivot}
Let $\rho$ be a damage-free color realization of a valid color plan, let $u$ be a $\rho$-pivot, let $D \seq S_\rho$ be the maximal set $\rho$-pivoted by $u$, let $Q \seq S \cup K$ be the maximal set such that $u$ is a $(\chi, Q)$-pivot, let $v \in K \setminus Q$, let $c \in [p+1, k]$, and let $\rho'$ be a color realization such that $\rho' = \rho[c \mapsto v]$, or $\rho^{-1}(v) = d$ and $\rho' = \rho[c \mapsto v, d \mapsto \rho(c)]$.
If $u$ is a $(\chi, D)$-pivot, then $u$ is not a $\rho'$-pivot.
\end{lemma}
\begin{proof}
Let $\rho, u, D, Q, v, c$, and $\rho'$ be as in the statement, and suppose that $u$ is a $\rho'$-pivot.
Let $D' \seq S_{\rho'}$ be a set $\rho'$-pivoted by $u$; since $\rho'(c) = v$ and $c \in [p+1, k]$, we have $v \in D'$.
By maximality of $Q$, we have $v \notin N(u)$, which means that there is a $(\chi_{\rho'}, u, v)$-link $w \in D'$.
Observe that $w \in D$ by maximality of $D$ because $S_{\rho'} \setminus S_\rho \seq \{v\}$.
Let $C = \{c' \in [b+1, p] \sep c' \in \chi(N[w]), \rho'(c') \in N[w]\}$ and observe that $C \ne \emptyset$; otherwise $w$ would be $\chi$-tight, which would contradict $v \notin Q$.

First, suppose that $\rho(c') \in N[w]$ for each $c' \in C$, which implies that $w$ is $\chi_\rho$-tight.
Let $c' \in C$ and let $x := \rho(c')$; observe that $x \ne w$ because $\chi_{\rho}$ is proper by Lemma~\ref{lem:valid-realization}(\ref{vraux:proper}) and because $c' \in \chi(N[w])$.
By maximality of $D$, we have $x \in D$, which implies that $w$ is a $(u, D)$-link.
Since $u$ is a $(\chi, D)$-pivot, we deduce that $w$ is $\chi$-tight and $v \in Q$, which is a contradiction.

Second, let $c' \in C$ be such that $\rho(c') \notin N[w]$.
In particular $x := \rho'(c') \ne \rho(c')$.
Since $c > p$, we have $c' \ne c$, which implies $x \ne v$.
Hence, by definition of $\rho'$, we have $v = \rho(c')$, which is a contradiction with $v \in N(w)$. 
We have proven that $u$ is not a $\rho'$-pivot.
\end{proof}

The following technical lemma informally states that if we have a $\rho$-pivot $u$ such that all candidates are close to $u$ (this is the condition $K \seq Q \cup \{u\}$) and can transform $\rho$ into $\rho'$ so that there is a $\rho'$-pivot $u'$ (and some special vertices $v^-$ and $v^+$), then some candidate chosen by $\rho'$ is far from $u'$ (this is the condition $\range(\rho') \nsubseteq Q'$).
This lemma will allow us transform $\rho'$ into a pivot-free color realization $\rho''$. 

\begin{lemma}\label{lem:nsubseteqQ2}
Let $\rho$ and $\rho'$ be a color realizations, let $u$ be a $\rho$-pivot, let $u'$ be a $\rho'$-pivot, let $Q, Q' \seq S \cup K$ be the maximal sets such that $u$ is a $(\chi, Q)$-pivot and $u'$ is a $(\chi, Q')$-pivot, let $D$ be a maximal set $\rho$-pivoted by $u$, and let $D'$ be a maximal set $\rho'$-pivoted by $u'$.
Let $v^- \in (N(u) \cap K) \setminus \range(\rho')$ be such that for some color $c \in [p+1, k]$, we have $v^+ := \rho'(c) \in N(v^-)$.
If $D \seq Q$, $K \seq Q \cup \{u\}$, and $u \notin D'$, then $\range(\rho') \nsubseteq Q'$.
\end{lemma}

\begin{proof}
Let $\rho, \rho', u, u', Q, Q', D, D', v^-, c$, and $v^+$ be as in the statement.
Since $v^-, v^+ \in K$ and $c > p$, we have $v^- \in Q \setminus D'$ and $v^+ \in Q \cap D'$; let $w^+$ be the vertex in $D' \cap N(u') \cap N[v^+]$.
Suppose for contradiction that $u = u'$.
By maximality of $Q$, we have $w^+ \in Q$.
Hence, by Lemma~\ref{lem:pivot-properties}(\ref{aux:component-of-GD}) applied to $Q$, we have $v^- = w^+$, which is a contradiction with $v^- \notin D'$. Hence, we have $u \ne u'$.

Now suppose for contradiction that $u' \ne v^-$.
Let $c' \in [p+1, k]$ be a color different from $c$, and let $w := \rho'(c')$.
Observe that $w \in K \cap D' \seq Q \cap D'$; let $x \in N(u) \cap N[w] \cap Q$ and $x' \in N(u') \cap N[w] \cap D'$.
Now $(u, v^-, v^+, w^+, u', x', x, w)$ is a closed walk\footnote{Assuming that each vertex has a loop because it is possible that, for example, $x = x'$.}.
Since $u, v^-, v^+, u'$, and $w$ are pairwise distinct and $v^- \ne w^+$, there is a cycle of length at most 8 in $G[\{u, v^-, v^+, w^+, u', x', x, w\}]$ containing $v^+ \notin S$, which is a contradiction with Observation~\ref{obs:short-cycles-are-in-S}.
Hence, we know that $u' = v^-$. The remainder of the proof is visualized by Figure~\ref{fig:nsubseteqQ2}.

\begin{figure}[h]
\begin{tikzpicture}
\tikzmath{ \dist = 20; \dis=14;}
\begin{scope}[every node/.style={draw, circle, minimum width=17pt, inner sep=2pt, fill=gray!25!white}]
\node[fill=white] (u) {$u$};
\node[below left = \dist pt of u, fill=white] (u') {$u'$};
\node[below left = \dist pt of u'] (vp) {$v^+$};
\node[below right = \dist pt of u', minimum width = 10pt] (a) {};
\node[right = \dist pt of a, rectangle, minimum height = 17pt] (w) {$w$};
\node[right = \dist pt of u] (w') {$w'$};
\end{scope}
\node[right= 0pt of u'] {$= v^-$};
\draw (w')--(u)--(u')--(vp);
\draw (u')--(a)--(w);
\node[below = 0pt of w] {$d$};
\node[right = 0pt of w'] {$d$};
\node[right = 0pt of vp] {$c$};
\end{tikzpicture}
\centering
\caption{An illustration of the proof of Lemma~\ref{lem:nsubseteqQ2}. The gray vertices are in $S_{\rho'}$ (we have $u \notin S_{\rho'}$ because $u \notin D'$ by maximality of $D'$). The only vertex in $S$ is $w$.}
\label{fig:nsubseteqQ2}
\end{figure}
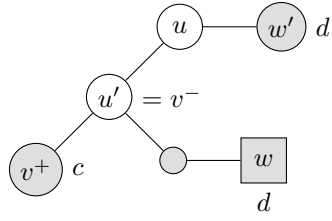

Since $u'$ is a $(u, Q)$-link, we know that $u'$ is $\chi$-tight, which means that it has degree $k-1$ because $u' \notin S$.
Let $d \in [k]$ be a color such that $d \notin \chi_{\rho'}(N(u'))$.
Since $D'$ is $\rho'$-pivoted, there is a vertex $w \in D' \setminus N(u')$ such that $\chi_{\rho'}(w) = d$.
Since $u \notin D'$, we have $w \notin Q$; otherwise there would be a cycle of length at most 5 containing $u \notin S$, which is again a contradiction with Observation~\ref{obs:short-cycles-are-in-S}.
Since $K \seq Q \cup \{u\}$ and $w \ne u$, we have $w \in S$.
Observe that $D \cap S = \emptyset$; otherwise there would be an outer path of length at most 3.
Since $D$ is $\rho$-pivoted, we have $d > b$ (otherwise $d \in \chi_\rho(D)$ would be impossible).
By choice of $d$, we have $w' := \rho'(d) \notin N(u')$.
Since $w' \in K \setminus N(u')$, we have $w' \ne u$ and $w' \in Q$.
Again by Observation~\ref{obs:short-cycles-are-in-S}, we obtain $w' \notin Q'$, which concludes the proof.
\end{proof}

\subsubsection{Avoiding the emergence of a different pivot}\label{subsub:different-pivot}

It may happen that the new color realization $\rho'$ has a different pivot than the original realization $\rho$.
The following lemma describes such a situation (assuming that $\rho'$ is created by only a few modifications).

\begin{lemma}\label{lem:two-pivots}
Let $\rho$ and $\rho'$ be color realizations such that $|\rho \setminus \rho'| \le 2$, let $u$ be a $\rho$-pivot, let $D$ be a maximal set $\rho$-pivoted by $u$, let $u'$ be a $\rho'$-pivot, and let $D'$ be a maximal set $\rho'$-pivoted by $u'$.
Suppose that $u$ is not $(u', D')$-link and $u'$ is not a $(u, D)$-link.
\begin{enumerate}[a)]
\item There is a vertex $w \in D \cap D'$ such that $w$ is a $(u, D)$-link, as well as a $(u', D')$-link. Moreover, $D \cap D' \seq N[w]$.\label{aux:w-exists}
\item If there is a vertex $v \in S \cap (D \symdif D')$, then $S \cap (D \cup D') = \{v\}$.\label{aux:unique-in-S}
\end{enumerate}
\end{lemma}

\begin{proof}
Let $\rho, \rho', u, u', D$, and $D'$ be as in the statement. For~(\ref{aux:w-exists}), let us choose $c_1, c_2 \in [p+1, k]$ so that $v_i := \rho(c_i) = \rho'(c_i)$ for $i \in [2]$; note that $c_1$ and $c_2$ exist because $k \ge p + 4$ and $|\rho \setminus \rho'| \le 2$.
Observe that $v_1, v_2 \in D \cap D' \cap \overline{S}$.
For $i \in [2]$, let $w_i$ (resp. $w_i'$) be the unique vertex in $N(u) \cap N[v_i] \cap D$ (resp. $N(u') \cap N[v_i] \cap D'$), and let $W = \{w_1, w_2, w_1', w_2'\}$.
Since $w_1, w_2 \in D$, we have $u \notin \{w_1, w_2\}$.
If $u = w_1'$, then we would have $w_1' \ne v_1$ and $u$ would be $(u', D')$-link, which is not the case.
Hence, $u \ne w_1'$, and analogously, we have $u \ne w_2'$.
Hence, $u \notin W$, and analogously, we have $u' \notin W$.
If $|W| > 1$, then there would be a cycle in 
$G[W \cup \{u, u', v_1, v_2\}]$, which would contradict Observation~\ref{obs:short-cycles-are-in-S} since 
$u, u', v_1, v_2 \notin S$.
Hence, let $w$ be the unique vertex in $W$, see Figure~\ref{fig:two-pivots}. Note that it is possible that $w \in \{v_1, v_2\}$.
Without loss of generality, $w \ne v_1$.
Since $v_1 \in D \cap D'$, $w$ is a $(u, D)$-link and a $(u', D')$-link as required.
Finally, $D \cap D' \seq N[w]$ because otherwise there would again be a short cycle not contained in $S$. 

\begin{figure}[h]
\begin{tikzpicture}[every node/.style={draw, circle, minimum width=17pt, inner sep=2pt, fill=gray!25!white}]
\tikzmath{ \dist = 20;}
\node[fill=white] (u) {$u$};
\node[below right= \dist pt of u] (w) {$w$};
\node[fill=white, above right= \dist pt of w] (u') {$u'$};
\node[below right= \dist pt of w] (v2) {$v_2$};
\node[below left = \dist pt of w] (v1) {$v_1$};
\node[rectangle, minimum height=17pt, below left = \dist pt of u] (v) {$v$};
\node[rectangle, minimum height=17pt, below right = \dist pt of u'] (v') {$v'$};
\draw (v)--(u)--(w)--(u')--(v');
\draw (v1)--(w)--(v2);
\end{tikzpicture}
\centering
\caption{An illustration of the proof of Lemma~\ref{lem:two-pivots}. The gray vertices are in $S_\rho$ (or, equivalently, in $S_{\rho'}$). The only vertices in $S$ are $v$ and $v'$. The depicted case is impossible because $\{u, w, u'\}$ induces an outer path.}
\label{fig:two-pivots}
\end{figure}
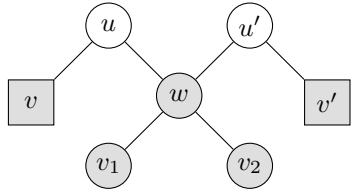

For~(\ref{aux:unique-in-S}), assume that, without loss of generality, $v$ is a vertex in $S \cap (D \setminus D')$.
Observe that $v \notin N[w]$; otherwise we would have $v \in D'$ by maximality of $D'$.
Suppose for contradiction that there is a vertex $v' \in S \cap (D \cup D')$ such that $v' \ne v$;
if possible, choose $v'$ so that $v' \in D'$.
First, suppose that $v' \in D'$, and let $P$ be a $v$-$v'$ path in $G[D \cup D' \cup \{u, u'\}]$; note that $P$ contains $u$ because $v \notin N[w]$.
Now $P' := P - \{v, v'\}$ is an outer path of length at most 4, which contradicts Definition~\ref{def:fen-core}.
Second, suppose that $v' \in D$.
By Lemma~\ref{lem:pivot-properties}(\ref{aux:link-in-S}), there is a $(u, D)$-link $x \in S$.
By Definition~\ref{def:links}, we have $x \in \kb$, which implies $x \in B$.
Let $c = \chi(x)$ and let $x' \in D'$ be a vertex such that $\chi_{\rho'}(x') = c$.
Since $c \le b$, we have $x' \in S$.
This contradicts our choice of~$v'$.
\end{proof}

Using Lemma~\ref{lem:two-pivots}, we can show that a different pivot cannot appear if $\rho'$ is obtained by only ``swapping'' two colors, i.e., $\rho' = \rho[c_1 \mapsto \rho(c_2), c_2 \mapsto \rho(c_1)]$ for some colors $c_1, c_2 \in [b+1, k]$.
Note that in this case, we know that $u$ is not $(u', D')$-link and $u'$ is not a $(u, D)$-link because $u, u' \notin S_\rho = S_{\rho'}$, which enables us to use Lemma~\ref{lem:two-pivots}.

\begin{lemma}\label{lem:not-different-pivot-swap}
If $\rho$ is a color plan, $u \in \overline{S_\rho}$ is a $\rho$-pivot, $\rho'$ is a color plan such that $\range(\rho) = \range(\rho')$ and $|\rho \setminus \rho'| = 2$, and $u' \in \overline{S_\rho}$ is a $\rho'$-pivot, then $u = u'$.
\end{lemma}
\begin{proof}
Let $\rho, \rho', u, u'$ be as in the statement, and suppose that $u \ne u'$.
Let $D \seq S_\rho$ be a maximal set $\rho$-pivoted by $u$, and let $D' \seq S_\rho$ be a maximal set $\rho'$-pivoted by $u'$.
By Lemma~\ref{lem:two-pivots}(\ref{aux:w-exists}), there is a vertex $w \in D \cap D'$ such that $D \cap D' \seq N[w]$ and $w$ is a $(u, D)$-link.
Since $w$ is $\chi_\rho$-tight and $u, u' \in N(w) \setminus S_\rho$, we have $|\chi_\rho(N[w])| \le k - 2$.
Hence, there are two distinct colors $c_1, c_2 \in [k]$ such that $c_1, c_2 \notin \chi_\rho(N[w])$.
For $i \in [2]$, let $x_i \in D \cap \chi_\rho^{-1}(c_i)$ and $x_i' \in D' \cap \chi_{\rho'}^{-1}(c_i)$; these vertices exist since $D$ is $\rho$-pivoted and $D'$ is $\rho'$-pivoted.
Since $x_1, x_2 \notin N[w]$, we have $x_1 \ne x_1'$ and $x_2 \ne x_2'$.
First, suppose that $\chi_\rho(x_1') = c_1$ and $\chi_\rho(x_2') = c_2$.
Clearly, $S \cap \{x_i, x_i'\} \ne \emptyset$ for each $i \in [2]$.
Since $c_1, c_2 \notin \chi_\rho(N[w])$, we have $x_1, x_2, x_1', x_2' \in D \symdif D'$, which yields a contradiction with Lemma~\ref{lem:two-pivots}(\ref{aux:unique-in-S}).

Second, suppose that, without loss of generality, $c_3 := \chi_\rho(x_1') \ne c_1$.
Since $|\rho \setminus \rho'| = 2$, we have $\rho' = \rho[c_1 \mapsto x_1', c_3 \mapsto x_1]$, which implies $\chi_\rho(x_2') = c_2$.
As in the previous case, we have $S \cap \{x_2, x_2'\} \ne \emptyset$ and $x_2, x_2' \in D \symdif D'$.
If $c_2 = c_3$, then $x_2, x_2' \in S$ because $\rho(c_3) = x_1'$; a contradiction with Lemma~\ref{lem:two-pivots}(\ref{aux:unique-in-S}).
Suppose that $c_2 \ne c_3$.
Since $D$ is $\rho$-pivoted, there is a vertex $x_3 \in D$ such that $\chi_\rho(x_3) = c_3$.
Since $\rho(c_3) = x_1'$, we have $x_3 \in S$; again a contradiction with Lemma~\ref{lem:two-pivots}(\ref{aux:unique-in-S}) because $x_3 \notin \{x_2, x_2'\}$.
\end{proof}

Now we move on to the situation in which $\rho'$ is obtained by ``shifting'' a color to a new candidate, i.e., $\rho' = \rho[c \mapsto v]$ for some color $c \in [b+1, k]$ and $v \in K \setminus \range(\rho)$.
This situation will be more complicated than the ``swap'' of two colors handled by Lemma~\ref{lem:not-different-pivot-swap}.
To ensure that $\rho'$ is pivot-free, we must be very careful in choosing the color $c$.
The following lemma finds a color $c$ that will work.
Observe that the three conditions---(\ref{gc:not-link}), (\ref{gc:connect}), and~(\ref{gc:D*})---ensure that neither of the three cases depicted in Figure~\ref{fig:different-pivot} occurs.
Namely,~(\ref{gc:not-link}) corresponds to the left-hand case,~(\ref{gc:connect}) corresponds to the middle case, and~(\ref{gc:D*}) corresponds to the right-hand case.

\begin{lemma}\label{lem:good-color}
If $\rho$ is a color plan, $u \in \overline{S_\rho}$ is a $\rho$-pivot, and $D\seq S_\rho$ is a maximal set $\rho$-pivoted by $u$, then there is a color $c \in [p+1, k]$ satisfying the following three properties:
\begin{enumerate}[(a)]
\item $v^- := \rho(c)$ is not a $(u, D)$-link;\label{gc:not-link}
\item $N(v^-) \cap \overline{D \cup \{u\}} \cap N[S_\rho \setminus D] = \emptyset$; and\label{gc:connect}
\item if there is a set $D^* \seq D$ such that $|D \setminus D^*| \le 2$ and $\dist(x, y) \le 2$ for every $x,y \in D^*$, then $v^- \in D^*$.\label{gc:D*}
\end{enumerate}
\end{lemma}
\begin{proof}
Let $\rho$, $u$, and $D$ be as in the statement.
Let $d \in [b+1, p]$ be a color.
If $\rho(d) \in D$, then let $w_d$ be the unique vertex in $N(u) \cap N[\rho(d)] \cap D$.
Otherwise, if there is a vertex $w \in D$ such that $\dist(\rho(d), w) \le 2$, then we set $w_d := w$, and if there is no such vertex $w$, then we set $w_d := \rho(d)$.
Note that if $1 \le \dist(\rho(d), D) \le 2$, then the vertex $w_d$ is uniquely determined because otherwise, there would be a cycle not contained in $S$ of length at most 8, which would contradict Observation~\ref{obs:short-cycles-are-in-S}.

Let $W_0 = \{w_d \sep d \in [b+1, p]\}$, let $W_1 = \{w \in D \sep N(w) \cap \overline{D\cup\{u\}} \cap N[S] \ne \emptyset\}$ and $W_2 = \{w \in D \cap N(S\cap D) \sep w$ is a $(u, D)$-link$\}$.
By Lemma~\ref{lem:pivot-properties}(\ref{aux:outer-S}), if $w \in W_1$, then $W_1 = \{w\}$ and $S \cap D = \emptyset$, which implies $W_2 = \emptyset$.
Moreover, $|W_2| \le 1$ because all vertices in $W_2$ are $S$-influenced, there is a connected component $C$ of $G[D]$ containing all $S$-influenced vertices by Lemma~\ref{lem:pivot-properties}(\ref{aux:contains-all-S-influenced}), and there is at most one $(u, D)$-link in $C$.
Hence, $|W_1 \cup W_2| \le 1$.
If there is a set $D^* \seq D$ as in~(\ref{gc:D*}), then let $W_3 := D \setminus D^*$, and otherwise, let $W_3 := \emptyset$.
Let $W_4 = W_0\cup W_1\cup W_2\cup W_3$, let $W_5 = \{w \in D \cap N(W_4) \sep w$ is a $(u, D)$-link$\}$, and let $W_6 = W_4 \cup W_5$.
Since $|W_4| \le p +1+2$ and there is at most one $(u, D)$-link in $N(w)$ for each $w \in W_4$, we have $|W_6| \le 2p + 6$.

Since $k \ge 3p +7$, there is a color $c \in [p+1, k]$ such that $w := \rho(c) \notin W_6$; let us choose $c$ so that, if possible, $w \notin N(u)$.
We will show that $c$ is the desired color.
First, suppose that $w$ is a $(u, D)$-link, i.e., there is a vertex $x \in D \cap N(w)$.
If $x \in S$, then we would have $w \in W_2$, and if $x = \rho(d)$ for $d \in [b+1, p]$, then we would have $w = w_d \in W_0$.
Hence, $x = \rho(d)$ for some color $d \in [p+1, k]$.
If $x \in W_4$, then we would have $w \in W_5$, which means that $x \notin W_4$.
Since $x$ is not a $(u, D)$-link, we have $x \notin W_5$, which contradicts the choice of $c$ (we would choose $d$ instead of $c$). Hence, $c$ satisfies~(\ref{gc:not-link}).

Second, suppose that there is a vertex $x \in N(w) \cap \overline{D \cup \{u\}} \cap N[S_\rho \setminus D]$.
If $x \in N[S]$, then we would have $w \in W_1$, which means that there is a color $d \in [b+1, k]$ such that $\rho(d) := y \in N[x] \setminus D$.
Since $D$ is $\rho$-pivoted, we have $d \le p$.
Since $\dist(y, w) \le 2$, we have $w = w_d \in W_0$, which is a contradiction.
Hence, $c$ satisfies~(\ref{gc:connect}).
Finally, $c$ satisfies~(\ref{gc:D*}) because otherwise we would have $w \in W_3$, which is not the case.
\end{proof}

We would like to again use Lemma~\ref{lem:two-pivots}, which requires that $u$ is not a $(u', D')$-link and $u'$ is not a $(u, D)$-link.
In the following lemma, we handle the case when the latter condition is satisfied (using the color found in Lemma~\ref{lem:good-color}).

\begin{lemma}\label{lem:not-different-pivot-shift}
Let $\rho$ be a color plan, $u \in \overline{S_\rho}$ be a $\rho$-pivot, let $D\seq S_\rho$ be a maximal set $\rho$-pivoted by $u$, let $c \in [p+1, k]$ be a color satisfying the properties~(\ref{gc:not-link}),~(\ref{gc:connect}), and~(\ref{gc:D*}) from Lemma~\ref{lem:good-color},
let $v^+ \in K \setminus \range(\rho)$, let $\rho' = \rho[c \mapsto v^+]$,
let $u' \in \overline{S_{\rho'}}$ be a $\rho'$-pivot, and let $D'\seq S_{\rho'}$ be a maximal set $\rho'$-pivoted by $u'$.
If $u$ is not a $(u', D')$-link, then $u' = u$.
\end{lemma}
\begin{proof}
Let $\rho, u, D, c, v^+, \rho', u'$, and $D'$ be as in the statement, let $v^- := \rho(c)$, and suppose for contradiction that $u \ne u'$.
If $u' \notin D$, then $u'$ is not a $(u, D)$-link, and if $u' \in D$, then $u' \in S_\rho$ and $u' = v^-$, which means that $u'$ is not a $(u, D)$-link by property~(\ref{gc:not-link}) from Lemma~\ref{lem:good-color}.
Hence, by Lemma~\ref{lem:two-pivots}(\ref{aux:w-exists}), there is a vertex $w \in D \cap D'$ such that $D \cap D' \seq N[w]$ and $w$ is a $(u, D)$-link, as well as a $(u', D')$-link.
Let $\{c_1, \ldots, c_q\} := [k] \setminus \chi_\rho(N[w])$.
Since $u \in N(w) \setminus S_\rho$ and $w$ is $\chi_\rho$-tight, we have $q \ge 1$.
Let us fix $i \in [q]$ and observe that there are vertices $x_i \in D \cap \chi_\rho^{-1}(c_i)$ and $x_i' \in D' \cap \chi_{\rho'}^{-1}(c_i)$ because $D$ is $\rho$-pivoted and $D'$ is $\rho'$-pivoted.
Since $x_i \notin N[w]$, we have $x_i \ne x_i'$.
Observe that if $x_i' \ne v^+$, then $\chi_\rho(x_i') = c_i$, and $x_i \in S$ or $x_i' \in S$.

First, suppose that $q = 1$.
If $u' \ne v^-$, then $u, u' \in N(w) \setminus S_\rho$, which is a contradiction since $w$ is $\chi_\rho$-tight and $|\chi_\rho(N[w])| = k-q = k-1$.
Hence, $u' = v^-$.
Observe that $c \in \chi_\rho(N[w])$, which implies $c \ne c_1$ and $x_1' \ne v^+$.
Hence, $\chi_\rho(x_1') = c_1$, $x_1' \in S_\rho$, and $x_1' \notin N[w]$ by definition of $c_1$.
Since $x_1' \in D'$, we have $x_1' \notin D$ because $D \cap D' \seq N[w]$.
Hence, either $x_1' \in N(u')$ or there is a vertex in $N(u') \cap \overline{D \cup \{u\}} \cap N(x_1')$, which is a contradiction with property~(\ref{gc:connect}) from Lemma~\ref{lem:good-color}.

Second, suppose that $q \ge 2$, which means that there is $i \in [q]$ such that $\chi_\rho(x_i') = c_i$ (because $x_i' \ne v^+$).
If $x_i' \in N[w] \cap S$, then we would have $c_i = \chi_{\rho'}(x_i') = \chi(x_i') = \chi_\rho(x_i') \in \chi_\rho(N[w])$, which is a contradiction.
Hence, there is a vertex $x \in \{x_i, x_i'\}$, such that $x \in S \setminus N[w]$.
Since $D \cap D' \seq N[w]$, we may use Lemma~\ref{lem:two-pivots}(\ref{aux:unique-in-S}) to deduce that $|(D \cup D') \cap S| = 1$.
This means that $q = 2$ because if $x_i, x_i' \notin S$, then $x_i' = v^+$.
Without loss of generality, $x_1 = v^-$, $x_1' = v^+$, and $x_2 \in S$ or $x_2' \in S$.
Let $D^* = N[w] \cap D$ and observe that $D \setminus D^* = \{x_1, x_2\}$ because $D \cap S \seq \{x_2\}$ and $\rho(d) \in N[w]$ for every color $d \in [b+1, k] \setminus \{c_1, c_2\}$.
Hence, by property~(\ref{gc:D*}) from Lemma~\ref{lem:good-color}, we have $v^- \in D^*$, which is a contradiction with $v^- = x_1$.
\end{proof}

Now we handle the case when $u$ \emph{is} a $(u', D')$-link, again using the color found in Lemma~\ref{lem:good-color}.
Note that this can happen only if $u$ is the ``new candidate'' (which is called $v^+$ in Lemma~\ref{lem:not-different-pivot-shift}) because $u \notin S_\rho$ and $u \in D' \seq S_{\rho'}$.
Informally, the conclusion of the following lemma states that there is a candidate $u^*$ outside of $D$.
This candidate will be useful because instead of ``shifting'' a color to $u$, we can ``swap'' the color of $u^*$ with a color in $\rho^{-1}(D)$, see Lemma~\ref{lem:not-different-pivot-swap}.

\begin{lemma}\label{lem:shift-to-u}
Let $\rho$ be a color plan, $u \in \overline{S_\rho}$ be a $\rho$-pivot, let $D\seq S_\rho$ be a set $\rho$-pivoted by $u$, let $c \in [p+1, k]$ be a color satisfying the properties~(\ref{gc:not-link}),~(\ref{gc:connect}), and~(\ref{gc:D*}) from Lemma~\ref{lem:good-color},
let $\rho' = \rho[c \mapsto u]$, let $u'$ be a $\rho'$-pivot, and let $D'$ be a set $\rho'$-pivoted by $u'$.
If $u$ is a $(u', D')$-link, then $\range(\rho) \nsubseteq D$.
\end{lemma}
\begin{proof}
Let $\rho, u, D, c, \rho', u'$, and $D'$ be as in the statement, and suppose for contradiction that $\range(\rho) \seq D$.
Let $v := \rho(c)$ and suppose for contradiction that $u' = v$.
By Lemma~\ref{lem:rho-pivot-properties}(\ref{aux:pivot-has-two-neighbors}), there is a vertex $w \in D' \cap N(u')$ such that $w \ne u$.
If $w \in D$, then $u'$ is a $(u, D)$-link, which is a contradiction with property~(\ref{gc:not-link}) from Lemma~\ref{lem:good-color}.
Hence, $w \in N(v) \cap \overline{D \cup \{u\}} \cap S_\rho$, which is a contradiction with property~(\ref{gc:connect}) from Lemma~\ref{lem:good-color}.
Hence, we obtain $u' \ne v$, which means that $u' \notin S_\rho$.

Let $\{c_1, \ldots, c_q\} := [k] \setminus \chi_{\rho}(N(u))$.
Since $u$ is $\chi_{\rho'}$-tight, we know that $u$ has degree $k-1$, which means that $q \ge 2$ because $u' \in N(u) \setminus S_\rho$.
For $i \in [q]$, let $x_i$ be a vertex in $D' \cap \chi_{\rho'}^{-1}(c_i)$; $x_i$ exists because $D'$ is $\rho'$-pivoted.
Let $i \in [q]$ be such that $x_i \ne u$ and observe that $\chi_\rho(x_i) = c_i$, which implies $x_i \notin N[u]$.
Hence, $x_i$ is in different connected component of $G[D']$ than $u$.
Since $u' \notin D$, we obtain $x_i \notin D$, which implies $x_i \in S$ because $\range(\rho) \seq D$.
Observe that now we have $D \cap S = \emptyset$; otherwise, there would be an outer path of length at most 3.

Let $X = \{x_i \sep i \in [q]\}$, and, without loss of generality, assume that $x_1 \ne u$.
If $|X \setminus \{u\}| = 1$, then $q = 2$, $u = x_2$, and since $c_2 \notin \chi_\rho(N(u))$, we have $v \notin N(u)$.
Since $D \cap S = \emptyset$, we obtain $D \setminus N(u) = \{\rho(c_1), \rho(c_2)\}$.
By property~(\ref{gc:D*}) from Lemma~\ref{lem:good-color}, we have $v \in N(u)$, which is a contradiction.
Hence, we may assume, without loss of generality, that $x_2 \ne u$.
Using the argument in the previous paragraph for $i \in [2]$, we obtain $x_1, x_2 \in S \setminus D$.
By Lemma~\ref{lem:pivot-properties}(\ref{aux:link-in-S}), there is a vertex $y \in D' \cap S$ that is a $(u', D')$-link.
By Definition~\ref{def:links}, we have $y \in K^+$, which implies $y \in B$.
Since $y \in B \setminus D$, we have $\chi(y) \notin \chi_\rho(D)$, which is a contradiction since $D$ is $\rho$-pivoted.
\end{proof}

\subsubsection{The case when $u$ is not a $(\chi, D)$-pivot}\label{subsub:not-pivot}

In this section, we handle the case when $u$ is not a $(\chi, D)$-pivot.
Informally, it suffices to swap two colors in $\rho$.
Let us remark that a different pivot cannot emerge by Lemma~\ref{lem:not-different-pivot-swap}.

\begin{lemma}\label{lem:not-pivot}
Let $\rho$ be a damage-free color realization of a valid color plan $\pi$, let $u$ be a $\rho$-pivot, and let $D \seq S_\rho$ be a maximal set $\rho$-pivoted by $u$.
If $u$ is not a $(\chi, D)$-pivot, then a damage-free and pivot-free color realization $\rho'$ of a valid color plan can be computed in polynomial time.
Moreover, if $\rho$ is almost $\ell$-safe for some $\ell \in \bb N$, then $\rho'$ is almost $(\ell+2)$-safe.
\end{lemma}

\begin{proof}
Let $\rho, \pi, u$, and $D$ be as in the statement.
Since $u$ is a $(\chi_\rho, D)$-pivot but not a $(\chi, D)$-pivot, there must be a $(u, D)$-link $v \in D$ that is $\chi_\rho$-tight but not $\chi$-tight.
This means that there is a color $c \in [b+1, p]$ such that $w := \rho(c) \in N[v]$ and $c \in \chi(N[v])$; let $w' \in N[v] \cap \chi^{-1}(c)$.
By Lemma~\ref{lem:valid-realization}(\ref{vraux:proper}), $\chi_\rho$ is proper, which implies $v \notin \{w, w'\}$.
Let $d \in [k]$ be a color such that $d \notin \chi_\rho(N[v])$; such a color exists by Lemma~\ref{lem:rho-pivot-properties}(\ref{aux:missing-color}).
Let $x \in D$ be such that $\chi_\rho(x) = d$.
If $x \in N[S]$, then $x \in N[v]$ by Lemma~\ref{lem:pivot-properties}(\ref{aux:contains-all-S-influenced}) because $v$ and $x$ are both $S$-influenced, which contradicts the choice of $d$.
Hence, $x = \rho(d) \notin N[v]$.
Let us define $\rho' = \rho[c \mapsto x, d \mapsto w]$, see Figure~\ref{fig:not-pivot}.

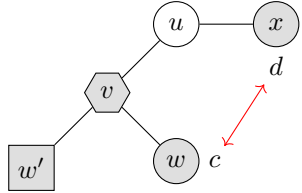
\begin{figure}[h]
\begin{tikzpicture}
\tikzmath{ \dist = 20;}
\begin{scope}[every node/.style={draw, circle, minimum width=17pt, inner sep=2pt, fill=gray!25!white}]
\node[fill= white] (u) {$u$};
\node [below left =  \dist pt of u, regular polygon,
    regular polygon sides=6] (v) {$v$};
\node [below left =  \dist pt of v, rectangle, minimum height=17pt] (w') {$w'$};
\node [below right =  \dist pt of v] (w) {$w$};
\node [right =  \dist pt of u] (x) {$x$};
\draw (x)--(u)--(v)--(w');
\draw (v)--(w);
\end{scope}
\node[right= 0pt of w] (c) {$c$};
\node[below= 0pt of x] (d) {$d$};
\draw[red, <->] (c)--(d);
\end{tikzpicture}
\centering
\caption{An illustration of the proof of Lemma~\ref{lem:not-pivot}. The gray vertices are in $S_\rho$ (or, equivalently, in $S_{\rho'}$). Note that $w' \in S$, that $w, x \notin S$, and that $v \in S$ and $v \notin S$ are both possible. The modification leading to $\rho'$ is depicted by the double-arrow.}
\label{fig:not-pivot}
\end{figure}

First, suppose that $y \in \range(\rho') \cap N(S)$ is damaged in $\rho'$, i.e., there is a color $c' \in [b+1, p]$ such that $c' \in \chi(N(y))$ and $z := \rho'(c') \in N(y)$.
By Lemma~\ref{lem:damaged}(\ref{daux:same-ranges}), we have $z \in \{w, x\}$ and $c' \in \{c, d\}$.
Suppose that $y \in D$.
Since $y$ is $S$-influenced, we have $y \in N[v]$ by Lemma~\ref{lem:pivot-properties}(\ref{aux:contains-all-S-influenced}). 
If $z = x$, then $x$ and $v$ would be in the same component of $G[D]$ and we would have $x \in N[v]$ by Lemma~\ref{lem:pivot-properties}(\ref{aux:component-of-GD}), which is not the case.
Hence, we have $z = w$ and $c' = d$, which implies that $y = v$ by Lemma~\ref{lem:pivot-properties}(\ref{aux:unique-neighbor}) because $w \notin S$.
However, now we have a contradiction with $d \notin \chi(N[v]) \seq \chi_\rho(N[v])$.
On the other hand, if $y \notin D$, then by Lemma~\ref{lem:pivot-properties}(\ref{aux:outer-S}), $D \cap S = \emptyset$, which is a contradiction since $w' \in D \cap S$ (by maximality of $D$).
We have shown that $\rho'$ is damage-free.

Second, let us show that $\rho'$ is a color realization of a valid color plan.
Recall that $x \notin N(S)$.
If $v \notin S$ or $d \in [p+1, k]$, then $\rho'$ realizes $\pi$ by Lemma~\ref{lem:is-valid-plan}(\ref{vpaux:trivial}), and if $v \in S$ and $d \in [b+1, p]$, then $\rho'$ realizes a valid color plan $\pi[d \mapsto v]$ by Lemma~\ref{lem:is-valid-plan}(\ref{vpaux:mapsto-u}). 

Third, suppose that there is a $\rho'$-pivot $u' \in \overline{S_{\rho'}}$ and let $D' \seq S_{\rho'}$ be a set $\rho'$-pivoted by $u'$.
By Lemma~\ref{lem:not-different-pivot-swap}, we have $u' = u$.
Let $y \in D'$ be the vertex such that $\chi_{\rho'}(y) = d$.
If $y \in S$, then by Lemma~\ref{lem:pivot-properties}(\ref{aux:contains-all-S-influenced}), we have $y \in N[v]$ (and $y \ne v$ because $\chi_{\rho'}(v) = \chi_\rho(v) \ne d$).
On the other hand, if $y \notin S$, then $y = \rho'(d) = w \in N(v)$.
In both cases, we have $y \in N(v)$.
However, $v$ is not $\chi_{\rho'}$-tight by construction of $\rho'$, and $v \in N(u) \cap N(y)$, which implies $y \notin D'$, a contradiction.

We have shown that $\rho'$ is damage-free and pivot-free color realization of a valid color plan.
Moreover, if $\rho$ is almost $\ell$-safe for some $\ell \in \bb N$, then $\rho'$ is almost $(\ell+2)$-safe by Observation~\ref{obs:safety}.
\end{proof}

\subsubsection{Operations preserving damage-freeness}\label{subsub:operations}

It will be useful to have a candidate for a color in $[p+1, k]$ that is \emph{not} $S$-influenced at our disposal.
The following lemma says that such a vertex can be found by slightly modifying $\rho$ (without creating a damaged vertex).

\begin{lemma}\label{lem:first-swap}
Let $\rho$ be a damage-free color realization of a valid color plan $\pi$, let $u$ be a $\rho$-pivot, and let $D \seq S_\rho$ be a maximal set $\rho$-pivoted by $u$.
If $u$ is a $(\chi, D)$-pivot, then we can, in polynomial time, compute a damage-free color realization $\rho'$ of a valid color plan such that $D$ is $\rho'$-pivoted by $u$ and there is a color $c \in [p+1, k]$ such that that $\rho'(c)$ is not $S$-influenced.
Moreover, if $\rho$ is $\ell$-safe, then $\rho'$ is $(\ell+2)$-safe, and if $\rho$ is almost $\ell$-safe, the $\rho'$ is almost $(\ell+2)$-safe.
\end{lemma}
\begin{proof}
Let $\rho, \pi$, $u$, and $D$ be as in the statement, and let $C$ be a connected component of $G[D]$ containing all $S$-influenced vertices; $C$ exists by Lemma~\ref{lem:pivot-properties}(\ref{aux:contains-all-S-influenced}).
We may assume that $\rho(c)$ is $S$-influenced for each $c \in [p+1, k]$; otherwise we may set $\rho' := \rho$.
Let $v_C$ be the unique vertex in $C \cap N(u)$.
By Lemma~\ref{lem:rho-pivot-properties}(\ref{aux:missing-color}), there is a color $c \in [k]$ such that $c \notin \chi_\rho(N[v_C]) = \chi_\rho(C)$.
Let $v \in D$ be such that $\chi_\rho(v) = c$; $v$ exists because $D$ is $\rho$-pivoted.
Since $v \notin C$, we know that $c \le p$ and $v$ is not $S$-influenced, which implies $v \notin S$, $c > b$, and $\rho(c) = v$.
Let $d \in [p+1, k]$ be such that $w := \rho(d) \ne v_C$ and let us define $\rho' = \rho[d \mapsto v, c \mapsto w]$.

First, suppose that $x \in \range(\rho')$ is damaged in $\rho'$, i.e., there is a color $c' \in [b+1, p]$ such that $c' \in \chi(N(x))$ and $y := \rho'(c') \in N(x)$.
By Lemma~\ref{lem:damaged}(\ref{daux:same-ranges}), we have $y \in \{v, w\}$, and since $c' \le p$, we have $y = w$ and $c' = c$.
If $x \notin D$, then $x \in N(w) \cap N(S) \cap \overline{D^+} \cap \overline{S}$ and by Lemma~\ref{lem:pivot-properties}(\ref{aux:outer-S}), $w$ is the only $S$-influenced vertex in $D$, which is a contradiction since $k \ge p+2$.
Hence, we have $x \in D$. 
Since $w \notin S$ and $x \in N(w)$, we may use Lemma~\ref{lem:pivot-properties}(\ref{aux:unique-neighbor}) to obtain $x = v_C$.
However, $c' = c \notin \chi_\rho(N[v_C])$, which contradicts $c' \in \chi(N(x))$.
We have shown that $\rho'$ is damage-free.

Second, let us show that $\rho'$ realizes a valid color plan.
By Lemma~\ref{lem:is-valid-plan}(\ref{vpaux:trivial}), if $w \notin N(S)$, then $\rho'$ realizes $\pi$.
Hence, suppose that $x$ is a vertex in $S \cap N(w)$, and let $\pi' = \pi[c \mapsto x]$.
Observe that $\pi'$ is a color plan realized by $\rho'$.
If $x \notin D$, then by Lemma~\ref{lem:pivot-properties}(\ref{aux:outer-S}), all $S$-influenced vertices are in $N[w]$, which is a contradiction with $k \ge p+3$ (here we use $w \ne v_C$).
Hence, $x \in D$, which implies $x = v_C$.
Since $c \notin \chi(N[v_C])$, $\pi'$ is valid by Lemma~\ref{lem:is-valid-plan}(\ref{vpaux:mapsto-u}).

Third, observe that $D$ is $\rho'$-pivoted by $u$ by Lemma~\ref{lem:rho-pivot-properties}(\ref{aux:rho'-pivoted}).
Fourth, observe that $\rho'$ can be efficiently computed.
Finally, the ``moreover'' part of the statement follows by Observation~\ref{obs:safety}.
\end{proof}

Now we describe when a ``swap'' of two colors does not create a damaged vertex.

\begin{lemma}\label{lem:pre-swap}
Let $\rho$ be a damage-free color realization of a valid color plan $\pi$, let $u$ be a $\rho$-pivot, let $D \seq S_\rho$ be a set $\rho$-pivoted by $u$, and let $c_1, c_2 \in [b+1, k]$ be colors such that for $i \in [2]$, $v_i := \rho(c_i)$ and if $N(v_i) \cap N[S] \cap S_\rho \ne \emptyset$, then $c_{3-i} > p$.
If $\rho' = \rho[c_1 \mapsto v_2, c_2 \mapsto v_1]$, then $\rho'$ is a damage-free color realization of a valid color plan.
Moreover, if $u$ is a $(\chi, D)$-pivot and $v_1, v_2 \in D$, then $D$ is $\rho'$-pivoted by $u$.
\end{lemma}
\begin{proof}
Let $\rho, \pi, u, D, c_1, c_2, v_1, v_2$, and $\rho'$ be as in the statement.
First, suppose that $w \in \range(\rho')$ is damaged in $\rho'$, i.e., there is a color $d \in [b+1, p]$ such that $d \in \chi(N(w))$ and $x := \rho'(d) \in N(w)$.
By Lemma~\ref{lem:damaged}(\ref{daux:same-ranges}), we have $x = v_i$ for some $i \in [2]$.
Since $w \in N(v_i) \cap N(S) \cap \range(\rho') \seq N(v_i) \cap N[S] \cap S_\rho$,
we have $d = c_{3-i} > p$, which is a contradiction.
Hence, $\rho'$ is damage-free.

Second, observe that if $v_1, v_2 \notin N(S)$, then $\rho'$ realizes $\pi$ by Lemma~\ref{lem:is-valid-plan}(\ref{vpaux:trivial}).
If $v_1, v_2 \in N(S)$, then $c_1, c_2 \in [p+1, k]$, and $\rho'$ again realizes $\pi$ by Lemma~\ref{lem:is-valid-plan}(\ref{vpaux:trivial}).
Finally, if, without loss of generality, $v_1 \in N(S)$ and $v_2 \notin N(S)$, then $c_2 > p$, and $\rho'$ realizes a valid color plan $\pi[c_1 \mapsto *]$ by Lemma~\ref{lem:is-valid-plan}(\ref{vpaux:mapsto-*}).

The ``moreover'' part of the statement follows by Lemma~\ref{lem:rho-pivot-properties}(\ref{aux:rho'-pivoted}).
\end{proof}

\subsubsection{Preventing $u$ from being a pivot}\label{subsub:preventing-u}

Now we are ready to show how $\rho$ can be transformed into $\rho'$ so that the $\rho$-pivot $u$ is not a $\rho'$-pivot. 
We start with the easiest case, in which some candidate far from $u$ is actually used by $\rho$.

\begin{lemma}\label{lem:KnseqQ-swap}
Let $\rho$ be a color realization of a valid color plan, let $u$ be a $\rho$-pivot, let $D \seq S_\rho$ be a maximal set $\rho$-pivoted by $u$, and let $Q \seq S \cup K$ be the maximal set such that $u$ is a $(\chi, Q)$-pivot.
If $\rho$ is almost $\ell$-safe, $u$ is a $(\chi, D)$-pivot, and $\range(\rho) \nsubseteq Q$,
then an almost $(\ell+4)$-safe, damage-free, and pivot-free color realization of a valid color plan can be computed in polynomial time.
\end{lemma}
\begin{proof}
Let $\rho, u, D$ and $Q$ be as in the statement.
Let $v \in \range(\rho) \setminus Q$ and let $c = \rho^{-1}(v)$.
By Lemma~\ref{lem:first-swap}, there is a damage-free and almost $(\ell+2)$-safe color realization $\rho^*$ of a valid color plan such that $D$ is $\rho^*$-pivoted by $u$ and there is a color $d \in [p+1, k]$ such that that $w := \rho^*(d)$ is not $S$-influenced (and $\rho^*$ can be computed in polynomial time).
Now we define $\rho' := \rho^*[c \mapsto w, d \mapsto v]$.
By Lemma~\ref{lem:pre-swap}, $\rho'$ is a damage-free color realization of a valid color plan, and by Observation~\ref{obs:safety}, $\rho'$ is almost $(\ell+4)$-safe.
Suppose that $u'$ is a $\rho'$-pivot. By Lemma~\ref{lem:not-different-pivot-swap}, we know that $u = u'$, and by Lemma~\ref{lem:u-is-not-rho'-pivot}, $u \ne u'$; a contradiction.
Hence, $\rho'$ is the desired color realization, and it can be computed in polynomial time.
\end{proof}

In the following lemma, we have two vertices $v_1, v_2 \in \range(\rho)$ such that $uv_1, v_1v_2 \in E(G)$ and a candidate $w \in K$. We show how the color of $v_1$ can be ``shifted'' to $w$, which will be used in the proof of Lemma~\ref{lem:KseqQ}.

\begin{lemma}\label{lem:most-technical}
Let $\rho$ be a damage-free and $4$-safe color realization of a valid color plan $\pi$, let $u$ be a $\rho$-pivot, let $D$ be a set $\rho$-pivoted by $u$, let $Q \seq S \cup K$ be a set such that $u$ is a $(\chi, Q)$-pivot, and let $c_1, c_2 \in [b+1, k]$ be colors such that $v_1 := \rho(c_1) \in D \cap N(u)$ and $v_2 := \rho(c_2) \in D \cap N(v_1)$.
If $D \seq Q$, $K \seq Q \cup \{u\}$ and $w \in K \setminus \range(\rho)$ is a vertex such that $w \ne u$, then there is a damage-free and almost $7$-safe color realization $\rho'$ of a valid color plan such that $v_1 \notin \range(\rho')$, $\rho'(c_2) = v_2$, and $u \notin \range(\rho')$.
Moreover, $\rho'$ can be computed in polynomial time.
\end{lemma}
\begin{proof}
Let $\rho, \pi, u, D, Q, c_1, c_2, v_1, v_2$, and $w$ be as in the statement.
If $w$ is not $S$-influenced, then we set $\rho' := \rho[c_1 \mapsto w]$.
By Lemma~\ref{lem:get-safety}, $\rho'$ is almost $5$-safe.
Suppose that $x \in \range(\rho')$ is damaged in $\rho'$, i.e., $x \in N(S)$ and there is a vertex $y \in N(x) \cap \range(\rho')$.
By Lemma~\ref{lem:damaged}(\ref{daux:shift}), $w \in \{x, y\}$, which is a contradiction because $w$ is not $S$-influenced.
Hence, $\rho'$ is damage-free.
If $c_1 > p$ or $v_1 \notin N(S)$, then $\rho'$ realizes $\pi$ by Lemma~\ref{lem:is-valid-plan}(\ref{vpaux:trivial}), and otherwise, $\rho'$ realizes $\pi[c_1 \mapsto *]$, which is a valid color plan by Lemma~\ref{lem:is-valid-plan}(\ref{vpaux:mapsto-*}).

\begin{figure}[h]
\begin{tikzpicture}
\tikzmath{ \dist = 20;}
\begin{scope}[every node/.style={draw, circle, inner sep=2pt, minimum width=17pt, fill=gray!20!white}]
\node[fill=white] (u) {$u$};
\node[below= \dist pt of u] (z) {$z$};
\node[left= \dist pt of u] (v1) {$v_1$};
\node[below= \dist pt of v1] (v2) {$v_2$};
\node[below right= \dist pt of u, rectangle, minimum height=17pt] (x) {$x$};
\node[below right= \dist pt of x] (y) {$y$};
\node[below= \dist pt of x, fill=white] (w) {$w$};
\end{scope}
\draw (z)--(u)--(v1)--(v2);
\draw (u)--(x)--(y)--(x)--(w);
\node[below=0pt of z, inner sep =2pt] (c) {$c$};
\node[below right=0pt of v1, inner sep =1pt] (c1) {$c_1$};
\draw[->, red] (c1)-- (z);
\draw[->, red] (c)-- (w);
\end{tikzpicture}
\centering
\caption{A depiction of the case when $w$ is $S$-influenced in the proof of Lemma~\ref{lem:most-technical}. The gray vertices are in $S_\rho$ and the only vertex in $S$ is $x$. The transformation leading to $\rho'$ is represented by the red arrows. Note that instead of the depicted situation, it may happen that $x \notin S$ and $x \in N(S)$.}
\label{fig:most-technical}
\end{figure}
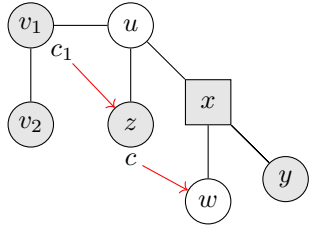

Suppose that $w$ is $S$-influenced and let $x$ be the vertex in $N(w) \cap N[S]$ such that $x \ne u$.
Since $w$ is $4$-safe and $k \ge p+5$, there is a vertex $y \in \rho([p+1, k])$ such that $\dist(y, S) \le \dist(w, S) \le 2$.
In other words, $N(y) \cap N[S] \ne \emptyset$.
If $x \notin D \seq Q$, then by Lemma~\ref{lem:pivot-properties}(\ref{aux:outer-S}) applied to $Q$, $N(y) \cap N[S] = \{w\}$, which is a contradiction with $\dist(y, S) \le \dist(w, S)$.
Hence, we may suppose that $x \in D$.
If $x \notin N(u)$, then $w \in N(x) \cap N(u)$ by Lemma~\ref{lem:pivot-properties}(\ref{aux:unique-neighbor}) applied to $Q$ because $w \notin S$, which is a contradiction since $w \notin D$ and $x \in D$.
Hence, $x \in N(u)$, see Figure~\ref{fig:most-technical}.
Since $u, w \in N(x) \setminus S_\rho$ and $x$ is $\chi_\rho$-tight, there is a color $c \in [b+1, k]$ such that $c \ne c_2$ and $c \notin \chi_\rho(N[x])$; let $z = \rho(c)$.
By Lemma~\ref{lem:pivot-properties}(\ref{aux:contains-all-S-influenced}), $z$ and $v_1$ are not $S$-influenced because they are in different connected components of $G[D]$ than $x$.
Hence, $\rho_1 := \rho[c \mapsto v_1, c_1 \mapsto z]$ is a damage-free color realization of a valid color plan $\pi_1$ by Lemma~\ref{lem:pre-swap}, and by Observation~\ref{obs:safety}, $\rho_1$ is $6$-safe.

Now we are ready to define $\rho' := \rho_1[c \mapsto w]$.
By Lemma~\ref{lem:get-safety}, $\rho'$ is almost $7$-safe.
Suppose that $x' \in \range(\rho')$ is damaged in $\rho'$, i.e., there is a color $d \in [b+1, p]$ such that $x' \in N(\chi^{-1}(d))$ and $y' := \rho'(d) \in N(x')$.
By Lemma~\ref{lem:damaged}(\ref{daux:shift}), $w \in \{x', y'\}$.
If $w = y'$, then $d =c$, $x' = x$, and we obtain a contradiction with the choice of $c$.
Hence, suppose that $w = x'$, which means that $x \in S$ (by choice of $x$).
Observe that $y' \notin Q \cup \{u\}$, which is a contradiction with $\range(\rho') \seq K \seq Q \cup \{u\}$.
Hence, $\rho'$ is damage-free.
If $c > p$ or $w \notin N(S)$, then $\rho'$ realizes $\pi_1$ by Lemma~\ref{lem:is-valid-plan}(\ref{vpaux:trivial}), and otherwise, $\rho'$ realizes $\pi_1[c \mapsto x]$, which is a valid color plan by choice of $c$ and Lemma~\ref{lem:is-valid-plan}(\ref{vpaux:mapsto-u}).

In both cases ($w$ is $S$-influenced or not), we have $v_1 \notin \range(\rho')$, $\rho'(c_2) = v_2$, and $u \notin \range(\rho')$, as desired.
Now we are finished because the proof easily translates into a polynomial time algorithm.
\end{proof}

The following lemma handles the case when all candidates are close to the $\rho$-pivot $u$ (see the assumption $K \seq Q^+$ in the statement of the lemma).
Notice that $v^- \ne v^+$ (because $v^- \in N(v^+)$) but $w \in \{v^-, v^+\}$ is allowed.
However, $w$ is used only when $v^-, v^+ \in \range(\rho)$, in which case $w \notin \{v^-, v^+\}$ since $w \notin \range(\rho)$.
We have made this comment because Lemma~\ref{lem:KnseqQ-shift} will use Lemma~\ref{lem:KseqQ} with $w = v^-$. 

\begin{lemma}\label{lem:KseqQ}
Let $\rho$ be a color realization of a valid color plan, let $u$ be a $\rho$-pivot, let $D \seq S_\rho$ be a maximal set $\rho$-pivoted by $u$, and let $Q \seq S \cup K$ be the maximal set such that $u$ is a $(\chi, Q)$-pivot, and let $Q^+ = Q \cup \{u\}$.
Let $v^-, v^+ \in Q \setminus S$ be two vertices such that $v^- \in N(u) \cap N(v^+)$, and let $w \in K \setminus \range(\rho)$ be a vertex such that $w \ne u$.
If $\rho$ is $0$-safe, $u$ is a $(\chi, D)$-pivot, and $K \seq Q^+$, then then an almost $11$-safe, damage-free, and pivot-free color realization of a valid color plan can be computed in polynomial time.
\end{lemma}
\begin{proof}
Let $\rho, u, D, Q, v^-$, and $v^+$ be as in the statement.
Since $u$ is a $(\chi, D)$-pivot, we have $D \seq Q$ by maximality of $Q$.
By Lemma~\ref{lem:first-swap}, there is a damage-free and $2$-safe color realization $\rho_1$ of a valid color plan such that $D$ is $\rho_1$-pivoted by $u$ and there is a color $c \in [p+1, k]$ such that that $x := \rho_1(c)$ is not $S$-influenced (and $\rho_1$ can be computed in polynomial time).

First, suppose that $v^+ \notin \range(\rho_1)$.
If $v^- = \rho(d)$ for some $d \in [b+1, k]$, then we set $\rho_2 := \rho_1[c \mapsto v^-, d \mapsto x]$, and otherwise we set $\rho_2 := \rho_1$.
By Lemma~\ref{lem:pre-swap}, $\rho_2$ is a damage-free color realization of a valid color plan $\pi^*$, and by Observation~\ref{obs:safety}, $\rho_2$ is $4$-safe.
Let $\rho_3 := \rho_2[c \mapsto v^+]$; observe that $\rho_3$ realizes $\pi^*$ by Lemma~\ref{lem:is-valid-plan}(\ref{vpaux:trivial}).
Suppose that $x \in \range(\rho_3)$ is damaged in $\rho_3$, which means that there is a vertex $y \in N(x) \cap \range(\rho_3)$.
By Lemma~\ref{lem:damaged}(\ref{daux:rhoc=rho'c}), we have $x = v^+$.
Since $\range(\rho_3) \seq K \seq Q^+$, we have $y \in Q^+$, and since $N(v^+) \cap Q^+ = \{v^-\}$, we obtain $y = v^-$, which is a contradiction with $v^- \notin \range(\rho_3)$.
Hence, $\rho_3$ is damage-free, and by Lemma~\ref{lem:get-safety}, $\rho_3$ is almost $5$-safe.

Second, suppose that $v^+ = \rho(d^+)$ for some $d^+ \in [b+1, k]$.
Let $\rho_2 = \rho_1[c \mapsto v^+, d^+ \mapsto x]$.
By Lemma~\ref{lem:pre-swap}, $\rho_2$ is a damage-free color realization of a valid color plan, and $D$ is $\rho_2$-pivoted by $u$,
and by Observation~\ref{obs:safety}, $\rho_2$ is $4$-safe.
If $v^- \notin \range(\rho_2)$, we set $\rho_3 := \rho_2$.
Suppose that $v^- \in \range(\rho_2)$.
Recall that $w \in K \setminus \range(\rho)$ is a vertex such that $w \ne u$, which means that we may define $\rho_3$ to be the color realization obtained by applying Lemma~\ref{lem:most-technical} to $\rho_2$.

In all cases, $\rho_3$ is a damage-free and almost $7$-safe color realization of a valid color plan such that $v^- \notin \range(\rho_3)$, $\rho_3(c) = v^+$, and $u \notin \range(\rho_3)$.
Suppose that $u'$ is a $\rho_3$-pivot, and let $D'$ be a maximal set $\rho_3$-pivoted by $u'$, and let $Q' \seq S \cup K$ be the maximal set such that $u'$ is a $(\chi, Q')$-pivot.
Observe that $u \notin D'$.
If $u'$ is not a $(\chi, D')$-pivot, then we can compute the desired (almost $9$-safe) color realization $\rho'$ in polynomial time by applying Lemma~\ref{lem:not-pivot} to $\rho_3$.
Otherwise, we have $\range(\rho_3) \nsubseteq Q'$ by Lemma~\ref{lem:nsubseteqQ2}, which means that we can use Lemma~\ref{lem:KnseqQ-swap} to compute the desired (almost $11$-safe) color realization $\rho'$ in polynomial time.
\end{proof}

The last remaining case is when some candidate is far from $u$ ($K \nsubseteq Q$) but all candidates chosen by $\rho$ are close ($\range(\rho) \seq Q$): recall that the other cases are handled by Lemmas~\ref{lem:KnseqQ-swap} and~\ref{lem:KseqQ}.
Note that in one case, the following lemma actually uses Lemma~\ref{lem:KseqQ}.

\begin{lemma}\label{lem:KnseqQ-shift}
Let $\rho, u, D$, and $Q$ be as in Lemma~\ref{lem:KseqQ}.
If $\rho$ is $0$-safe, $u$ is a $(\chi, D)$-pivot, $K \nsubseteq Q$, and $\range(\rho) \seq Q$,
then an almost $11$-safe, damage-free, and pivot-free color realization $\rho'$ of a valid color plan can be computed in polynomial time.
\end{lemma}

\begin{proof}
Let $\rho, u, D$, and $Q$ be as in the statement.
First, suppose that $K \seq Q \cup \{u\}$ and $\range(\rho) \nsubseteq D$.
Let $v$ be a vertex in $\range(\rho) \setminus D$.
By maximality of $D$, we have $v \notin N(u)$; let $x$ be the vertex in $Q \cap N(u) \cap N(v)$.
Suppose for contradiction that $x \in S_\rho$.
By maximality of $D$, we obtain $x \in D$.
Since $x$ is a $(u, Q)$-link, we obtain $x \in K^+$ by Definition~\ref{def:links}, which implies $x \in B_\rho$.
By Lemma~\ref{lem:valid-realization}(\ref{vraux:candidate}), $x$ is a $\chi_\rho$-candidate, which implies that $x$ is $\chi_\rho$-tight because $x$ is $\chi$-tight.
By maximality of $D$, we obtain $v \in D$, which is a contradiction.
Hence, $x \notin S_\rho$.
Since $x$ is $\chi$-tight, we have $x \in K$.
Now the desired color realization $\rho'$ can be computed by Lemma~\ref{lem:KseqQ} (using $v^+ := v$, $v^- := x$, and $w := x$).

Second, suppose that $\range(\rho) \seq D$ or $K \nsubseteq Q \cup \{u\}$.
Let $v \in K \setminus Q$ be such that, if possible, $v \ne u$.
Let $\rho' = \rho[c \mapsto v]$, where $c \in [p+1, k]$ is a color satisfying the properties stated in Lemma~\ref{lem:good-color}.
Suppose for contradiction that $v' \in \range(\rho')$ is damaged in $\rho'$, i.e., there is a color $d \in [b+1, p]$ such that $w := \rho'(d) = \rho(d) \in N(v')$ and there is a vertex $w' \in N(v') \cap \chi^{-1}(d)$.
By Lemma~\ref{lem:damaged}(\ref{daux:rhoc=rho'c}), we have $v' = v$ because $c >p$.
Since $\range(\rho) \seq Q$, we have $w \in Q$.
Since $v$ is $0$-safe in $\rho$ and $v \in N(S)$, we have $\rho(c') \in N(S) \cap Q$ for every $c' \in [p+1, k]$; in particular, $\rho(c')$ is $S$-influenced.
If $v \ne u$, then $v \in N(S) \cap N(w) \cap \overline{Q \cup \{u\}}$, which means, by Lemma~\ref{lem:pivot-properties}(\ref{aux:outer-S}) applied to $Q$, that
$w$ is the only $S$-influenced vertex in $Q$, which is a contradiction with $k \ge p+2$.
Hence, suppose that $v = u$, and observe that $w' \in S \cap N(u)$.
By maximality of $Q$, we have $w' \in Q$, and by Lemma~\ref{lem:pivot-properties}(\ref{aux:contains-all-S-influenced}), $\rho(c') \in N(w')$ for each $c' \in [p+1, k]$.
In particular, $w'$ is a $(u, Q)$-link, which is a contradiction with Definition~\ref{def:links} because $w' \notin B$ (since $\chi(w') = d \in [b+1, p]$).
We have proven that $\rho'$ is damage-free.

Let $\pi$ be the valid color plan realized by $\rho$, and observe that $\rho'$ realizes $\pi$ by Lemma~\ref{lem:is-valid-plan}(\ref{vpaux:trivial}).
By Lemma~\ref{lem:get-safety}, $\rho'$ is almost $1$-safe.
Finally, suppose for contradiction that $u'$ is a $\rho'$-pivot, and let $D'$ be a maximal set $\rho'$-pivoted by $u'$.
Observe that $u' \ne u$ by Lemma~\ref{lem:u-is-not-rho'-pivot}.
By Lemma~\ref{lem:not-different-pivot-shift}, $u$ is $(u', D')$-link.
In particular, $u \in D'$, which implies $v = u$.
By choice of $v$, we obtain $\range(\rho) \seq D$, and by Lemma~\ref{lem:shift-to-u}, $u$ is not $(u', D')$-link, which is a contradiction.
Hence, $\rho'$ is pivot-free, which concludes the proof. 
\end{proof}

\subsubsection{Computing a feasible color realization}\label{subsub:compute-feasible}

Finally, we have arrived at the culmination of Section~\ref{sub:getting-rid}.

\begin{lemma}\label{lem:get-feasible-realization}
If $(\chi, B)$ is not failing, then given a valid color plan of $(\chi, B)$, a feasible color realization can be computed in polynomial time.
\end{lemma}
\begin{proof}
Suppose that $(\chi, B)$ is not failing and let $\pi$ be a valid color plan of $(\chi, B)$.
By Lemma~\ref{lem:compute-damage-free}, an almost $0$-safe and damage-free color realization $\rho$ of $\pi$ can be computed in polynomial time.
If $\rho$ is pivot-free, then we can compute a feasible color realization by Lemma~\ref{lem:handle-block}.
Hence, suppose that $u$ is a $\rho$-pivot.
By Lemma~\ref{lem:not-critical}, $\pi$ is not critical, so by Lemma~\ref{lem:compute-damage-free}, $\rho$ is $0$-safe.

If $u$ is not a $(\chi, D)$-pivot, then let $\rho'$ be the color realization given by Lemma~\ref{lem:not-pivot}.
Suppose that $u$ is a $(\chi, D)$-pivot, and let $Q \seq S \cup K$ be the maximal set such that $u$ is a $(\chi, Q)$-pivot.
Suppose that $K \nsubseteq Q$.
If $\range(\rho) \nsubseteq Q$, then we define $\rho'$ to be the color realization given by Lemma~\ref{lem:KnseqQ-swap}, and otherwise, we define $\rho'$ to be the color realization given by Lemma~\ref{lem:KnseqQ-shift}.
Now suppose that $K \seq Q$.
Since $D$ is $\rho$-pivoted, we have $[k] \seq \chi_\rho(D)$, which implies $[b] \seq \chi(D) \seq \chi(Q)$.
Since $(\chi, B)$ is not failing, we know that $|K| > k - b$ and there are two vertices $v^-, v^+ \in Q \setminus S$ such that $v^- \in N(u) \cap N(v^+)$, see Definition~\ref{def:failing}.
Since $|\range(\rho)| = k-b$ and $u \notin K$ (because $K \seq Q$), there is a vertex $w \in K \setminus \range(\rho)$ such that $w \ne u$.
Hence, we may define $\rho'$ to be the color realization given by Lemma~\ref{lem:KseqQ}.

Observe that, in all cases, $\rho'$ is almost $11$-safe, damage-free, and pivot-free color realization of a valid color plan, and it can be computed in polynomial time.
Hence, by Lemma~\ref{lem:handle-block} applied to $\rho'$, we can compute a feasible color realization, which concludes the proof.
\end{proof}

\subsection{Computing a partial $b$-coloring}\label{sub:coloring-neighbors}

As before, let $S$ be a fen-core, $S^+$ be an extension of $S$, $(\chi, B)$ be an $S$-profile, and $\kout = \{u \in \overline{S} \sep \red_\chi(u) \ge 0\}$.
Recall that $\chi(S) \seq [p]$ and $\chi(B) = [b]$.
Let $\rho$ be a feasible color realization.
We will construct a partial $b$-coloring $\psi$ of $G$ described by $(\chi, B)$ such that $\psi \supseteq \chi_\rho$.

For each $v \in \overline{S_\rho}$, let $L_v = \{u \in S_\rho \sep u \in N(v)$ or there is a vertex $x \in B_\rho \cap N(u) \cap N(v)\}$.
Informally, $L_v$ is the set of vertices that possibly restrict the color of $v$:
$v$ cannot be colored with the color of $u \in L_v$ if $u \in N(v)$ or if the vertex $x \in N(u) \cap N(v)$ is, in some sense, tight.

\subsubsection{Description of the algorithm}\label{subsub:description}

\newcommand{\Continue}{\com{continue}\xspace}

\begin{algorithm}[h]
\DontPrintSemicolon
\KwIn{$G$, $k$, $S$, $\chi$, $B$, and $\rho$.}
\KwOut{A partial $b$-coloring $\psi$ of $G$ described by $(\chi, B)$ such that $\psi \supseteq \chi_\rho$.}\SetKwProg{Fn}{Procedure}{:}{}
\SetKwFunction{Resolve}{Resolve}
$\psi := \chi_\rho$;\label{line:init} 
$\textsf{unresolved} := B_\rho$\;
\lIf{there is $u \in B_\rho$ such that $|N(u) \cap B_\rho| \ge k-p-8$}{\Resolve{$u$}\label{line:res1}}

\ElseIf{there is $v \in \overline{S_\rho}$ such that $|L_v| \ge k-3$\label{line:3}}
{\Resolve{$u$}, where $u$ is an arbitrary vertex of $N(v) \cap B_\rho$\label{line:res2}}

$\{u_1, \ldots, u_m\} := \textsf{unresolved}$, where $\dist(u_i, S^+) \le \dist(u_j, S^+)$ when $1 \le i \le j \le m$\;\label{line:define-prec}

\lFor{$i$ from $1$ to $m$}{\Resolve{$u_i$}\label{line:resolverest}}

\Return $\psi$

\BlankLine
\BlankLine

\Fn{\Resolve{$u$}}{
$U := \{v \in N(u) \sep \psi(v)$ is undefined$\}$\;
\For{\Each $v \in U$}{
$W_v := N(v) \cup \{w \in V(G) \sep w$ is $\psi$-linked to $v$ via a vertex in $B_\rho\}$\;\label{line:wi}
$C_v := \psi(W_v)$\label{line:ci}
}
let $H$ be a bipartite graph such that $V(H) := [k] \cup U$ and\;
$E(H) := \{cv \sep v \in U, c \in [k] \setminus C_v\}$\;
$C := [k] \setminus \psi(N[u])$\;
$H' := H[U \cup C]$\;
let $\mu$ be a maximal matching in the graph $H'$\;
\For{\Each $v \in U$}{
\lIf{$v$ is matched with $c$ in $\mu$}{$\psi(v) := c$\label{line:color-matched}}
\lElse{$\psi(v) := c$, where $c$ is any color such that $vc \in E(H)$\label{line:color-unmatched}}
}
$\textsf{unresolved} := \textsf{unresolved} \setminus \{u\}$\;
}
\caption{Constructing a partial $b$-coloring}\label{alg:partial-b-coloring}
\end{algorithm}

\newcommand{\pred}{\textnormal{pred}}

We will prove that Algorithm~\ref{alg:partial-b-coloring} produces a partial $b$-coloring of $G$.
Let us provide some intuition behind the algorithm.
The algorithm builds up a coloring $\psi$, which is initialized as $\chi_\rho$, see line~\ref{line:init}.
Recall that the vertices in $B_\rho$ should become $b$-vertices.
For each vertex $u \in B_\rho$, the algorithm uses a subroutine $\Resolve(u)$, which colors the neighbors of $u \in B_\rho$ (formally, each call of the subroutine modifies the ``global'' variable $\psi$).
Usually, it suffices to call $\Resolve$ in a BFS order starting with $S^+$, see lines~\ref{line:define-prec} and~\ref{line:resolverest}.
However, it may happen that some vertex, called the \emph{exceptional} vertex, needs to be resolved preferentially, see lines~\ref{line:res1} and~\ref{line:res2}.
Intuitively, the exceptional vertex is close to almost all vertices of $B_\rho$, which might create a $\psi$-pivot (even though $\rho$ is pivot-free), see Figure~\ref{fig:exceptional-vertex}.

\begin{figure}[t]
\begin{minipage}{0.45\textwidth}
\centering
\begin{tikzpicture}
\tikzmath{ \dist = 20; \dis=14;}
\begin{scope}[every node/.style={draw, circle, minimum width=10pt, inner sep=2pt, fill=gray!25!white}]
\node[minimum width=17pt] (u) {$u$};
\node[minimum width=17pt, fill=white, left= \dist pt of u ] (v1) {$v_1$};
\node[minimum width=17pt, fill=white, right= \dist pt of u ] (v2) {$v$};
\node[minimum width=17pt, below= \dis pt of v1] (w1) {$w_1$};
\node[minimum width=17pt, fill=white, below= \dis pt of v2] (x) {$v_2$};
\node[minimum width=17pt, below= \dis pt of x] (w2) {$w_2$};
\node[above= \dist pt of u ] (a) {};
\node[above left= \dist pt of u ] (b) {};
\node[above right= \dist pt of u ] (c) {};
\draw (a)--(u)--(b);
\draw (c)--(u)--(v1)--(w1);
\draw (u)--(v2)--(x)--(w2);
\end{scope}
\node[left = 0pt of w1] {$c_1$};
\node[right = 0pt of w2] {$c_2$};
\node[red, left = 0pt of v1] {$c_2$};
\node[red, right = 0pt of x] {$c_1$};
\end{tikzpicture}
\end{minipage}
\hfill
\begin{minipage}{0.45\textwidth}
\centering
\begin{tikzpicture}
\tikzmath{ \dist = 20; \dis=14;}
\begin{scope}[every node/.style={draw, circle, minimum width=10pt, inner sep=2pt, fill=gray!25!white}]
\node[minimum width=17pt, fill=white] (u) {$v$};
\node[left= \dist pt of u ] (u1) {$u_1$};
\node[right= \dist pt of u ] (u2) {$u_2$};
\node[below= \dis pt of u1, fill=white ] (v1) {$v_1$};
\node[below= \dis pt of u2, fill=white ] (v2) {$v_2$};
\node[below= \dis pt of v1] (w1) {$w_1$};
\node[below= \dis pt of v2] (w2) {$w_2$};
\node[above= \dist pt of u ] (a) {};
\node[above left= \dist pt of u ] (b) {};
\node[above right= \dist pt of u ] (c) {};
\draw (a)--(u)--(b);
\draw (c)--(u)--(u1)--(v1)--(w1);
\draw (u)--(u2)--(v2)--(w2);
\end{scope}
\node[left = 0pt of w1] {$c_1$};
\node[right = 0pt of w2] {$c_2$};
\node[red, left = 0pt of v1] {$c_2$};
\node[red, right = 0pt of v2] {$c_1$};
\end{tikzpicture}
\end{minipage}
\caption{Two cases in which there is an exceptional vertex. The gray vertices are in $B_\rho$ and the colors assigned by $\chi_\rho$ are drawn in black (only for some vertices).
The colors which might be assigned if lines~\ref{line:res1},~\ref{line:3}, and~\ref{line:res2} were removed from the algorithm are drawn in red.
In both cases, $\protect\Resolve(w_1)$ and $\protect\Resolve(w_2)$ are executed first because $\{w_1, v_1, u, v, v_2, w_2\}$ or $\{w_1, v_1, u_1, v, u_2, v_2, w_2\}$ induces a subpath of a long outer path, and $v$ becomes a $\psi$-pivot (and hence cannot be colored).
\textbf{Left:} All colors except for $c_1$ and $c_2$ are present in $N[u]$ and $u$ is $\chi_\rho$-tight.
\textbf{Right:} All colors except for $c_1$ and $c_2$ are present in $N(v)$ and $u_1, u_2$ are $\chi_\rho$-tight.}
\label{fig:exceptional-vertex}
\end{figure}
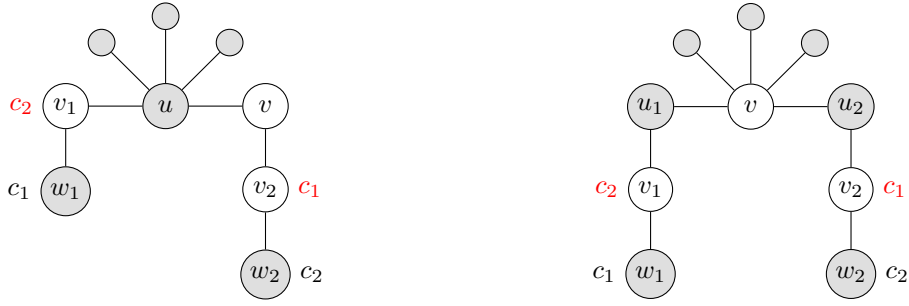

The subroutine $\Resolve(u)$ works as follows.
For each vertex $v \in N(u)$ that is not already colored by $\psi$, we compute which colors are possible for $v$, see line~\ref{line:ci}.
We represent this information as an auxiliary bipartite graph $H$ between $U := N(u) \setminus \dom(\psi)$ and $[k]$; recall that we used a similar graph $H$ in Algorithm~\ref{alg:find-damage-free}.
Since our goal is to make $u$ a $b$-vertex, we need to match the colors not present in the neighborhood of $u$ with $U$, see line~\ref{line:color-matched}.
Since $u$ is allowed to have multiple neighbors of the same color, the unmatched vertices can be colored with any possible color, see line~\ref{line:color-unmatched}.

\subsubsection{Correctness of the algorithm}\label{subsub:correct}

For $u, v \in B_\rho$, let $u \prec v$ denote the fact that $\Resolve(u)$ is evaluated before $\Resolve(v)$.
If $\Resolve(u)$ is called on line~\ref{line:res1} or~\ref{line:res2}, then we say that $u \in B_\rho$ is the \emph{exceptional} vertex (it is unique because line~\ref{line:3} contains an \textbf{else if}).
Note that if $u$ is the exceptional vertex, then $u \prec v$ for each $v \in B_\rho$ distinct from $u$.

The following lemma states that either there is exactly one vertex $u$ satisfying the condition on line~\ref{line:res1} and all vertices $v$ satisfying the condition on line~\ref{line:3} are in $N(u)$, or there is no such vertex $u$ and there is exactly one vertex $v$ satisfying the condition on line~\ref{line:3}.

\begin{lemma}\label{lem:exceptional-vertex}
~
\begin{enumerate}[(a)]
\item There is at most one vertex $u \in B_\rho$ such that $|N(u) \cap B_\rho| \ge k-p-8$.
\item If there is a vertex $u \in B_\rho$ as in (a), then each vertex $v \in \overline{S_\rho}$ such that $|L_v| \ge k-3$ is in $N(u)$.
\item If there is no vertex $u \in B_\rho$ as in (a), then there is at most one vertex $v \in \overline{S_\rho}$ such that $|L_v| \ge k-3$
\end{enumerate}
\end{lemma}
\begin{proof}
For $v \in \overline{S_\rho}$ and $u \in L_v$, we define $u^v := u$ if $u \in N(v)$, and otherwise we define $u^v$ to be the vertex in $N(u) \cap N(v) \cap B_\rho$.

For (a), suppose for contradiction that there are two such vertices $u$ and $u'$. Since $|B_\rho| = k$, we obtain $|B_\rho \cap N(u) \cap N(u')| \ge k-2p-16$. Since $k \ge 3p+18$ and $|S| = p$, there are two distinct vertices $v, v' \in \range(\rho) \cap N(u) \cap N(u')$. Now we obtain a contradiction with Observation~\ref{obs:short-cycles-are-in-S} using the cycle $G[\{u, v, u', v'\}]$.

For (b), suppose for contradiction that there are vertices $u \in B_\rho$ and $v \in \overline{S_\rho}$ such that $|N(u) \cap B_\rho| \ge k-p-8$, $|L_v| \ge k-3$, and $v \notin N(u)$.
Since $|S| \le p$, $S_\rho = S \cup B_\rho$, and $L_v \seq S_\rho$, we have $|L_v \cap \overline{S}| \ge k-3-p$ and $|N(u) \cap L_v \cap \overline{S}| \ge k-11-2p$.
Since $k \ge 2p+13$, there are two distinct vertices $w_1, w_2 \in N(u) \cap L_v \cap \overline{S}$.
Now we obtain a contradiction with Observation~\ref{obs:short-cycles-are-in-S} because there is a cycle $C$ of length at most 6 in $G[\{u, w_1, w_1^v, v, w_2^v, w_2\}]$ such that $V(C) \nsubseteq S$ since $w_1 \in V(C) \setminus S$.
Note that if $w_1^v = w_2^v$, then $v \notin V(C)$.

For (c), suppose for contradiction that there are two such vertices $v$ and $w$, i.e., $|L_v \cap L_w| \ge k-6$.
Since $|S| \le p$, we have $|L_v \cap L_w \cap \overline{S}| \ge k-p-6$.
If there are two distinct vertices $u, u' \in N(v) \cap N(w) \cap B_\rho$, then $G[\{v, u, w, u'\}]$ is a cycle, which contradicts Observation~\ref{obs:short-cycles-are-in-S} since $v \notin S$.
Now suppose that $N(v) \cap N(w) \cap B_\rho = \emptyset$.
Since $k \ge p+8$, there are two distinct vertices $u_1, u_2 \in L_v \cap L_w \cap \overline{S} \seq B_\rho$.
Observe that $\{u_1^v, u_2^v\} \cap \{u_1^w, u_2^w\} = \emptyset$ since $N(v) \cap N(w) \cap B_\rho = \emptyset$.
Hence, there is a cycle $C$ of length at most 8 in $G[\{v, u^v_1, u_1, u^w_1, w, u^w_2, u_2, u^v_2\}]$ containing $u_1 \notin S$, which contradicts  Observation~\ref{obs:short-cycles-are-in-S}.
Note that for $x \in \{v, w\}$, if $u^x_1 = u^x_2$, then $x \notin V(C)$.
Finally, suppose that $N(v) \cap N(w) \cap B_\rho = \{u\}$. If there is a vertex $x \in (L_v \cap L_w) \setminus N[u]$, then $G[\{v, u, w, x^w, x, x^v\}]$ is a cycle, contradicting Observation~\ref{obs:short-cycles-are-in-S}.
Otherwise, $L_v \cap L_w \seq N[u]$, which means $|N[u] \cap B_\rho| \ge |L_v \cap L_w \cap B_\rho| \ge |L_v \cap L_w \cap \overline{S}| \ge k-p-6$.
Hence, $|N(u) \cap B_\rho| \ge k-p-7 \ge k-p-8$, which contradicts the assumption that there is no vertex as in (a).
\end{proof}

For each vertex $u \in B_\rho$, let $\psi_u$ (resp. $\psi^u$) be the value of the variable $\psi$ at the beginning of the execution of $\Resolve(u)$ (resp. after $\Resolve(u)$ is executed).

The following lemma states that there are at most 3 newly colored vertices close to $u$ at the beginning of the evaluation of $\Resolve(u)$.
Note that we assume that the algorithm has successfully colored the neighborhoods of all vertices $u' \prec u$.

\begin{lemma}\label{lem:at-most-three}
Let $u \in B_\rho$ be a vertex and let $U$ and $W_v$ for $v \in U$ be the objects computed during the execution of $\Resolve(u)$ in Algorithm~\ref{alg:partial-b-coloring}.
If for each $u' \in B_\rho$ such that $u' \prec u$ and each $v' \in N(u')$, $\psi_u(v')$ is defined,
then $|\ca W| \le 3$, where \[\ca W = \left\lbrace w \in N(u) \cup \bigcup_{v \in U} W_v \sep w \notin S_\rho, \psi_u(w) \text{ is defined}   \right\rbrace.\]
\end{lemma}
\begin{proof}
Suppose for contradiction that there are distinct vertices $w_1, w_2, w_3, w_4 \in \ca W$.
Let us fix $i \in [4]$.
If $w_i \in N(u)$, then let $v_i := w_i$ and $x_i := w_i$.
Suppose that $w_i \in W_v$ for some $v \in U$; let $v_i := v$.
If $w_i \in N(v_i)$, let $x_i := w_i$, and otherwise let $x_i$ be the $(\psi_u, v_i, w_i)$-link in $B_\rho$.
Since $w_i \notin S_\rho$, we know that $\chi_\rho(w_i)$ is undefined, which means that there is a vertex $u_i \in B_\rho \cap N(w_i)$ such that $u_i \prec u$ and $\psi_u(w_i)$ was computed during the execution of $\Resolve(u_i)$.
Let us fix $j \in [4]$ and observe that $u_i \notin \{w_j, v_j\}$ because $w_j \notin S_\rho$ and either $v_j = w_j$ or $v_j \in U \seq \overline{S_\rho}$.
Suppose that $u_i = x_j$.
Observe that we have $w_j \ne x_j$ and $v_j \in N(x_j)$.
Since $x_j = u_i \prec u$, we deduce that $\psi_u(v_j)$ is defined using our assumption, which is a contradiction with $v_j \in U$.
Hence, $u_i \ne x_j$.

If $u_i = u_j$ for $1 \le i < j \le 4$, then there would be a cycle $C$ in $G[\{u, v_i, x_i, w_i, u_i, w_j, x_j, v_j\}]$ of length at most 8 containing $w_i \notin S$, which contradicts Observation~\ref{obs:short-cycles-are-in-S}.
Hence, we know that $|\{u_1, u_2, u_3, u_4\}| = 4$.
Without loss of generality, we may assume that if $u_i$ is the exceptional vertex for some $i \in [4]$, then $i = 4$.
Let us fix $i \in [3]$, and observe that $\dist(u_i, S^+) \le \dist(u, S^+)$ because $u_i \prec u$, see line~\ref{line:define-prec}.
If $u \in S$, then $u_i \in S$ and $G[\{v_i, x_i, w_i\}]$ is an outer path, which contradicts Definition~\ref{def:fen-core}.
Hence, $u \notin S$.
Let $P_i$ be a shortest $u_i$-$S^+$ path.
Clearly, $u \notin V(P_i)$.
Let $P_i'$ be the $u$-$S^+$ path obtained from $P_i$ by adding the edges $u_iw_i, w_ix_i, x_iv_i, v_iu$ if $x_i \ne w_i$, or the edges $u_iw_i, w_iv_i, v_iu$ if $x_i = w_i \ne v_i$, or the edges $u_iv_i, v_iu$ otherwise.
Now $P_1', P_2', P_3'$ are three distinct $u$-$S^+$ paths, which contradicts Observation~\ref{obs:S+}.
\end{proof}

Now we are ready to show that the algorithm colors all neighbors of $u$.

\begin{lemma}\label{lem:positive-degree-in-H}
If $u \in B_\rho$, then $\psi^u(v)$ is defined for each $v \in N(u)$.
\end{lemma}

\begin{proof}
Suppose for contradiction that there is a vertex $u \in B_\rho$ such that some vertex $v \in N(u)$ is not colored by $\psi^u$; let us choose the $\prec$-minimal such vertex $u$.
Let $U, H, W_v, C_v$ be the objects computed during the execution of $\Resolve(u)$.
Observe that $v \in U$ and that $v$ has degree 0 in $H$, see line~\ref{line:color-unmatched}.
By definition of $E(H)$, $C_v = [k]$, i.e., for each color $c \in [k]$, there is a vertex $w_c \in W_v$ such that $\psi_u(w_c) = c$ and either $w_c \in N(v)$ or $w_c$ is $\psi_u$-linked to $v$ via a vertex $x_c \in B_\rho$.
If $w_c \in N(v)$, we set $x_c := w_c$.
Let $X = \{w_c, x_c \sep c \in [k]\}$ and observe that $v$ is an $(\psi_u, X)$-pivot and $\psi_u(X) = [k]$.
If $X \seq S_\rho$, then $v$ would be a $\rho$-pivot, which is impossible since $\rho$ is pivot-free, see Definition~\ref{def:feasible}.
Hence, $X \nsubseteq S_\rho$.

Let $y \in X \setminus S_\rho$.
Since $\psi_u(y)$ is defined, there is a vertex $u' \in B_\rho \cap N(y)$ such that $u' \prec u$, which implies that $u$ is not the exceptional vertex.
Since $\psi_u(v)$ is undefined, the exceptional vertex (if it exists) is not adjacent to $v$ (using the $\prec$-minimality of $u$).
Hence, by Lemma~\ref{lem:exceptional-vertex}, $|L_v| < k - 3$ (see line~\ref{line:3}), which means that there are four distinct colors $c_1, c_2, c_3, c_4 \in [k]$ such that for each $i \in [4]$, $w_{c_i} \notin L_v$.
Since $W_v \cap S_\rho \seq L_v$, we have $|\{w \in W_v \setminus S_\rho \sep \psi_u(w)$ is defined$\}| \ge 4$, which is a contradiction with Lemma~\ref{lem:at-most-three}.
Note that the assumption of Lemma~\ref{lem:at-most-three} holds by the $\prec$-minimality of $u$.
\end{proof}

We will need the following lemma, which characterizes why a vertex in $U$ cannot be assigned a color. Intuitively, $E_1$ corresponds to constraints imposed by $B_\rho$, and $E_2$ to the remaining constraints.

\begin{lemma}\label{lem:number-of-edges-in-H'}
Let $u \in B_\rho$ and $U, C, H'$, and $C_v$ for $v \in U$ be the objects computed in $\Resolve(u)$.
If $H^*$ is the bipartite complement of $H'$, then $E(H^*)$ can be partitioned into two sets, $E_1$ and $E_2$, such that for each $c \in C$, there is at most one edge incident to $c$ in $E_1$, and $|E_2| \le 1$ if $u$ is the exceptional vertex, and $|E_2| \le 4$ otherwise.  
\end{lemma}
\begin{proof}
Observe that for each edge $vc \in E(H^*)$, we have $c \in C_v$, i.e., there is a vertex $w$ such that $\psi_u(w) = c$ and either $w \in N(v)$ or $w$ is $\psi_u$-linked to $v$ via a vertex $x \in B_\rho$.
For each edge $e \in E(H^*)$, let us denote the corresponding vertices $w$ and $x$ by $w_e$ and $x_e$; if $x_e$ does not exist, we define $x_e := w_e$.
We define $E_1$ to be the set of edges $e \in E(H^*)$ such that $w_e \in B_\rho$.
If there were two distinct edges $e = vc \in E_1$ and $e' = v'c \in E_1$, then $w_e = w_{e'}$ and there would be a cycle in $G[\{u, v, x_e, w_e, x_{e'}, v'\}]$, which would contradict Observation~\ref{obs:short-cycles-are-in-S}.

Now observe that there is at most one edge $e = vc \in E(H^*)$ such that $w_e \in S$;
if there were two such edges $e = vc$ and $e' = v'c'$, then there would be a short outer path in $G[\{x_{e'}, v', u, v, x_e\}]$.
If $u$ is the exceptional vertex, then $\psi_u = \chi_\rho$, which means that $w_e \in S_\rho = S \cup B_\rho$ for each $e \in E(H^*)$.
Hence, $|E_2| \le 1$ as desired.
Now suppose that $u$ is not the exceptional vertex.
In this case, there are at most three edges $e \in E(H^*)$ such that $w_e \notin S_\rho$, by Lemma~\ref{lem:at-most-three}.
This implies $|E_2| \le 4$, which concludes the proof.
\end{proof}

Now we are ready to show that each vertex in $B_\rho$ is a $\psi^u$-candidate.
Note that this implies that $u$ is a $b$-vertex in $\psi^u$ because all neighbors of $u$ are colored in $\psi^u$ by Lemma~\ref{lem:positive-degree-in-H}.

\begin{lemma}\label{lem:remains-to-be-candidate}
If $u \in B_\rho$, then each $u' \in B_\rho$ is a $\psi^u$-candidate.
\end{lemma}
\begin{proof}
Suppose for contradiction that there is a vertex $u \in B_\rho$ such that some vertex $u' \in B_\rho$ is not a $\psi^u$-candidate; let us choose the $\prec$-minimal such vertex $u$.
Let $U$, $H$, $C$, $\mu$, and $W_v, C_v$ for $v \in U$ be the objects computed during the execution of $\Resolve(u)$.
Since the initial value of $\psi$ is $\chi_\rho$ and each vertex in $B_\rho$ is a $\chi_\rho$-candidate by Lemma~\ref{lem:valid-realization}(\ref{vraux:candidate}), we know that each vertex in $B_\rho$ is a $\psi_u$-candidate.

\subparagraph*{Case 1: $u \ne u'$}
Since $u'$ is a $\psi_u$-candidate but not a $\psi^u$-candidate, there is a vertex $v \in N(u') \cap U$ newly colored in $\Resolve(u)$, i.e., $\psi_u(v)$ is undefined but $\psi^u(v)$ is defined.
By Observation~\ref{obs:short-cycles-are-in-S}, $N(u) \cap N(u') = \{v\}$ because otherwise there would be a short cycle not contained in $S$ (since $v \notin S$).
In other words, $v$ is the only newly colored vertex in $N(u')$; let $c = \psi^u(v)$.
Let $r$ be the $\psi_u$-redundancy of $u'$, see Definition~\ref{def:S-profile}.
If $r \ge 1$, then the $\psi^u$-redundancy of $u'$ is at least $r-1 \ge 0$, which means that $u'$ is a $\psi^u$-candidate, a contradiction.
Hence, $r = 0$, i.e., $u'$ is $\psi_u$-tight.
Since $u'$ is not a $\psi^u$-candidate, there is a vertex $w \in N[u']$ such that $\psi_u(w) = c$.
Observe that $w \in W_v$, which implies that $c \in C_v$ and that $cv \notin E(H)$.
This is a contradiction with $\psi^u(v) = c$, see lines~\ref{line:color-matched} and~\ref{line:color-unmatched}.

\subparagraph*{Case 2: $u = u'$}

By Lemma~\ref{lem:positive-degree-in-H}, each vertex $v \in U$ is colored by $\psi^u$, which implies that there is a color $c$ such that $c \notin \psi^u(N[u])$.
In other words, $c \in C$ and $c$ is not matched with any vertex in $\mu$.
By Hall's marriage theorem, there is a set $C' \seq C$ such that $|N_H(C')| < |C'|$; let us consider the smallest such set $C'$.

\subparagraph*{Case 2a: $|C'| = |U|$}

Since $u$ is a $\psi_u$-candidate, we have $|C| \le |U|$, which implies that $C = C'$ and that $u$ is $\psi_u$-tight.
By minimality of $C'$, $|N_H(C)| = |C| - 1$; let $\{v\} = U \setminus N_H(C)$.
Let us show that for each color $c \in [k]$, there is a vertex $w_c \in \psi_u^{-1}(c) \cap W_v$. 
For $c \in C$, such a vertex $w_c$ exists by definition of $H$ since $v \notin N_H(C)$.
Recall that $[k] \setminus C = \psi_u(N[u])$.
Hence, for $c = \psi_u(u)$, we choose $w_{c} := u$, and for $c \in \psi_u(N(u))$, there is a vertex $w_c$ that is $\chi_\rho$-linked to $v$ via $u$ (here we use the fact that $u$ is a $\psi_u$-tight vertex in $B_\rho$).

First, suppose that $u$ is the exceptional vertex.
In this case, $\psi_u = \chi_\rho$.
Now the existence of $w_c$ for each $c \in [k]$ implies that $v$ is a $\rho$-pivot, which is a contradiction with the feasibility of $\rho$.
Second, suppose that $u$ is not the exceptional vertex.
By Lemma~\ref{lem:at-most-three}, there are at most three colors $c \in [k]$ such $w_c \notin S_\rho$; let $C^*$ be the subset of $[k]$ without these at most three colors.
Clearly, $|C^*| \ge k-3$.
Since $\{w_c \sep c \in C^*\} \seq W_v \cap S_\rho \seq L_v$, we obtain $|L_v| \ge k-3$.
Hence, by Lemma~\ref{lem:exceptional-vertex}, the exceptional vertex is in $N(v)$.
This is a contradiction with $v \in U$ and Lemma~\ref{lem:positive-degree-in-H}.

\subparagraph*{Case 2b: $|C'| < |U|$}
Let $H^*, E_1, E_2$ be the objects obtained when Lemma~\ref{lem:number-of-edges-in-H'} is applied to $u$.
First, suppose that $u$ is the exceptional vertex.
If $C' = \{c\}$, then $u$ would be $\rho$-blocked by $c$, which is a contradiction since $\rho$ is feasible.
Hence, let $c_1, c_2 \in C'$ be two distinct colors.
Since $|N_H(C')| < |C'| < |U|$, there are at least two vertices $v_1, v_2 \in U$ such that $v_1, v_2 \notin N_H(C')$.
Hence, $c_1v_1, c_1v_2, c_2v_1, c_2v_2 \notin E(H)$.
This is a contradiction with Lemma~\ref{lem:number-of-edges-in-H'} since there can be at most three edges in $H^*[\{c_1, c_2, v_1, v_2\}]$: two in $E_1$ and one in $E_2$.

Second, suppose that $u$ is not the exceptional vertex.
Observe that $|N(u) \cap B_\rho| \le k-p-9$, see line~\ref{line:res1}.
There are at most $p$ colors in $\chi_\rho(N[u] \cap S)$, which implies that $|\chi_\rho(N[u])| \le k - 9$.
By Lemma~\ref{lem:at-most-three}, $|\psi_u(N[u])| - |\chi_\rho(N[u])| \le 3$.
Hence, $|\psi_u(N[u])| \le k-6$, which implies $|U| \ge 6$.
Since $|N_H(C')| < |C'| < |U|$, there are distinct vertices $v_1, v_2 \in U \setminus N_H(C')$.
If there are distinct colors $c_1, \ldots, c_5 \in C'$, then $|E'| = 10$ and $E' \seq E(H^*)$, where $E' = \{c_iv_1, c_iv_2 \sep i \in [5]\}$.
However, by Lemma~\ref{lem:number-of-edges-in-H'}, $|E' \cap E_1| \le 5$ and $|E'\cap E_2| \le 4$, which is a contradiction.

Now suppose that $C' = \{c_1, \ldots, c_m\}$, where $m \in [4]$.
Since $|N_H(C')| < |C'|$ and $|U| \ge 6$, there are distinct vertices $v_1, \ldots, v_{7-m} \in U \setminus N_H(C')$.
Let $E' = \{c_iv_j \sep i \in [m], j \in [7-m]\}$ and observe that again $E' \seq E(H^*)$.
Moreover, $|E' \cap E_1| \le m$ and $|E' \cap E_2| \le 4$.
Since $m(7-m) > m+4$, we again have a contradiction with Lemma~\ref{lem:number-of-edges-in-H'}.
\end{proof}

Finally, we are ready to prove the correctness of Algorithm~\ref{alg:partial-b-coloring}.

\begin{lemma}\label{lem:algo-1-works}
The coloring $\psi'\colon S_\rho \cup N(B_\rho) \rightarrow [k]$ computed by Algorithm~\ref{alg:partial-b-coloring} is a partial $b$-coloring of $G$.
\end{lemma}
\begin{proof}
First, let us show that $\psi'$ is proper.
Suppose for contradiction that there is an edge $v_1v_2 \in E(G)$ such that $c := \psi'(v_1) = \psi'(v_2)$.
For $i \in [2]$, let $u_i \in B_\rho$ be the $\prec$-minimal vertex such that $\psi_{u_i}(v_i)$ is defined.
Suppose that $u_1 = u_2$ and let $u := u_1$.
If $u$ is the $\prec$-minimal vertex, then $v_1, v_2 \in S_\rho$, which is a contradiction since $\chi_\rho$ is proper by Lemma~\ref{lem:valid-realization}(\ref{vraux:proper}).
Otherwise, $\psi'(v_1)$ and $\psi'(v_2)$ were computed in $\Resolve(u')$, where $u'$ is the immediate $\prec$-predecessor of $u$.
This means that $G[\{v_1, u', v_2\}]$ is a cycle, which is a contradiction with Observation~\ref{obs:short-cycles-are-in-S} since $v_1 \notin S$.
Hence, $u_1 \ne u_2$.
Without loss of generality, assume that $u_1 \prec u_2$, which implies that $v_2 \in N(B_\rho) \setminus S_\rho$.
Let $w \in B_\rho \cap N(v_2)$ be the vertex such that $\psi'(v)$ was computed during the execution of $\Resolve(w)$; note that $w$ is the immediate $\prec$-predecessor of $u_2$.
Let $U, H$ and $C_x$ for $x \in U$ be the sets computed during the execution of $\Resolve(w)$.
Clearly, $v_2 \in U$.
Since $v_1 \in N(v_2)$, we have $c \in C_{v_2}$, which implies $v_2c \notin E(H)$.
This is a contradiction with $\psi'(v_2) = c$, see lines~\ref{line:color-matched} and~\ref{line:color-unmatched}.

Second, we need to show that each vertex $u \in B_\rho$ is a $b$-vertex in $\psi'$.
Let us fix $u \in B_\rho$.
By Lemma~\ref{lem:remains-to-be-candidate}, $u$ is a $\psi'$-candidate, and by Lemma~\ref{lem:positive-degree-in-H}, $\psi'(v)$ is defined for each $v \in N(u)$.
Hence, by Definition~\ref{def:S-profile}, $u$ is indeed a $b$-vertex in $\psi'$.
Moreover, for each color $c \in [k]$, there is a vertex $u_c \in B_\rho$ such that $\psi'(u_c) = c$.
Therefore, $\psi'$ is a partial $b$-coloring of $G$.
\end{proof}

\subsection{Finishing the proof of Theorem~\ref{thm:fen}}\label{sub:fen-finish}

As before, let $S$ be a fen-core, $S^+$ be an extension of $S$, $(\chi, B)$ be an $S$-profile, and $\kout = \{u \in \overline{S} \sep \red_\chi(u) \ge 0\}$.
Recall that $\chi(S) \seq [p]$ and $\chi(B) = [b]$.

The last remaining step after Algorithm~\ref{alg:partial-b-coloring} has computed a partial $b$-coloring is to color the rest of the graph.
Now we will use that a feasible color realization must be almost $13$-safe.

\begin{lemma}\label{lem:finish-coloring}
If $\rho$ is a feasible color realization and $\psi_0\colon S_\rho \cup N(B_\rho) \rightarrow [k]$ is a partial $b$-coloring described by $(\chi, B)$, then there is a (total) $k$-$b$-coloring $\psi$ described by $(\chi, B)$.
Moreover, $\psi$ can be computed in polynomial time given $\psi_0$.
\end{lemma}
\begin{proof}
Let $\rho$ and $\psi_0$ be as in the statement.
Since $\rho$ is feasible, all vertices in $\kout \setminus B_\rho$ are $13$-safe in $\rho$, possibly except for one vertex $u$ such that either $u \in N(B_\rho)$ or $u$ has at most $p+2$ neighbors in $S \cup N(B_\rho)$, see Definition~\ref{def:color-realization}.
Suppose that such a special vertex $u$ exists.
If $u \in N(B_\rho)$, then $\psi_0(u)$ is defined.
Otherwise, at most $p+2$ neighbors of $u$ are colored by $\psi_0$, which means that we may assign any color not in $\psi_0(N(u))$ to $u$ (using $k \ge p+3$).

We shall color the rest of the graph greedily so that if $\dist(v,S^+) < \dist(w, S^+)$, then $v$ is colored before $w$.
Let us consider the coloring $\psi_v$ immediately before a vertex $v \in V(G)$ is colored; note that $v \notin S_\rho \cup N(B_\rho)$.
If $v$ has degree less than $k$, then there is a color that can be assigned to $v$.
Suppose that $v$ has degree at least $k$, which means $v \in \kout \setminus B_\rho$.
Let $U_0 = \{w \in N(v) \sep$ there is a $v$-$S^+$ path containing $w\}$.
If the special vertex $u$ handled in the previous paragraph exists, then we set $U := U_0 \cup \{u\}$, and otherwise we set $U := U_0$.
By Observation~\ref{obs:S+}, there are at most 2 $v$-$S^+$ paths, which implies $|U_0| \le 2$ and $|U| \le 3$.
Observe that if $\dist(w,S^+) \le \dist(v, S^+)$ for some $w \in N(v)$, then $w \in U$.
In particular, $S^+ \cap N(v) \seq U$.
Moreover, if $w \in N(v) \setminus N(B_\rho)$ is colored by $\psi_v$, then $\dist(w,S^+) \le \dist(v, S^+)$ or $w = u$, which means that $w \in U$.
For $w \in B_\rho$, let $U_w = (N(w) \cap N(v)) \setminus U$, and notice that by Observation~\ref{obs:short-cycles-are-in-S}, $|U_w| \le 1$ (since $v \notin S$).
Hence, $|\psi_v(N(v))| \le 3 + |W|$, where $W = \{w \in B_\rho \sep U_w \ne \emptyset\}$.
Observe that $B\cap W = \emptyset$ because if $w \in B$ and $x \in N(w)\cap N(v)$, then $G[\{w, x, v\}]$ is a $v$-$S^+$ path, which means $x \in U$ and $w \notin W$. 
In other words, $W \seq \range(\rho)$.

Let $c \in [p+1, k]$ be a color, let $w = \rho(c) \in W$, let $x \in N(w)\cap N(v)$, and suppose that $\dist(w, S^+)\le\dist(v,S^+)$.
Observe that there is a $w$-$S^+$ path $P$ not containing $v$.
If $x \in V(P)$, let $P' = G[V(P) \setminus \{w\} \cup \{v\}]$, and otherwise, let $P' = G[V(P) \cup \{v, x\}]$.
In both cases, $P'$ is a $v$-$S^+$ path containing $x$, which means $x \in U$ and $w \notin W$, a contradiction.
Hence, $\dist(v, S^+)< \dist(w,S^+)$.
Since $v$ is $13$-safe in $\rho$, we know that $|W \cap \rho([p+1, k])| \le 13$.
Hence, $|W| \le p+13$ and $|\psi_v(N(v))| \le p+16 < k$, which means that there is a color that can be assigned to $v$ while keeping the coloring proper.

After coloring each vertex, we obtain a proper coloring $\psi$. Since $\psi_0 \seq \psi$ and $\psi_0$ is a $k$-$b$-coloring described by $(\chi, B)$, $\psi$ is also a $k$-$b$-coloring described by $(\chi, B)$.
Moreover, $\psi$ can be computed in polynomial time given $\psi_0$, which concludes the proof.
\end{proof}

\begin{algorithm}[h]
\DontPrintSemicolon
\KwIn{A graph $G$ with feedback edge number $p_G$, an integer $k$}
\KwOut{A $k$-$b$-coloring of $G$ or NO if it does not exist.}

\lIf{$k < 96p_G + 18$}{use the tree-width based algorithm, see Proposition~\ref{prop:tree-width}}

compute a fen-core $S$ of $G$ using Lemma~\ref{lem:S-exists}\;\label{ln:core}
\For{\Each $S$-profile $\zeta := (\chi, B)$}
{
\If{$\zeta$ is candidate-failing or pivot-failing, see Definition~\ref{def:failing}}{\Continue with next $S$-profile}

\For{\Each color plan $\pi$ of $\zeta$}{
\If{$\pi$ is not valid, see Definition~\ref{def:color-plan}}{\Continue with next color plan}
compute a feasible color realization $\rho$ using $\pi$ and Lemma~\ref{lem:get-feasible-realization}\;\label{ln:feasible}
run Algorithm~\ref{alg:partial-b-coloring} to compute a partial coloring $\psi_0$ of $G$\;
extend $\psi_0$ to a $k$-$b$-coloring $\psi$ of $G$ using Lemma~\ref{lem:finish-coloring}\;
\Return $\psi$\;
}
}
\Return NO\;
\caption{Finding a $b$-coloring}\label{alg:high-level}
\end{algorithm}

Now we put everything together and prove that Algorithm~\ref{alg:high-level} solves \bcoloring in \FPT time parameterized by the feedback edge number of the input graph.

\begin{proof}[of Theorem~\ref{thm:fen}]
Let $G$ be a graph with feedback edge number $p_G$, and let $k$ be an integer.
First, we compute a smallest feedback edge set $F$ using the DFS algorithm.
If $k < 96p_G + 18$, then using $F$, we compute a tree-decomposition for $G$ of width $p_G+1$, and we use the algorithm for \bcoloring parameterized by tree-width plus $k$, see Proposition~\ref{prop:tree-width}.
From now on, assume that $k \ge 96p_G + 18$.

We shall prove that if there is a $k$-$b$-coloring of $G$, then Algorithm~\ref{alg:high-level} outputs one such coloring, and otherwise it outputs NO.
Let $S$ be the fen-core of $G$ computed on line~\ref{ln:core}; it exists by Lemma~\ref{lem:S-exists}.
Let $|S| = p$.

First, suppose that there is a $k$-$b$-coloring $\psi_0$ of $G$.
For each color $c \in [k]$ that has a $b$-vertex in $S$ in $\psi_0$, let us fix one such vertex $u_c$, let $B = \{u_c \sep c$ has a $b$-vertex in $S\}$, and let $|B| = b$.
Let $\psi_1$ be a coloring of $G$ obtained from $\psi_0$ by permuting the colors so that $\psi_1(S) \seq [p]$ and $\psi_1(B) = [b]$ and let $\chi = \psi_1 \upharpoonright S$.
Clearly, $\psi_1$ is a $k$-$b$-coloring of $G$ described by $\zeta := (\chi, B)$, see Definition~\ref{def:S-profile}.
By Lemma~\ref{lem:failing-imlies-no-b--coloring}, $\zeta$ is not failing.
In particular, it is not plan-failing, which means that there is a valid color plan of $\zeta$ and line~\ref{ln:feasible} is reached in the execution of Algorithm~\ref{alg:high-level}.
Since $\zeta$ is not failing, a feasible color realization $\rho$ is computed on line~\ref{ln:feasible} by Lemma~\ref{lem:get-feasible-realization}.
By Lemmas~\ref{lem:algo-1-works} and~\ref{lem:finish-coloring}, the returned coloring $\psi$ is indeed a $k$-$b$-coloring of $G$.

Second, suppose that a coloring $\psi$ is computed by Algorithm~\ref{alg:high-level}.
Let $\zeta$ and $\pi$ be the values of the corresponding variables when Algorithm~\ref{alg:high-level} terminates.
Clearly, $\zeta$ is not candidate-failing nor pivot-failing and $\pi$ is a valid color plan of $\zeta$, i.e., $\zeta$ is not failing.
Hence, by Lemmas~\ref{lem:get-feasible-realization}, \ref{lem:algo-1-works}, and~\ref{lem:finish-coloring}, $\psi$ is indeed a $k$-$b$-coloring of $G$.
We have proven the correctness of Algorithm~\ref{alg:high-level}.

\smallskip
Let us now analyze the running time of Algorithm~\ref{alg:high-level}.
If $k < 96p_G + 18$, then the running time is $2^{\ca O(p_G^2)}$ by Proposition~\ref{prop:tree-width} (which will actually dominate the running time of the algorithm).
Observe that the number of $S$-profiles is at most $p^p \cdot 2^p \in p^{\ca O(p)}$ and the number of color plans is at most $(p+1)^{p-b} \in p^{\ca O(p)}$.
Moreover, it can be routinely verified that Algorithm~\ref{alg:partial-b-coloring} runs in polynomial time.
One can also in polynomial time determine whether an $S$-profile is candidate-failing or pivot-failing.
By Lemmas~\ref{lem:S-exists}, \ref{lem:get-feasible-realization}, and~\ref{lem:finish-coloring}, all other parts of Algorithm~\ref{alg:high-level} can be executed in polynomial time.
By Definition~\ref{def:fen-core}, $p \in \ca O(p_G)$, which means that Algorithm~\ref{alg:high-level} achieves the desired running time.
\end{proof}

\section{Conclusion}\label{sec:conclusion}
Even though the feedback edge number is a restrictive parameter, generalizing the polynomial-time algorithm for the \textsc{$b$-Chromatic Number} problem on trees into an \FPT algorithm for \bcoloring under this parameterization turned out highly non-trivial.
The natural open question is whether our algorithm can be further generalized into an \FPT algorithm parameterized by the feedback \emph{vertex} number or whether this generalization is \textsf{W[1]}-hard, see Figure~\ref{fig:hierarchy}.
Recall that our algorithm is based on certain locality: the complexity of the graph is contained in the fen-core $S$, and the rest of the graph is tree-like.
Such a locality does not hold when parameterized by the feedback vertex number because the feedback vertices (whose removal yields a forest) may have neighbors everywhere in the graph. 
For this reason, we would not be surprised if this problem was \textsf{W[1]}-hard.

Another natural question arising from our results is whether \bcoloring is \FPT when parameterized by the distance to cluster (instead of the distance to co-cluster we have considered), see Figure~\ref{fig:hierarchy}. Recall that bounded distance to cluster means that there is a small set whose removal yields a disjoint union of cliques.
Let us remark that a closely related and more restrictive parameter, namely the twin-cover number, was shown to be sufficient for an \FPT algorithm in \cite{JaffkeLS23}.
This parameter provides an additional condition, which guarantees that the vertices contained in the same clique have the same neighborhood in the cluster-modulator.
Is this additional condition necessary or is \bcoloring parameterized by the distance to cluster in \FPT?

Finally, another open problem asked by \cite{JaffkeLS23} is the complexity of \bcoloring when parameterized by the modular width of the graph, which is another parameter generalizing the twin-cover number.

\bibliographystyle{abbrvnat}
% use the following instead if you encounter problems
%\bibliographystyle{alpha}
\bibliography{biblio}
\end{document}